\documentclass[nofonts, a4paper,10pt]{article}
\usepackage[dvipdfmx]{graphicx}
\usepackage{amsfonts}
\usepackage{amsmath}
\usepackage{amsthm}
\usepackage{amssymb}
\usepackage{fancyhdr}
\usepackage{hyperref}
\usepackage{mathtools}
\usepackage{indentfirst}
\usepackage{placeins}
\usepackage{listings}
\usepackage[labelsep=period]{caption}
\usepackage{subfigure}
\usepackage{xcolor}
\usepackage{comment}
\usepackage{url}
\usepackage{bm}
\usepackage{multirow}
\usepackage[numbers]{natbib}

\def\longrightharpoonup{\relbar\joinrel\rightharpoonup}
\def\longleftharpoondown{\leftharpoondown\joinrel\relbar}

\def\longrightleftharpoons{
  \mathop{
    \vcenter{
       \hbox{
       \ooalign{
          \raise1pt\hbox{$\longrightharpoonup\joinrel$}\crcr
  	  \lower1pt\hbox{$\longleftharpoondown\joinrel$}
	}
      }
    }
  }
}

\newtheorem{theorem}{Theorem}
\newtheorem{corollary}{Corollary}
\newtheorem{remark}{Remark}
\newtheorem{lemma}{Lemma}

\makeatletter
\let\@fnsymbol\@arabic
\makeatother

\allowdisplaybreaks

\begin{document}
\title{Jarzynski's equality, fluctuation theorems, and variance reduction:
Mathematical analysis and numerical algorithms}
\date{}
%\author{Wei Zhang, Carsten Hartmann, Christof Sch\"{u}tte}
%\author{
%  Carsten Hartmann\thanks{Brandenburgische Technische Universit\"at
%Cottbus-Senftenberg, D-03046 Cottbus, Germany} \and  Christof Sch\"utte
%\thanks{Freie Universit\"at Berlin, D-14195 Berlin, Germany}\,
%\thanks{Zuse Institute Berlin, D-14195 Berlin, Germany}\,
%\and
%Wei Zhang \footnotemark[3]
%\and 
%}
\author{
  Carsten Hartmann \textsuperscript{1} \and  Christof Sch\"utte
  \textsuperscript{2,3} \and Wei Zhang \textsuperscript{3}
\and 
}

\footnotetext[1]{Institut f\"ur Mathematik, Brandenburgische Technische Universit\"at
Cottbus-Senftenberg, D-03046 Cottbus, Germany; carsten.hartmann@b-tu.de}
\footnotetext[2]{Institut f\"ur Mathematik, Freie Universit\"at Berlin, D-14195 Berlin, Germany; christof.schuette@fu-berlin.de}
\footnotetext[3]{Zuse Institute Berlin, D-14195 Berlin, Germany; wei.zhang@fu-berlin.de}

\maketitle
\begin{abstract}
  In this paper, we study Jarzynski's equality and fluctuation theorems
  for diffusion processes. While some of the results considered in the current
  work are known in the (mainly physics) literature, we review and generalize these nonequilibrium
  theorems using mathematical arguments,
  therefore enabling further investigations in the mathematical community.
  On the numerical side, variance reduction approaches such as importance sampling
  method are studied in order to compute free energy differences based on Jarzynski's equality.
\end{abstract}
\tableofcontents

\section{Introduction}
\label{sec-intro}
Nonequilibrium work relations concern the behavior of
dynamical systems which are out of equilibrium under nonequilibrium
driving forces. Different from linear response theory~\cite{kubo1966,MARCONI2008} where systems are required to be close to
equilibrium, nonequilibrium work relations refer to a set of equalities which hold for general systems far away from equilibrium. 
And the most remarkable ones include Jarzynski's
equality~\cite{jarzynski1997-master,jarzynski1997} and Crooks's fluctuation
theorem~\cite{crooks1999}. 
In particular, Jarzynski's equality relates free energy differences
to the work that is applied to the system in order to drive the system from one state
to another within a finite period of time. Since its first report in
$1997$~\cite{jarzynski1997-master,jarzynski1997}, considerable amount of
research work has been done both numerically and experimentally to study the
computation of free energy differences,
by driving the system out of equilibrium using nonequilibrium forces~\cite{bias-error-2003,optimum-bias-2008,optimal-estimator-minh-2009,path_sampling_zuckerman2004,compare_free_energy_methods_2006}.
In recent years, inspired by the work~\cite{generalized-jarzynski-feedback-2010}, 
there has also been growing research interest to generalize both Jarzynski's equality and 
fluctuation theorems to nonequilibrium systems under discrete feedback controls~\cite{fluct-thm-2012,generalized-fluct-thm-feedback,detailed-fluct-thm-repeated-discrete-feedback-2010,nonequ-feedback-control-sagawa-2012}. 

Although Jarzynski's equality ensures that free energy differences can be
calculated by pulling the system using any control forces (protocols)
and the transition can be done within any finite time, the efficiency 
of Monte Carlo estimators for free energy computation based on Jarzynski's
equality crucially depends on the control protocols and therefore careful design is
needed. Various techniques, such as importance sampling in trajectory
space~\cite{path_sampling_zuckerman2004,optimum-bias-2008}, the use of both
forward and reversed trajectories~\cite{crooks-path-ensemble-pre2000,compare_free_energy_methods_2006,optimal-estimator-minh-2009,escorted-simulation2011}, the interacting particle system techniques~\cite{Rousset-stoltz-ips-2006}, and the escorted free energy
simulation method~\cite{pre-escort-fe-simulation2008,escorted-simulation2011},
have been proposed in order to improve the efficiency of Monte Carlo estimators.
Meanwhile, we note that several recent works have considered optimal control protocols which
minimize either average work or average heat~\cite{optimal-protocol2008,
optimal-finite-time-seifert-2007,extracting-work-feedback-2011,optimal-protocols-transport-2011}.
However, it is important to point out that, although these protocols
are optimal in certain sense and are physically interesting, they do not
necessarily provide the optimal Monte Carlo estimators in the sense of smallest variance. Readers are referred to ~\cite{bias-error-2003,biased-sampling-dellago-2005,Jarzynskia2008,dellago-hummer-2014,
compare_free_energy_methods_2006} for detailed discussions on related issues.

In the aforementioned literature, the concept of free energy is often defined as a function of
physical parameters, e.g., temperature, volume or pressure, which characterize
the macroscopic status of physical system.
This is termed as the alchemical transition case
in~\cite{tony-free-energy-compuation}. Free energy also
plays an important role in the study
of model reduction of complex (molecular) systems along a given reaction coordinate or
collective variables. In this context, free energy is often defined as a function
of reaction coordinate which in turn depends on the state of the system.
And calculating free energy differences along a given reaction coordinate has
attracted considerable attentions in the study of molecular
systems~\cite{peptide_cmd,enhanced_sampling2014,eric_recent_techniques,basic_ingredients_free_energy,tony-free-energy-compuation}.
Similar to the alchemical transition case, Jarzynski-like equalities and their
applications in free energy calculation have
been considered in~\cite{LELIEVRE2007, Tony-constrained-langevin2012}.

Motivated by the development of nonequilibrium work relations and their
potential applications, the goal of the current work is to understand these results from a
mathematical point of view, and to study variance reduction approaches, such
as importance sampling, in Monte Carlo methods for free energy calculation based on Jarzynski's equality.
In the alchemical transition case, we provide mathematical proofs of both Jarzynski's equality and
fluctuation theorems in a general setting based on the theory of stochastic
differential equations, making them more accessible for readers in
mathematical community (we refer to the previous study~\cite{Ge2008} for a 
mathematical proof of Jarzynski's equality). It is worth emphasizing that 
the nonequilibrium diffusion processes in our setting are allowed to be
irreversible and can have multiplicative noise. Furthermore, the Jarzynski's
equality is generalized to allow noisy control protocols. 
This generalization may be useful to study systems in
experiments~\cite{Hummer-szabo-2001}, since the implementations of control protocols through physical devices are typically imprecise to some extent.
As an advantage of our mathematical approach, it allows us to elucidate the
connection between thermodynamic integration identity and Jarzynski's
equality, which were usually considered as two distinct identities involving free energy differences. 
Such a connection is indeed known in physics community~\cite{Crooks1998}, but we believe it is helpful to present its mathematical derivation.
In the reaction coordinate case, we prove a fluctuation theorem and
derive a Jarzynski-like equality based on the fluctuation theorem. These results complement the
previous mathematical studies in~\cite{LELIEVRE2007,Tony-constrained-langevin2012}.
In both the alchemical transition case and the reaction coordinate case, following our previous studies~\cite{ce_paper2014,Hartmann2017-ptrf,Hartmann2016-Nonlinearity}, we
investigate variance reduction approaches in order to compute free energy differences 
using Monte Carlo method based on Jarzynski's equality.  

The paper is organized as follows. In Section~\ref{sec-alchemical}, we
study the Jarzynski's equality and fluctuation theorem in the alchemical
transition case. In particular, the cases when the control protocols are noisy will be considered. Information-theoretic formulation of 
Jarzynski's equality, the importance sampling method, as well as the
cross-entropy method will be discussed in the context of free energy calculation.
In Section~\ref{sec-coordinate}, we study the Jarzynski-like equality and the fluctuation
theorem in the reaction coordinate case. 
Information-theoretic formulations and variance reduction approaches will be discussed following a similar reasoning as in
Section~\ref{sec-alchemical}.
Two simple numerical examples are studied in detail in Section~\ref{sec-examples} to illustrate the
numerical issues of Monte Carlo estimators for free energy calculation as well as the variance reduction ideas proposed in this work.  
In Appendix~\ref{app-1} two asymptotic regimes of nonequilibrium
processes(fast mixing and slow driving) and, in particular, connections between Jarzynski's equality and thermodynamic integration
identity will be discussed.
Appendix~\ref{app-2} records the thermodynamic integration identity in the
reaction coordinate case. 
Appendix~\ref{app-3} contains an alternative proof of the fluctuation theorem
(Theorem~\ref{thm-fluct-relation}) in the alchemical transition case.
The proof of the fluctuation theorem in the reaction coordinate case
(Theorem~\ref{thm-fluct-relation-coordinate}) is given in Appendix~\ref{app-4}.

\section{Jarzynski's equality and fluctuation theorem: alchemical transition case}
\label{sec-alchemical}
 In this section, we study the Jarzynski's equality and the fluctuation theorem in the
 alchemical transition case. In Subsection~\ref{sub-sec-setup}, we introduce
 the dynamical systems which will be studied in this section and fix notations.
Jarzynski's equality and fluctuation theorem will be studied
from Subsection~\ref{sub-sec-jarzynski} to Subsection~\ref{subsec-fluct-thm}.
Finally, Information-theoretic formulation of Jarzynski's equality, as well as the
cross-entropy method will be discussed in Subsection~\ref{subsec-is} and
Subsection~\ref{subsec-ce}, respectively. 

\subsection{Mathematical setup}
\label{sub-sec-setup}
Consider the stochastic process $x(s) \in \mathbb{R}^n$ which satisfies the stochastic
differential equation (SDE)
\begin{align}
  \begin{split}
    d x(s) & =  b(x(s), \lambda(s))\, ds + \sqrt{2\beta^{-1}} \sigma(x(s),
    \lambda(s)) \,dw^{(1)}(s)\,,\quad s \ge 0\,,
\end{split}
  \label{dynamics-1}
\end{align}
where $\beta>0$ is a constant, $w^{(1)}(s)$ is a $d_1$-dimensional Brownian
motion with $d_1 \ge n$. Both the drift vector $b : \mathbb{R}^n \times \mathbb{R}^m \rightarrow \mathbb{R}^n$ and the matrix $\sigma : \mathbb{R}^n \times \mathbb{R}^m \rightarrow \mathbb{R}^{n \times d_1}$ are
  smooth functions depending on the \textit{control protocol} $\lambda(s) \in
  \mathbb{R}^m$, which we assume is governed by 
\begin{align}
  d\lambda(s) = f(\lambda(s), s)\,ds + \sqrt{2\epsilon}\, \alpha(\lambda(s),
  s)\,dw^{(2)}(s)\,.  \label{lambda-dynamics}
\end{align}
In the above, $\epsilon \ge 0$ is related to the intensity of the noise,
$\lambda(0)\in\mathbb{R}^m$ is fixed, $f : \mathbb{R}^m \times \mathbb{R}^+ \rightarrow
\mathbb{R}^m$, $\alpha : \mathbb{R}^m \times \mathbb{R}^+ \rightarrow \mathbb{R}^{m \times d_2}$
are smooth functions, and $w^{(2)}(s)$ is a $d_2$-dimensional Brownian motion
independent of $w^{(1)}(s)$. Notice that
in equation (\ref{lambda-dynamics}), functions $f, \alpha$ are assumed to be independent of $x(s)$,
and therefore the control protocol $\lambda(s)$ is of feedback form with
respect to itself but does not depend on the system state $x(s)$. More generally, 
in Subsection~\ref{subsec-fluct-thm}, we will also consider the case 
when the control protocol is of feedback form with respect to both processes $x(s)$
and $\lambda(s)$, i.e.,
\begin{align}
  d\lambda(s) = f(x(s), \lambda(s), s)\, ds + \sqrt{2\epsilon}\, \alpha(x(s),
  \lambda(s), s)\, dw^{(2)}(s)\,.
  \label{lambda-dynamics-full}
\end{align}
In both cases (\ref{lambda-dynamics}) and (\ref{lambda-dynamics-full}), the
infinitesimal generator of the dynamics $\lambda(s)$ for fixed $x(s)$ is given by 
\begin{align}
  \mathcal{L}_2 = f \cdot \nabla_\lambda + \epsilon\,(\alpha\alpha^T):\nabla^2_\lambda\,,
\label{l-lambda}
\end{align}
where $\nabla_\lambda$ denotes the gradient operator with respect to the
variable $\lambda \in \mathbb{R}^m$ and 
\begin{align*}
  (\alpha\alpha^T):\nabla^2_\lambda \phi \coloneqq \sum_{1 \le i,j \le m}
  (\alpha\alpha^T)_{ij} \frac{\partial^2\phi}{\partial\lambda_i\partial\lambda_j} \,,
\end{align*}
for a smooth function $\phi$ of variable $\lambda \in \mathbb{R}^m$.

For fixed parameter $\lambda\in \mathbb{R}^m$, the dynamics (\ref{dynamics-1}) reads
\begin{align}
    d x(s) & =  b(x(s), \lambda)\, ds + \sqrt{2\beta^{-1}} \sigma(x(s),
  \lambda)\, dw^{(1)}(s)\,,\quad s \ge 0\,,
  \label{dynamics-1-fixed-lambda}
\end{align}
and its infinitesimal generator is 
\begin{align}
  \mathcal{L}_{1} = b(\cdot, \lambda) \cdot \nabla +
  \frac{1}{\beta} a(\cdot, \lambda) : \nabla^2\,, 
  \label{l-1}
\end{align}
where the matrix $a=\sigma\sigma^T$ and $\nabla$ denotes the gradient operator with
respect to $x \in \mathbb{R}^n$. Correspondingly, the infinitesimal generator of
the joint process $(x(s), \lambda(s))$ is 
\begin{align}
  \mathcal{L} = \mathcal{L}_1 + \mathcal{L}_2\,,
  \label{l-x-lambda}
\end{align}
since the two Brownian motions $w^{(1)}(s)$, $w^{(2)}(s)$ are independent. 
Throughout this article, we assume that the drift and noise coefficients
satisfy appropriate Lipschitz and growth conditions, such that equations
(\ref{dynamics-1})-(\ref{lambda-dynamics-full}) have unique strong solutions~\cite{oksendalSDE}.
For each fixed parameter $\lambda\in \mathbb{R}^m$, we further assume that the process
$x(s)$ in (\ref{dynamics-1-fixed-lambda}) is ergodic and has a unique
invariant measure $\mu_\lambda$ satisfying 
\begin{align}
  \mu_\lambda(dx) = \rho(x,\lambda) dx\,, \quad \int_{\mathbb{R}^n}
  \rho(x,\lambda) dx = 1\,.
  \label{invariant-mu}
\end{align}
Furthermore, we introduce the potential 
\begin{align}
  V(x,\lambda) = -\beta^{-1} \ln \rho(x,\lambda) + \mbox{constant}\,,
  \label{potential}
\end{align}
where the constant only depends on the parameter $\lambda$. Equivalently,  
we have $\rho(x,\lambda) = \frac{1}{Z(\lambda)} e^{-\beta V(x,\lambda)}$,
and the normalization constant $Z(\lambda)$ is given by 
\begin{align}
  Z(\lambda) = \int_{\mathbb{R}^n} e^{-\beta V(x,\lambda)} dx \,.
  \label{normal-const}
\end{align}
The free energy of the system (\ref{dynamics-1-fixed-lambda}) for a fixed parameter $\lambda \in \mathbb{R}^m$ is defined as 
\begin{align}
F(\lambda) = -\beta^{-1}\ln Z(\lambda)\,.
\label{free-energy}
\end{align}

To proceed, we follow the previous study~\cite{effective_dyn_2017} and introduce the quantity
\begin{align}
  J_i(x,\lambda) = b_i - \frac{1}{\beta \rho}
  \sum_{j=1}^n \frac{\partial (a_{ij} \rho)}{\partial x_j}
\,, \quad 1 \le i \le n\,.
  \label{flux-def}
\end{align}
Note that both here and in the following, $J_i$, $b_i$ denote the $i$th component
of the vectors $J,\,b$, respectively. 
Also, the dependence of the functions on the variables $x$ and $\lambda$ will be omitted when no ambiguities arise. 
Since the probability measure $\mu_\lambda$ in (\ref{invariant-mu}) is the
invariant measure of the dynamics (\ref{dynamics-1-fixed-lambda}), 
we can verify that 
\begin{align}
  \mbox{div} \Big(J(x, \lambda) e^{-\beta V(x, \lambda)}\Big) \equiv 0\,, \quad \rho-\mbox{a.e.}\hspace{0.2cm} x \in \mathbb{R}^n\,,
  \label{div-j-zero}
\end{align}
for every $\lambda \in \mathbb{R}^m$. Thus, (\ref{dynamics-1}) can be written as 
\begin{align}
  d x_i(s) & =  J_i ds 
  + \frac{1}{\beta \rho} 
\sum_{j=1}^n \frac{\partial(a_{ij} \rho)}{\partial x_j} ds
+ \sqrt{2\beta^{-1}} \sum_{j=1}^{d_1}\sigma_{ij}\,dw^{(1)}_j(s)\,, \quad 1 \le i \le n\,,
\label{dynamics-1-q}
\end{align}
or, in vector form,
\begin{align}
  d x(s) & =  \Big(J 
  - a \nabla V + \frac{1}{\beta} \nabla \cdot a\Big)\,ds + \sqrt{2\beta^{-1}}
  \sigma\,dw^{(1)}(s)\,,
\label{dynamics-1-q-vector}
\end{align}
where $\nabla \cdot a$ denotes the vector in $\mathbb{R}^n$ with components 
\begin{align}
  (\nabla \cdot a)_i = \sum_{j=1}^n \frac{\partial a_{ij}}{\partial x_j}\,, \quad 1 \le
  i \le n\,.
  \label{nabla-dot}
\end{align}

Finally, we introduce two physical quantities which are associated to the
trajectories of the stochastic processes $x(s),\lambda(s)$ and will be
relevant for our subsequent study. 
  For each trajectory $x(s)$, $\lambda(s)$ of the dynamics (\ref{dynamics-1}), (\ref{lambda-dynamics-full}) on the time
  interval $[t_1,t_2] \subseteq [0,T]$, 
the \textit{change of internal energy} and the
  \textit{work} done to the system are defined as
\begin{align}
  \begin{split}
    \Delta \mathcal{U}_{(t_1,\,t_2)} =& V\big(x(t_2), \lambda(t_2)\big) - V\big(x(t_1), \lambda(t_1)\big)\,\\
		 W_{(t_1,\,t_2)} =& \int_{t_1}^{t_2} \nabla_{\lambda} V(x(s), \lambda(s)) \circ d\lambda(s) \,,
\end{split}
\label{u-q-w}
\end{align}
\begin{comment}
%    Q^{\mbox{tot}}_{(t_1,\,t_2)} =& \int_{t_1}^{t_2} \Big[a^{-1}\Big(J - a\nabla V + \frac{1}{\beta} \sigma \nabla \cdot \sigma^T\Big)\Big](x(s), \lambda(s)) \circ dx(s)\\
Note that, in (\ref{u-q-w}),
$\nabla\cdot \sigma^T \in \mathbb{R}^{d_1}$ is the vector with components
$(\nabla\cdot \sigma^T)_{i} = \sum\limits_{l=1}^n \frac{\partial
\sigma_{li}}{\partial x_l}$, $1 \le i \le d_1$, 
The total heat
$Q^{\mbox{tot}}_{(t_1,\,t_2)}$ was introduced
in~\cite{sekimoto-heat} and further investigations
can be found in~\cite{sst-langevin-hatano,ift-seifert-heat}. The sign is
chosen such that the total heat is positive when there is a heat flow from the
system into the environment (heat bath).
integration by parts, equation (\ref{lambda-dynamics}) (or (\ref{lambda-dynamics-full})),
    Q^{\mbox{tot}}_{(t_1,\,t_2)}= & - \Delta \mathcal{U}_{(t_1,t_2)} + W_{(t_1,t_2)} +
    \int_{t_1}^{t_2} \Big(a^{-1}J + \frac{1}{\beta} a^{-1} \sigma \nabla
    \cdot \sigma^{T}\Big)\big(x(s), \lambda(s)\big) \circ dx(s)\,,\\
In the special case when $a=\sigma = \mbox{id}$, $J\equiv
0$ (detailed balance), and $\epsilon=0$, the quantities in (\ref{u-q-w})
become more familiar and in particular
the identity $Q^{\mbox{tot}}_{(t_1,t_2)} = -\Delta \mathcal{U}_{(t_1,t_2)} + W_{(t_1,t_2)}$
(the first law of thermodynamics) is recovered from (\ref{q-w-1}). 
\end{comment}
respectively. Note that, in (\ref{u-q-w}), the notation `$\circ$'
indicates that Stratonovich integration has been used. 
Using the relation between Stratonovich integration and Ito
integration, we can verify the alternative expression
\begin{align}
  \begin{split}
    W_{(t_1,\,t_2)} = & \int_{t_1}^{t_2} \Big(\nabla_{\lambda} V \cdot f + \epsilon\,
    \alpha\alpha^T:\nabla^2_\lambda V\Big)\big(x(s), \lambda(s),s\big)\, ds \\
    & + \sqrt{2\epsilon} \int_{t_1}^{t_2} \big(\alpha^T \nabla_{\lambda}
    V\big) \big(x(s), \lambda(s),s\big)\cdot dw^{(2)}(s)\,,
  \end{split}
  \label{q-w-1}
\end{align}
where Ito integration has been used. 

In the following, we will omit the subscripts and adopt the notation
$W=W_{(t_1,t_2)}$ when we consider the time interval $[t_1,t_2] = [0,T]$. Similarly, $W(t)$ will be used to denote
the work $W_{(0,t)}$ for $t \in [0,T]$.
\subsection{Jarzynski's equality under noisy control protocol}
\label{sub-sec-jarzynski}
Jarzynski's equality can be derived using different approaches~\cite{Jarzynskia2008}.
In this subsection, we will provide a simple argument to obtain the
(generalized) Jarzynski's equality, where the nonequilibrium processes $x(s)$ can be
irreversible for fixed parameter $\lambda$, the diffusion coefficient $\sigma$
in the equation (\ref{dynamics-1}) of $x(s)$ can be
position dependent (multiplicative noise), and the control protocol
$\lambda(s)$ can be stochastic ($\epsilon > 0$).
The proof has some similarities with the one in~\cite{Hummer-szabo-2001} using
the Feynman-Kac formula. As an advantage of our method, it allows us to figure out
the connections between thermodynamic integration and Jarzynski's
equality by analyzing the related PDEs. See Remark~\ref{rmk-1} and
Appendix~\ref{app-1} for more details.

Before starting, we first introduce the quantity 
\begin{align}
  \begin{split}
  g(x, \lambda, t) 
    =& \mathbf{E}_{x,\lambda,t} \Big(\varphi(x(T), \lambda(T))\,e^{-\beta
    W_{(t,T)}} \Big)\\
    =& \mathbf{E}_{x,\lambda,t} \Big[\varphi(x(T), \lambda(T))\,e^{-\beta \int_t^T \nabla_{\lambda} V\big(x(u),
  \lambda(u)\big) \circ\, d\lambda(u)}\Big]\,,
\end{split}
\label{g-def}
\end{align}
for fixed $x \in \mathbb{R}^n$, $\lambda \in \mathbb{R}^m$ and $0 \le t \le T$,
where $\varphi : \mathbb{R}^n \times \mathbb{R}^m \rightarrow \mathbb{R}$ is a
bounded and continuous test
function, $\mathbf{E}_{x,\lambda,t}$ denotes the conditional expectation with
respect to the path ensemble of the dynamics
(\ref{dynamics-1}), (\ref{lambda-dynamics-full}) starting from $x(t) = x$ and 
$\lambda(t) = \lambda$ at time $t$.
The following lemma is a direct application of the Feynman-Kac
formula~\cite{oksendalSDE}, and we provide its proof for completeness. 
\begin{lemma}
  Consider the dynamics $x(s), \lambda(s)$ given in (\ref{dynamics-1}), (\ref{lambda-dynamics-full}).
  The function $g$ defined in (\ref{g-def}) satisfies the equation
\begin{align}
  \begin{split}
    &\partial_t g + \mathcal{L}_{1} g + \mathcal{L}_2 g
    -2\epsilon\beta \big(\alpha^T\nabla_\lambda V\big) \cdot \big(\alpha^T \nabla_\lambda
    g\big)
   + \Big(\epsilon \beta^2 |\alpha^T\nabla_{\lambda} V|^2 - \beta \mathcal{L}_2V \Big) g =
  0\,, \quad 0 \le t < T\,,\\
    &g(\cdot,\cdot, T) = \varphi \,,
\end{split}
\label{g-pde}
\end{align}
where $\mathcal{L}_1$ is the operator defined in (\ref{l-1}), which 
is the infinitesimal generator of the dynamics (\ref{dynamics-1}) for 
$x(s)$ when $\lambda \in \mathbb{R}^m$ is fixed, and 
$\mathcal{L}_2$ is the operator defined in (\ref{l-lambda}) for the process $\lambda(s)$ when $x \in \mathbb{R}^n$ is fixed.
\label{lemma-g}
\end{lemma}
\begin{proof}
  Using the tower property of the conditional expectation, we have 
  \begin{align}
    \begin{split}
     g(x,\lambda, t) 
      =\, &\mathbf{E}_{x,\lambda,t}\Big[\varphi(x(T), \lambda(T))\,e^{-\beta \int_t^T \nabla_{\lambda} V\big(x(u),
  \lambda(u)\big) \circ\, d\lambda(u)}\Big] \\
  =\, &\mathbf{E}_{x,\lambda,t}\Big[e^{-\beta \int_t^{s} \nabla_{\lambda} V\big(x(u),
  \lambda(u)\big) \circ\, d\lambda(u)} g(x(s), \lambda(s), s) \Big]\,, 
\end{split}
\label{lemma-g-1}
  \end{align}
  for all time $s \in [t,T]$.
  Let us define $Y(s) = e^{-\beta \int_t^{s} \nabla_{\lambda} V\big(x(u),
  \lambda(u)\big) \circ\, d\lambda(u)}$. Changing Stratonovich
  integration into Ito integration as in (\ref{q-w-1}) and 
  applying Ito's formula to the process $Y(s)$, we get
  \begin{align*}
    dY(s) =  Y(s) \Big[-\beta\mathcal{L}_2V\, ds +  \epsilon
    \beta^2\big|\alpha^T \nabla_\lambda V\big|^2\, ds -
    \sqrt{2\epsilon}\beta \big(\alpha^T \nabla_\lambda V\big) \cdot \,dw^{(2)}(s)\Big]\,.
  \end{align*}
  In a similar way, applying Ito's formula to $g(x(s), \lambda(s), s)$, gives
  \begin{align*}
    dg = \big(\partial_t g + \mathcal{L}_1 g + \mathcal{L}_2 g\big) ds +
    \sqrt{2\beta^{-1}}
    \big(\sigma^T\nabla g\big) \cdot dw^{(1)}(s) + \sqrt{2\epsilon}
    \big(\alpha^T\nabla_\lambda g\big)\cdot \,dw^{(2)}(s)\,.
  \end{align*}
  Note that, here and in the following, we drop the dependence of the
  functions on the states $x(s)$, $\lambda(s)$ and the time $s$ in order to simplify notation.
  Applying Ito's formula to the product $Y(s) g(x(s), \lambda(s), s)$, we obtain
  \begin{align}
    \begin{split}
    & e^{-\beta \int_t^{s} \nabla_{\lambda} V\big(x(u),
  \lambda(u)\big) \circ\, d\lambda(u)} g(x(s), \lambda(s), s) \\
  =\,& g(x,\lambda, t) + \int_t^{s} Y(u) 
  \Big(-\beta\mathcal{L}_2 V +  \epsilon \beta^2|\alpha^T \nabla_\lambda V|^2 \Big)
  g(x(u), \lambda(u), u)\, du \\
    & + \int_t^{s} Y(u) \big(\partial_t g + \mathcal{L}_1 g + \mathcal{L}_2 g\big)\,du  - 
  2\epsilon \beta \int_t^{s} Y(u) \big(\alpha^T\nabla_\lambda V\big)\cdot \big(\alpha^T\nabla_\lambda g\big)\,du + M(s)\,,
\end{split}
\label{lemma-yg}
  \end{align}
  where $M(s)$ is a (local) martingale. 
  Taking expectations in (\ref{lemma-yg}) and using (\ref{lemma-g-1}), we get
  \begin{align*}
    \mathbf{E}_{x,\lambda,t}\bigg[
    &-\beta\int_t^{s} Y(u) 
      (\mathcal{L}_2 V) g\, du + \epsilon\beta^2 \int_t^s Y(u) |\alpha^T\nabla_\lambda V|^2 g\, du \\
    & + \int_t^{s} Y(u) \big(\partial_t g + \mathcal{L}_1 g + \mathcal{L}_2
      g\big)\,du  - 2\epsilon \beta \int_t^{s} Y(u)
      \big(\alpha^T\nabla_\lambda V\big)\cdot \big(\alpha^T\nabla_\lambda g\big)\,du 
    \bigg] = 0 \,.
  \end{align*}
  Notice that $Y(t)=1$, $x(t) = x$ and $\lambda(t) = \lambda$ at time $t$.
  Dividing the last equation by $(s-t)$ and letting $s \rightarrow t+$, we obtain 
  (\ref{g-pde}) which concludes the proof. 
\end{proof}
Now we can prove the Jarzynski equality as stated below.
\begin{theorem}[Generalized Jarzynski equality]
  Let 
  $x(s)$ and $\lambda(s)$ be given by 
  (\ref{dynamics-1}) and (\ref{lambda-dynamics}), respectively. 
  Then, for any bounded smooth test function $\varphi : \mathbb{R}^n \times \mathbb{R}^m \rightarrow \mathbb{R}$, 
  we have 
\begin{align}
  \mathbf{E}_{\lambda(0),0}\Big[\varphi(x(t), \lambda(t))\,e^{-\beta
  W(t)}\,\Big]  = \mathbf{E}_{\lambda(0),0}\Big[ e^{-\beta \big(F(\lambda(t)) -
  F(\lambda(0))\big)} \mathbf{E}_{\mu_{\lambda(t)}} \varphi(\cdot, \lambda(t))\Big] \,,
  \label{generalized-jarzynski-varphi}
\end{align}
  where $F(\cdot)$ is the free energy in (\ref{free-energy}) and
  $W(t)=W_{(0,t)}$ is the work defined in (\ref{u-q-w}) on the time interval
  $[0,t]$. $\mathbf{E}_{\mu_{\lambda(t)}}$ denotes the expectation with
  respect to the probability measure $\mu_{\lambda(t)}$ on $\mathbb{R}^n$. And
  $\mathbf{E}_{\lambda(0), 0}$ denotes the conditional expectation over the
  realizations of $x(s)$ and $\lambda(s)$, starting from fixed $\lambda(0)\in \mathbb{R}^m$ and the initial distribution
$x(0) \sim \mu_{\lambda(0)}$.
In particular, choosing $\varphi\equiv 1$, we have 
  \begin{align}
    \mathbf{E}_{\lambda(0), 0}\Big[e^{-\beta W(t)}\Big] = \mathbf{E}_{\lambda(0),0}
    \Big[e^{-\beta \big(F(\lambda(t)) - F(\lambda(0))\big)}\Big]\,.
  \label{generalized-jarzynski}
  \end{align}
\label{thm-1}
\end{theorem}
\begin{proof}
  It suffices to prove the equality (\ref{generalized-jarzynski-varphi}) for $t=T$. From the definitions of the function $g$ 
  in (\ref{g-def}) and the function $Z(\lambda)$ in (\ref{normal-const}), 
  it is easy to see that (\ref{generalized-jarzynski-varphi}) is equivalent to 
\begin{align}
\int_{\mathbb{R}^n} g(x, \lambda(0), 0) e^{-\beta V(x, \lambda(0))} dx = 
  \mathbf{E}_{\lambda(0),0}\bigg[\int_{\mathbb{R}^n} g(x,\lambda(T), T)\,
  e^{-\beta V(x, \lambda(T))} \, dx\bigg]\,.
\label{thm-1-eqn-1}
\end{align}
Noticing that the process $\lambda(s)$ in (\ref{lambda-dynamics}) is independent
of $x(s)$ and motivated by the form of (\ref{thm-1-eqn-1}), we consider the quantity 
$\int_{\mathbb{R}^n} e^{-\beta V(x,\lambda(s))} g(x,\lambda(s), s) dx$ as a
function of time $s$.
Applying Ito's formula, we compute 
\begin{align}
  \begin{split}
    & d\bigg[\int_{\mathbb{R}^n} e^{-\beta V(x,\lambda(s))} g(x,\lambda(s), s) dx \bigg] \\
  = & \bigg[\int_{\mathbb{R}^n} e^{-\beta V(x,\lambda(s))}
    \Big( \partial_t g\,+ \mathcal{L}_2 g + \big(
    \epsilon \beta^2 |\alpha^T\nabla_{\lambda} V|^2 - \beta \mathcal{L}_2V  \big) g
  - 2\epsilon\beta \big(\alpha^T\nabla_\lambda V\big) \cdot
\big(\alpha^T\nabla_\lambda g\big) \Big)\, dx\bigg] ds \\
&  + \sqrt{2\epsilon} \bigg[\int_{\mathbb{R}^n} e^{-\beta V(x,\lambda(s))}
 \alpha^T\big(\nabla_\lambda g - \beta \nabla_\lambda V\, g\big) dx\,\bigg]\cdot \,dw^{(2)}(s)\,,
 \end{split}
 \end{align}
where the functions under the integral above are evaluated at $(x,\lambda(s), s)$.
Since the function $g$ satisfies the equation (\ref{g-pde}) in Lemma~\ref{lemma-g}, we find
 \begin{align}
   \begin{split}
    & d\bigg[\int_{\mathbb{R}^n} e^{-\beta V(x,\lambda(s))} g(x,\lambda(s), s) dx \bigg] \\
 =& - \bigg[\int_{\mathbb{R}^n}e^{-\beta V(x,\lambda(s))} \mathcal{L}_1 g \,
   dx\bigg]\,ds
 + \sqrt{2\epsilon} \bigg[\int_{\mathbb{R}^n} e^{-\beta V(x,\lambda(s))}
 \alpha^T\big(\nabla_\lambda g - \beta \nabla_\lambda V\, g\big) dx\,\bigg]\cdot \,dw^{(2)}(s)\,.
 \end{split}
  \label{int-identity}
\end{align}
Recalling that $\mu_{\lambda}$ in (\ref{invariant-mu}) and $\mathcal{L}_1$ are the invariant measure
and the infinitesimal generator of dynamics (\ref{dynamics-1-fixed-lambda}),
we have $\mathcal{L}^*_{1} \big(e^{-\beta V(x,\lambda)}\big) = 0$, where
$\mathcal{L}^*_{1}$ is the formal $L^2$ adjoint of $\mathcal{L}_{1}$. Integrating by parts, we
conclude that the first term on the right hand side of equation (\ref{int-identity})
vanishes and therefore 
\begin{align*}
  d\bigg[\int_{\mathbb{R}^n} e^{-\beta V(x,\lambda(s))} g(x,\lambda(s), s) dx\bigg]
 = \sqrt{2\epsilon} \bigg[\int_{\mathbb{R}^n} e^{-\beta V(x,\lambda(s))}
 \alpha^T\big(\nabla_\lambda g - \beta \nabla_\lambda V\, g\big) dx\,\bigg]\cdot \,dw^{(2)}(s)\,.
\end{align*}
Taking expectation and noticing that $g(\cdot, \cdot, T) \equiv \varphi$,
we obtain (\ref{thm-1-eqn-1}) and the equality (\ref{generalized-jarzynski-varphi}) readily follows. 
\end{proof}
\begin{remark}
  \begin{enumerate}
    \item
  While Lemma~\ref{lemma-g} holds in both cases when the control protocol $\lambda(s)$
  satisfies either dynamics (\ref{lambda-dynamics}) or dynamics (\ref{lambda-dynamics-full}), 
  a close examination reveals that the proof of Theorem~\ref{thm-1} above is valid
  only when the process $\lambda(s)$ is independent of the process $x(s)$,
  i.e., when $\lambda(s)$ satisfies dynamics (\ref{lambda-dynamics}).
\item
 When $\epsilon = 0$, the control protocol is deterministic and the work becomes
\begin{align}
  W(t) = \int_0^{t} \nabla_{\lambda} V\big(x(s), \lambda(s)\big)\cdot \dot{\lambda}(s)\,ds 
    =  \int_0^{t} \nabla_{\lambda} V\big(x(s), \lambda(s)\big)\cdot f(\lambda(s),s)\, ds \,.
    \label{w-eps-0}
\end{align}
In this case, we recover the standard Jarzynski equality~\cite{jarzynski1997-master,jarzynski1997,Jarzynskia2008}, 
since (\ref{generalized-jarzynski}) becomes 
\begin{align}
  \mathbf{E}_{\lambda(0), 0}\Big[e^{-\beta W(t)}\Big] = e^{-\beta \Delta F(t)} \,,
  \label{jarzynski-1}
\end{align}
where 
      \begin{align}
	\Delta F(t) = F\big(\lambda(t)\big) - F\big(\lambda(0)\big)
	\label{delta-f}
      \end{align}
      is the free
energy difference and the conditional expectation is taken
with respect to dynamics (\ref{dynamics-1}), starting from the equilibrium
distribution $\mu_{\lambda(0)}$.
\item
  Besides the Jarzynski's equality, the thermodynamic integration identity is another
      well known representation of the free energy that can be used to calculate free energy differences~\cite{frenkel2001understanding}.
      Based on the argument in this subsection, in Appendix~\ref{app-1} we
      will derive the thermodynamic integration identity from Jarzynski's equality, and therefore 
      provide connections of these two methods. 
\end{enumerate}
\label{rmk-1}
\end{remark}
In~\cite{pre-escort-fe-simulation2008}, the authors proposed the escorted free
energy calculation method based on an identity for dynamics involving an extra force term. 
In the following, we briefly discuss this identity and provide a proof of it using the same argument of Theorem~\ref{thm-1}.
Let us consider the dynamics 
\begin{align}
  \begin{split}
    d\bar{x}(s) & =  b(\bar{x}(s), \lambda(s))\, ds + u(\bar{x}(s), \lambda(s))\,ds
    +\sqrt{2\beta^{-1}} \sigma(\bar{x}(s), \lambda(s)) \,dw^{(1)}(s)\,,\quad s \ge 0\,,
\end{split}
  \label{dynamics-1-escorted}
\end{align}
where $u: \mathbb{R}^{n}\times \mathbb{R}^{m} \rightarrow \mathbb{R}^{n}$ is a smooth vector field with compact support and $\lambda(s)$ satisfies (\ref{lambda-dynamics}).
We define the modified work 
\begin{align}
  \overline{W}_{(t_1, t_2)} = \int_{t_1}^{t_2} \nabla_\lambda V(\bar{x}(s), \lambda(s)) \circ d\lambda(s) +
  \int_{t_1}^{t_2} \Big(u\cdot \nabla V - \frac{1}{\beta} \nabla\cdot
  u\Big)(\bar{x}(s),\lambda(s))\,ds \,,
  \label{work-w-escorted}
\end{align}
for $0 \le t_1 \le t_2 \le T$.
\begin{corollary}
  Let 
  $\bar{x}(s)$ and $\lambda(s)$ be given by (\ref{dynamics-1-escorted}) and (\ref{lambda-dynamics}), respectively. 
  Then, for any bounded smooth test function $\varphi : \mathbb{R}^n \times \mathbb{R}^m \rightarrow \mathbb{R}$, 
  we have 
\begin{align}
  \overline{\mathbf{E}}_{\lambda(0),0}\Big[\varphi(\bar{x}(t), \lambda(t))\,e^{-\beta
  \overline{W}(t)}\,\Big]  = \overline{\mathbf{E}}_{\lambda(0),0}\Big[ e^{-\beta \big(F(\lambda(t)) -
  F(\lambda(0))\big)} \mathbf{E}_{\mu_{\lambda(t)}} \varphi(\cdot, \lambda(t))\Big] \,,
  \label{jarzynski-varphi-escorted}
\end{align}
$\forall~0 \le t \le T$, where $F(\cdot)$ is the free energy in (\ref{free-energy}) and
  $\overline{W}(t)=\overline{W}_{(0,t)}$ is the modified work in (\ref{work-w-escorted}). $\mathbf{E}_{\mu_{\lambda(t)}}$ denotes the expectation with
  respect to the probability measure $\mu_{\lambda(t)}$ on $\mathbb{R}^n$, 
  while $\overline{\mathbf{E}}_{\lambda(0), 0}$ denotes the conditional expectation over the
  realizations of $\bar{x}(s)$ and $\lambda(s)$, starting from fixed
  $\lambda(0)\in \mathbb{R}^m$ and the initial distribution $\bar{x}(0) \sim \mu_{\lambda(0)}$.
In particular, choosing $\varphi\equiv 1$, we have 
  \begin{align}
    \overline{\mathbf{E}}_{\lambda(0), 0}\Big[e^{-\beta \overline{W}(t)}\Big]
    = \overline{\mathbf{E}}_{\lambda(0),0}
    \Big[e^{-\beta \big(F(\lambda(t)) - F(\lambda(0))\big)}\Big]\,.
  \label{jarzynski-escorted}
  \end{align}
\label{corollary-1-escorted}
\end{corollary}
\begin{proof}
  We only sketch the proof since it is similar to the proof of Theorem~\ref{thm-1}.
  Similar to (\ref{g-def}), we introduce the function 
\begin{align}
  \begin{split}
  g(x, \lambda, t) 
    =& \overline{\mathbf{E}}_{x,\lambda,t} \Big(\varphi(\bar{x}(T), \lambda(T))\,e^{-\beta
    \overline{W}_{(t,T)}} \Big)\,,
\end{split}
\label{g-def-escorted}
\end{align}
  where $\bar{x}(t) = x\in \mathbb{R}^n$, $\lambda(t) = \lambda\in \mathbb{R}^m$ and $t\in[0,T]$.
  Using the same argument of Lemma~\ref{lemma-g}, we can verify that $g$
  satisfies the PDE 
\begin{align}
  \begin{split}
    &\partial_t g + \mathcal{L}_{1} g + \mathcal{L}_2 g
    + u\cdot\nabla g
    -2\epsilon\beta \big(\alpha^T\nabla_\lambda V\big) \cdot \big(\alpha^T
    \nabla_\lambda g\big)\\
    &+ \Big(\epsilon \beta^2 |\alpha^T\nabla_{\lambda} V|^2 - \beta
    \mathcal{L}_2V -\beta u\cdot \nabla V + \nabla\cdot u\Big) g = 0\,, \quad 0 \le t < T\,,\\
\end{split}
\label{g-pde-escorted}
\end{align}
with the terminal condition $g(\cdot,\cdot, T) = \varphi$.
  Applying Ito's formula as we did in Theorem~\ref{thm-1}, we can get
 \begin{align}
   \begin{split}
    & d\bigg[\int_{\mathbb{R}^n} e^{-\beta V(x,\lambda(s))} g(x,\lambda(s),
     s)\,dx \bigg] \\
     =& - \bigg[\int_{\mathbb{R}^n}e^{-\beta V(x,\lambda(s))}
     \Big(\mathcal{L}_1 g + u\cdot \nabla g - (\beta u\cdot\nabla V)g +
     (\nabla\cdot u) g\Big) \, dx\bigg]\,ds\\
     & + \sqrt{2\epsilon} \bigg[\int_{\mathbb{R}^n} e^{-\beta V(x,\lambda(s))}
 \alpha^T\big(\nabla_\lambda g - \beta \nabla_\lambda V\, g\big) dx\,\bigg]\cdot \,dw^{(2)}(s)\,.
 \end{split}
  \label{int-identity-escorted}
\end{align}
Since $u$ is smooth and has compact support, the first term on the right hand
  side above vanishes using integration by parts formula. 
  (\ref{jarzynski-varphi-escorted}) is obtained following the same argument in the proof of Theorem~\ref{thm-1}.
\end{proof}

\subsection{Fluctuation theorem}
\label{subsec-fluct-thm}
In this subsection we study the fluctuation theorem in the alchemical
transition case. Note that the main result below (Theorem~\ref{thm-fluct-relation}) has been obtained in~\cite{Chetrite2008}, where
comprehensive analysis as well as several concrete examples have been
presented. The main purpose of this subsection is to provide a both concise and
mathematical derivation which directly leads to
Theorem~\ref{thm-fluct-relation}. A different proof which is similar (but
shorter) to the argument in~\cite{Chetrite2008} can be found in Appendix~\ref{app-3}.

First of all, we introduce the ``reversed'' dynamics, which is closely related to the dynamics $x(s)$ in (\ref{dynamics-1}), or its vector form (\ref{dynamics-1-q-vector}). 
Notice that different reversals of stochastic dynamics have been studied in the literature
in both mathematics and physics communities. We refer to~\cite{haussmann1986,Chetrite2008} and the references therein.
In our case, we consider the dynamics $x^R(s)$ on the time interval $s \in [0,T]$, which is governed by 
\begin{align}
  d x^R(s) & = \Big(-J-a\nabla V + \frac{1}{\beta}\nabla \cdot a\Big)\big(x^R(s),
  \lambda^R(s)\big) ds + \sqrt{2\beta^{-1}} \sigma\big(x^R(s),
  \lambda^R(s)\big)\,dw^{(1)}(s)\,,
\label{dynamics-1-reversed}
\end{align}
where $\lambda^R(s)$ is the control protocol satisfying the SDE
\begin{align}
  \begin{split}
  d\lambda^R(s) =& -f\big(x^R(s), \lambda^R(s), T-s\big)\, ds +
    2\epsilon \big(\nabla_\lambda \cdot (\alpha\alpha^T)\big) \big(x^R(s),
  \lambda^R(s), T-s\big)\,ds \\
  &+ \sqrt{2\epsilon}\,\alpha\big(x^R(s), \lambda^R(s),T-s\big)\, dw^{(2)}(s)\,.
\end{split}
  \label{lambda-dynamics-inverse-full}
\end{align}
Comparing to dynamics (\ref{lambda-dynamics-full}),
we note that there is an extra term $\nabla_\lambda \cdot (\alpha\alpha^T)$ in
(\ref{lambda-dynamics-inverse-full}). 
The infinitesimal generator of the system (\ref{dynamics-1-reversed}) and (\ref{lambda-dynamics-inverse-full}) is given by 
\begin{align}
  \begin{split}
  \mathcal{L}^R = & \Big(-J-a\nabla V + \frac{1}{\beta} \nabla\cdot a\Big) \cdot \nabla
    + \frac{1}{\beta} a : \nabla^2 + \Big(2\epsilon \nabla_\lambda \cdot
    (\alpha\alpha^T) - f\Big) \cdot \nabla_\lambda + \epsilon\,\alpha\alpha^T :\nabla^2_\lambda\,\\
  =& \mathcal{L}^R_1 + \mathcal{L}^R_2 \,,
\end{split}
  \label{l-reversed}
\end{align}
where $\mathcal{L}^R_1$ is the infinitesimal generator of the dynamics
(\ref{dynamics-1-reversed}) when $\lambda^R(s)$ is fixed, and
similarly $\mathcal{L}^R_2$ is the infinitesimal generator of the dynamics
(\ref{lambda-dynamics-inverse-full}) when $x^R(s)$ is fixed.
We will also use the notation $\mathcal{L}^R_{(x, \lambda,T-t)}$ to emphasize
that functions in the operator (\ref{l-reversed}) are evaluated at $(x,
\lambda,T-t)$.

The following fluctuation result
 concerns the relation between dynamics
(\ref{dynamics-1-q-vector}), (\ref{lambda-dynamics-full}) and the reversed ones (\ref{dynamics-1-reversed}), (\ref{lambda-dynamics-inverse-full}).
\begin{theorem}
  Let $0 \le t' < t \le T$, $x,x' \in \mathbb{R}^n$ and $\lambda, \lambda' \in
  \mathbb{R}^m$. For any continuous function $\eta \in C\big(\mathbb{R}^n \times
  \mathbb{R}^m \times [0,T]\big)$ with compact support, we have 
\begin{align}
  \begin{split}
  &e^{-\beta V(x',\lambda')}\,
    \mathbf{E}^R_{x',\lambda',t'}\bigg[\exp\bigg(\int_{t'}^t \eta\big(x^R(s),
      \lambda^R(s), T-s\big)
      ds\bigg) \delta\big(x^R(t)-x\big)\,\delta\big(\lambda^{R}(t)-\lambda\big)\bigg]\\
    =&e^{-\beta V(x,\lambda)}\,\mathbf{E}_{x,\lambda,T-t}\bigg[e^{-\beta
    \mathcal{W}} \exp\bigg(\int_{T-t}^{T-t'} \eta\big(x(s), \lambda(s), s\big) ds\bigg) 
\delta\big(x(T-t')-x'\big)\delta\big(\lambda(T-t')-\lambda'\big) \bigg]\,,
  \end{split}
  \label{fluct-relation}
\end{align}
where 
\begin{align}
  \mathcal{W} = \int_{T-t}^{T-t'} \nabla_\lambda V\big(x(s), \lambda(s)\big) \circ
  d\lambda(s) - \frac{1}{\beta} \int_{T-t}^{T-t'} \Big[\mbox{div}_\lambda \big(f -
\epsilon \nabla_\lambda\cdot(\alpha\alpha^T)\big)\Big] \big(x(s),
\lambda(s),s\big) ds\,,
  \label{w-div-f}
\end{align}
$x^R(\cdot), \lambda^R(\cdot)$ satisfy the dynamics (\ref{dynamics-1-reversed}),
(\ref{lambda-dynamics-inverse-full}), and $x(\cdot), \lambda(\cdot)$ satisfy
the dynamics (\ref{dynamics-1-q-vector}), (\ref{lambda-dynamics-full}), respectively.
  Here, $\delta(\cdot)$ denotes the Dirac delta function (see
  Remark~\ref{rmk-delta} below) and $\mbox{div}_\lambda$ denotes the divergence operator with respect to
$\lambda \in \mathbb{R}^m$.  $\mathbf{E}^R_{x',\lambda',t'}$ is the conditional expectation with respect to the path
ensemble of the dynamics (\ref{dynamics-1-reversed}), 
(\ref{lambda-dynamics-inverse-full}) starting from $x^R(t') = x'$ and
$\lambda^R(t')=\lambda'$ at time $t'$, 
while $\mathbf{E}_{x,\lambda,T-t}$ is the conditional expectation with respect to the dynamics
(\ref{dynamics-1-q-vector}) and (\ref{lambda-dynamics-full}).
\label{thm-fluct-relation}
\end{theorem}
\begin{proof}
  We consider the quantities on both sides of the equality (\ref{fluct-relation}).
  For the left hand side of (\ref{fluct-relation}), 
let us fix the values $(x',\lambda',t') \in  \mathbb{R}^n \times \mathbb{R}^m \times [0,T]$ and define the function $u$ by 
\begin{align}
u\big(x,\lambda,t\,;x',\lambda',t'\big) = 
\mathbf{E}_{x',\lambda',t'}^R\bigg[
\exp\bigg(\int_{t'}^t \eta\big(x^R(s), \lambda^R(s), T-s\big) ds\bigg) 
\delta\big(x^R(t)-x\big)\delta\big(\lambda^{R}(t)-\lambda\big)\bigg]\,,
\label{u-exp-form}
\end{align}
for $(x,\lambda,t) \in  \mathbb{R}^n \times \mathbb{R}^m  \times [0,T]$.
It is known that $u$ satisfies the PDE 
\begin{align}
  \begin{split}
  &\frac{\partial u}{\partial t} = \big(\mathcal{L}^R_{(x, \lambda,T-t)})^* u
  + \eta(x,\lambda,T-t) \,u \,, \quad \forall\, (x, \lambda,t) \in  \mathbb{R}^n\times \mathbb{R}^m \times (t',T] \,,\\
  & u(x, \lambda,t\,;x',\lambda',t')=\delta(x-x')\,\delta(\lambda-\lambda')\,,
\quad \mbox{if}~~t=t'\,,
  \end{split}
  \label{pde-u-forward}
\end{align}
where the operator $\mathcal{L}_{(x, \lambda,T-t)}^R$ is defined in (\ref{l-reversed}) and 
$\big(\mathcal{L}^R_{(x, \lambda,T-t)}\big)^*$ denotes its formal $L^2$ adjoint.
  Direct calculation shows that, after some cancellation, we have  
\begin{align}
  \begin{split}
  \big(\mathcal{L}^R_{(x, \lambda,T-t)}\big)^*\phi = & \Big[\mbox{div}(J + a\nabla V) + \mbox{div}_\lambda
  \Big(f -\epsilon \nabla_\lambda\cdot (\alpha\alpha^T)\Big)\Big]\phi + \Big(J + a\nabla V +
  \frac{1}{\beta} \nabla\cdot a\Big) \cdot \nabla \phi \\
  & + \frac{1}{\beta} a :
  \nabla^2 \phi + f \cdot \nabla_\lambda \phi + \epsilon\,\alpha\alpha^T : \nabla^2_\lambda \phi\,,
\end{split}
  \label{l-r-trans}
\end{align}
  for a smooth function $\phi$.  

For the right hand side of (\ref{fluct-relation}), we define the function $g$
for fixed $(x', \lambda',t')$ as
\begin{align*}
  g(x,\lambda,t) = 
  \mathbf{E}_{x,\lambda,T-t}\bigg[&e^{-\beta \mathcal{W}}
\exp\bigg(\int_{T-t}^{T-t'} \eta\big(x(s), \lambda(s), s\big) ds\bigg) \\
& \times \delta\big(x(T-t')-x'\big)\delta\big(\lambda(T-t')-\lambda'\big)
\bigg]\,,
\end{align*}
where $\mathcal{W}$ is defined in (\ref{w-div-f}), and the dynamics $x(\cdot), \lambda(\cdot)$ satisfies SDEs (\ref{dynamics-1-q-vector}),
(\ref{lambda-dynamics-full}). 
Using the same argument as in Lemma~\ref{lemma-g}, we can verify that the function
$g$ satisfies the PDE
\begin{align}
  \begin{split}
    &\frac{\partial g}{\partial t} = \overline{\mathcal{L}}_{( x,\lambda,T-t)}\, g\,,\qquad
  \forall\, (x,\lambda,t) \in  \mathbb{R}^n \times \mathbb{R}^m \times (t',T] \,,\\
  & g(x,\lambda,t) = \delta(x-x')\delta(\lambda-\lambda') \,, \qquad
  \mbox{if}~~t=t'\,,
\end{split}
\label{g-delta}
\end{align}
where the operator $\overline{\mathcal{L}}_{(x,\lambda,T-t)}$ is defined as 
\begin{align}
  \begin{split}
  \overline{\mathcal{L}}_{(x,\lambda,T-t)}\,\phi =& 
\Big[\epsilon \beta^2 |\alpha^T\nabla_{\lambda} V|^2 - \beta \mathcal{L}_2V +
   \mbox{div}_\lambda \Big(f - \epsilon \nabla_\lambda\cdot
   (\alpha\alpha^T)\Big) + \eta \Big] \phi \\
   & + \mathcal{L}_{1} \phi + \mathcal{L}_2 \phi 
  -2\epsilon\beta \big(\alpha^T\nabla_\lambda V\big) \cdot \big(\alpha^T\nabla_\lambda \phi\big)  
\end{split}
    \label{bar-l}
\end{align}
for a smooth function $\phi$, and the functions in (\ref{bar-l}) are evaluated
at $(x,\lambda,T-t)$. Motivated by the right hand side of
(\ref{fluct-relation}), now a key step is to consider
the function
$\omega(x,\lambda,t) = e^{-\beta V(x, \lambda)} g(x,\lambda,t)$.
Recalling the relation (\ref{div-j-zero}), a direct calculation shows that 
\begin{align}
  \begin{split}
  e^{-\beta V} \mathcal{L}_1 g = &
    e^{-\beta V} \Big(J - a\nabla V + \frac{1}{\beta} \nabla \cdot a\Big) \cdot \nabla
    \big(e^{\beta V}\omega\big)  + \frac{e^{-\beta V}}{\beta} a:\nabla^2 \big(e^{\beta V}\omega\big)\\
    =& 
    \Big(J - a\nabla V + \frac{1}{\beta} \nabla \cdot a\Big) \cdot \nabla \omega 
    + \beta \Big[\big(J - a\nabla V + \frac{1}{\beta} \nabla \cdot a\big) \cdot \nabla V
    \Big] \omega \\
    & + \frac{1}{\beta} a:\nabla^2 \omega + 2 (a\nabla V)\cdot  \nabla \omega 
    + \frac{e^{-\beta V}\omega}{\beta} a:\nabla^2\big(e^{\beta V}\big)\\
    =& 
\Big[\mbox{div}(J + a\nabla V)\Big]
    \omega  + 
    \Big(J + a\nabla V + \frac{1}{\beta} \nabla \cdot a\Big) \cdot \nabla \omega 
    + \frac{1}{\beta} a:\nabla^2 \omega \,,\\
    e^{-\beta V} \mathcal{L}_2 g = &
    e^{-\beta V} \Big[f\cdot \nabla_{\lambda} (e^{\beta V} \omega) +
    \epsilon\,\alpha\alpha^T:\nabla^2_{\lambda} (e^{\beta V} \omega) \Big] \\
    =& 
    \mathcal{L}_2 \omega + \beta (\mathcal{L}_2 V)\omega + 2\epsilon\beta
    \big(\alpha^T\nabla_\lambda V\big) \cdot \big(\alpha^T\nabla_\lambda
    \omega\big) + \epsilon\beta^2|\alpha^T\nabla_\lambda V|^2\, \omega\,, \\
    e^{-\beta V} \nabla_\lambda g = &
    e^{-\beta V} \nabla_\lambda \big(e^{\beta V} \omega\big) = \beta \big(\nabla_\lambda
    V\big) \omega + \nabla_\lambda \omega\,.
  \end{split}
  \label{e-v-l1-l2}
\end{align}
Combining (\ref{l-reversed}), (\ref{g-delta}), (\ref{bar-l}),
(\ref{e-v-l1-l2}), we can conclude that $\omega$ satisfies PDE
\begin{align*}
  &\frac{\partial \omega}{\partial t} = e^{-\beta V}
\overline{\mathcal{L}}_{(x,\lambda,T-t)}\,g =
\big(\mathcal{L}^R_{(x,\lambda,T-t)}\big)^*
  \,\omega + \eta(x,\lambda, T-t)\,\omega \,,\quad \forall\,( x, \lambda,t) \in  \mathbb{R}^n\times \mathbb{R}^m
\times (t',T]  \,,\\
  &\omega(x,\lambda,t) = e^{-\beta V(x',\lambda')}
  \delta(x-x')\delta(\lambda-\lambda')\,,\quad \mbox{if}~~ t=t'\,.
\end{align*}
Comparing the latter with (\ref{pde-u-forward}),
we obtain that $e^{-\beta V(x',\lambda')} u(x,\lambda,t\,;x',\lambda',t') =
\omega(x,\lambda,t)$, which is equivalent to the equality (\ref{fluct-relation}). 
\end{proof}
\begin{remark}
  We have adopted the Dirac delta function both in Theorem~\ref{thm-fluct-relation} and in its proof above, 
  in order to simplify the derivations. Precisely, (\ref{fluct-relation}) should be understood in the sense of
  distributions, or equivalently,  
\begin{align}
  \begin{split}
    &\int_{\mathbb{R}^n} \int_{\mathbb{R}^m} e^{-\beta V(x',\lambda')}\,
    \mathbf{E}^R_{x',\lambda',t'}\bigg[\exp\bigg(\int_{t'}^t \eta\big(x^R(s),
      \lambda^R(s), T-s\big)
      ds\bigg) \varphi\big(x^R(t), \lambda^{R}(t), x', \lambda'\big) \bigg]
      dx' d\lambda' \\
    =&\int_{\mathbb{R}^n} \int_{\mathbb{R}^m} e^{-\beta V(x,\lambda)}\,\mathbf{E}_{x,\lambda,T-t}\bigg[e^{-\beta
    \mathcal{W}} \exp\bigg(\int_{T-t}^{T-t'} \eta\big(x(s), \lambda(s), s\big) ds\bigg) 
    \varphi\big(x, \lambda, x(T-t'), \lambda(T-t')\big)\bigg] dx\,d\lambda\,,
  \end{split}
  \label{fluct-relation-test-function}
\end{align}
  for all test functions $\varphi(x,\lambda, x', \lambda')$ which are smooth enough with compact support.
  We emphasize that the above proof can be reformulated more rigorously, by introducing test functions and applying integration by parts.
  \label{rmk-delta}
\end{remark}
\textbf{From fluctuation theorems to Jarzynski's equality}.
It is well known that Jarzynski's equality can be obtained from the fluctuation theorem~\cite{Chetrite2008}. 
In the remaining part of this subsection, we consider the case when
the control protocol $\lambda(s)$ satisfies the dynamics (\ref{lambda-dynamics}) and
show that Theorem~\ref{thm-1} is a consequence of Theorem~\ref{thm-fluct-relation}.
In this case, (\ref{lambda-dynamics-inverse-full}) governing the reversed
protocol $\lambda^R(\cdot)$ simplifies to 
\begin{align}
  \begin{split}
  d\lambda^R(s) =& -f\big(\lambda^R(s), T-s\big)\, ds 
    +  2\epsilon \big(\nabla_\lambda \cdot (\alpha\alpha^T)\big)
    \big(\lambda^R(s),T-s\big)\,ds \\
				  & + \sqrt{2\epsilon}\,
    \alpha\big(\lambda^R(s),T-s\big)\,dw^{(2)}(s)\,,
  \end{split}
  \label{lambda-dynamics-inverse}
\end{align}
and therefore is independent of the process $x^R(\cdot)$ in (\ref{dynamics-1-reversed}). 
For simplicity, we only prove the equality (\ref{generalized-jarzynski-varphi}) for $t=T$.

In order to derive the equality (\ref{generalized-jarzynski-varphi}) in Theorem~\ref{thm-1}, 
we set $t'=0, t=T$ and $\eta
=-\mbox{div}_\lambda \big(f-\epsilon
\nabla_\lambda\cdot(\alpha\alpha^T)\big)$, which is a function independent of $x\in
\mathbb{R}^n$. Multiplying $\varphi(x',\lambda')$ on both sides
of the equality (\ref{fluct-relation}), integrating with respect to $x, x', \lambda'$, 
and recalling the definition (\ref{u-q-w}) of the work $W$, we obtain
\begin{align}
  \begin{split}
      &\int_{\mathbb{R}^n} e^{-\beta V(x,\lambda)}\,\mathbf{E}_{x,\lambda,0}\Big(\varphi(x(T), \lambda(T))\,
      e^{-\beta W}\Big)\,dx \\
    =&\int_{\mathbb{R}^n} \int_{\mathbb{R}^m} \varphi(x',\lambda')\,e^{-\beta V(x',\lambda')}\,
  \mathbf{E}^R_{x',\lambda',0}\bigg[\exp\bigg(\int_{0}^T \eta\big(\lambda^R(s),
  T-s\big) ds\bigg) \delta(\lambda^{R}(T)-\lambda)\bigg] dx' d\lambda' \,.
\end{split}
\label{thm-2-identity-0}
\end{align}
Notice that the conditional expectation on the right hand side of 
(\ref{thm-2-identity-0}) is actually independent of $x'$ (This is only true
when the control protocol doesn't depend on the dynamics. See Remark~\ref{rmk-1}.). We have 
\begin{align}
\begin{split}
      &\int_{\mathbb{R}^n} e^{-\beta
      V(x,\lambda)}\,\mathbf{E}_{x,\lambda,0}\Big(\varphi(x(T),
      \lambda(T))\,e^{-\beta W}\Big)\, dx \\
      =& \int_{\mathbb{R}^m} \big[\mathbf{E}_{\mu_{\lambda'}}
  \varphi(\cdot,\lambda')\big]\,Z(\lambda') \mathbf{E}^R_{\lambda',0}\bigg[\exp\bigg(\int_{0}^T 
  \eta \big(\lambda^R(s), T-s\big) ds\bigg) \delta(\lambda^{R}(T)-\lambda)\bigg]\,
  d\lambda'\,,
\end{split}
\label{thm-2-identity-1}
\end{align}
where $Z(\cdot)$ is the normalization constant in (\ref{normal-const}).

More generally, let us define the function
\begin{align*}
  \psi(\lambda,t) 
  =& \int_{\mathbb{R}^m} \big[\mathbf{E}_{\mu_{\lambda'}}
  \varphi(\cdot,\lambda')\big]\,Z(\lambda')
  \mathbf{E}^R_{\lambda',0}\bigg[\exp\bigg(\int_{0}^{T-t} 
    \eta(\lambda^R(s), T-s) ds\bigg) \delta(\lambda^{R}(T-t)-\lambda)\bigg]\, d\lambda'\,.
\end{align*}
Similarly to the function $u$ in (\ref{u-exp-form}) which satisfies the PDE
(\ref{pde-u-forward}), we know that $\psi$ satisfies 
\begin{align}
  \begin{split}
  & \frac{\partial \psi}{\partial t} + \big(\mathcal{L}^R_2\big)^* \psi -
\Big[\mbox{div}_\lambda \Big(f-\epsilon \nabla_\lambda \cdot
(\alpha\alpha^T)\Big)\Big]\psi = 0 \,,\quad \forall~ (\lambda,t) \in 
  \mathbb{R}^m \times [0, T) \,,\\ 
  & \psi(\lambda,T) = Z(\lambda)\mathbf{E}_{\mu_{\lambda}}
  \varphi(\cdot,\lambda)\,\,,
\end{split}
\label{psi-pde}
\end{align}
where $\mathcal{L}^R_2 = \big(2\epsilon \nabla_\lambda\cdot (\alpha\alpha^T)
-f\big) \cdot \nabla_\lambda + \epsilon\,\alpha\alpha^T:\nabla^2_\lambda$,
and the functions in (\ref{psi-pde}) are evaluated at $(\lambda, t)$.
Calculating $(\mathcal{L}^R_2)^*$, one can conclude that (\ref{psi-pde}) is equivalent to 
\begin{align}
  \begin{split}
  & \frac{\partial \psi}{\partial t} + \mathcal{L}_2 \psi = 0 \,,\quad
    \forall~ (\lambda,t) \in 
  \mathbb{R}^m \times [0, T) \,,\\ 
  & \psi(\lambda,T) = Z(\lambda)\mathbf{E}_{\mu_{\lambda}} \varphi(\cdot,\lambda)\,,
\end{split}
\label{psi-pde-1}
\end{align}
where $\mathcal{L}_2$ is the infinitesimal generator defined in (\ref{l-lambda})
for the dynamics (\ref{lambda-dynamics}), and 
therefore the Feynman-Kac formula implies that 
$$\psi(\lambda,t) =
\mathbf{E}_{\lambda,t}\Big[Z\big(\lambda(T)\big)\mathbf{E}_{\mu_{\lambda(T)}} \varphi(\cdot,\lambda(T))\Big].$$
Combining this with the identity in (\ref{thm-2-identity-1}), we conclude that
\begin{align*}
    &\int_{\mathbb{R}^n} e^{-\beta
  V(x,\lambda)}\,\mathbf{E}_{x,\lambda,0}\Big(\varphi(x(T), \lambda(T))\,
  e^{-\beta W}\Big) dx = \psi(\lambda,0) = 
\mathbf{E}_{\lambda,0}\big[Z\big(\lambda(T)\big)\mathbf{E}_{\mu_{\lambda(T)}} \varphi(\cdot,\lambda(T))\big]\,,
\end{align*}
which is equivalent to the equality (\ref{generalized-jarzynski-varphi}) in
Theorem~\ref{thm-1} for $t=T$. $\hfill\square$
\vspace{0.2cm}

  In the above analysis, we have assumed that the control protocol $\lambda(s)$ is perturbed by noise.
Let us now consider the case when $\lambda(s)$ is deterministic, i.e., when $\epsilon = 0$ in dynamics
  (\ref{lambda-dynamics}). In this case, we have 
\begin{align}
  \dot{\lambda}(s) = f(\lambda(s), s)\,, \quad 0 \le s \le T\,,
  \label{ode-control}
\end{align}
  and $\lambda^R(s) = \lambda(T-s)$.
  It is well known that Crooks's
  relations~\cite{crooks-path-ensemble-pre2000} can be derived from the
  fluctuation relation~\cite{Chetrite2008,escorted-simulation2011}. In the
  following remark, for simplicity we will only state Crooks's relations for the escorted dynamics
  (\ref{dynamics-1-escorted}). Results corresponding to the original dynamics (\ref{dynamics-1}) can be recovered by choosing $u\equiv 0$.

\begin{remark}[Crooks's relations for the escorted dynamics]
  Consider the reversed version of the escorted dynamics
  (\ref{dynamics-1-escorted}), which satisfies 
\begin{align}
  \begin{split}
    d \bar{x}^{R}(s)  =& \Big(-J-a\nabla V + \frac{1}{\beta}\nabla \cdot
    a\Big)\big(\bar{x}^R(s), \lambda^R(s)\big)\,ds 
    - u(\bar{x}^{R}(s), \lambda^R(s))\,ds \\
    & +\sqrt{2\beta^{-1}} \sigma(\bar{x}^{R}(s), \lambda^R(s)) \,dw^{(1)}(s)\,,
    \quad s \ge 0\,.
\end{split}
  \label{dynamics-1-escorted-reversed}
\end{align}
  By slightly modifying the proof of Theorem~\ref{thm-fluct-relation}, we can prove
\begin{align}
  \begin{split}
    &e^{-\beta V(x',\lambda(T))}\,
    \overline{\mathbf{E}}^R_{x',0}\bigg[\exp\bigg(\int_{0}^T
    \eta\big(\bar{x}^{R}(s), T-s\big) ds\bigg) \delta\big(\bar{x}^{R}(T)-x\big)\bigg]\\
    =&e^{-\beta V(x,\lambda(0))}\,\overline{\mathbf{E}}_{x,0}\bigg[e^{-\beta \overline{W}}
    \exp\bigg(\int_{0}^{T} \eta\big(\bar{x}(s), s\big)\,ds\bigg)
    \delta\big(\bar{x}(T)-x'\big) \bigg]\,,
    \quad \forall~x,x' \in \mathbb{R}^n\,,
  \end{split}
  \label{fluct-relation-eps0}
\end{align}
  where $\overline{W}=\overline{W}_{(0,T)}$ is the modified work in (\ref{work-w-escorted}) and $\eta \in C\big(\mathbb{R}^n \times [0,T]\big)$ is continuous with compact support. 
The notations $\overline{\mathbf{E}}_{x,0}$ and
  $\overline{\mathbf{E}}^R_{x',0}$ denote the ensemble averages with respect
  to the escorted dynamics $\bar{x}(\cdot)$ in (\ref{dynamics-1-escorted}) and
  its reversed counterpart $\bar{x}^{R}(\cdot)$ in (\ref{dynamics-1-escorted-reversed}) starting from fixed state at time $s=0$, respectively.

  Since any (bounded) continuous function $\mathcal{G}$ on the path space can be approximated by linear combinations of functions
  which are of the form $\exp\big(\int_0^T \eta(\bar{x}(s),s)\,ds\big)$ (for instance, by
  discretizing $[0,T]$ into subintervals), integrating
  (\ref{fluct-relation-eps0}) gives
  \begin{align}
    \frac{\overline{\mathbf{E}}_{\lambda(0), 0} \big(e^{-\beta
    \overline{W}}\mathcal{G}\big)}{\overline{\mathbf{E}}^{R}_{\lambda(T), 0}\,(\mathcal{G}^{R})}
    = e^{-\beta \Delta F(T)}\,,
    \label{crooks-eps0-1}
  \end{align}
  where $\mathcal{G}^R\big(x(\cdot)\big)=\mathcal{G}\big(x(T-\cdot)\big)$ for all path
  $x(\cdot) \in C\big([0,T], \mathbb{R}^n\big)$, and $\Delta F(T)$ is the free energy difference in (\ref{delta-f}). 
  The notation $\overline{\mathbf{E}}_{\lambda(0), 0}$ is the path ensemble average of the
  forward dynamics $\bar{x}(s)$ starting from $\bar{x}(0)\sim \mu_{\lambda(0)}$, and
  $\overline{\mathbf{E}}^R_{\lambda(T), 0}$ is defined similarly for the reversed
  dynamics $\bar{x}^{R}(s)$. If we formally write
  $\overline{\mathcal{P}}[\bar{x}(\cdot)\,|\,\bar{x}(0)]$,
  $\overline{\mathcal{P}}^R[\bar{x}^{R}(\cdot)\,|\,\bar{x}^{R}(0)]$ as the probability densities on the
  path space for the dynamics $\bar{x}(s)$, $\bar{x}^{R}(s)$ starting
  from $\bar{x}(0)$ and $\bar{x}^{R}(0)$ respectively, we obtain from (\ref{crooks-eps0-1}) that 
  \begin{align}
    \frac{\overline{\mathcal{P}}[x(\cdot)\,|\,x(0)]}
    {\overline{\mathcal{P}}^R[x(T-\cdot)\,|\,x(T)]}
    = e^{-\beta (\Delta \mathcal{U}(T)- \overline{W})}\,,\quad \forall~x(\cdot) \in C\big([0,T], \mathbb{R}^n\big)\,,
    \label{micro-reversibility}
  \end{align}
  where $\Delta \mathcal{U}(T)$ is the change of internal energy in (\ref{u-q-w}).

  Furthermore, notice that for the work function $\mathcal{G}\big(x(\cdot)) = \overline{W}$ in (\ref{work-w-escorted}), 
  we have 
  \begin{align*}
    \mathcal{G}^R\big(x(\cdot)) = &
     \mathcal{G}\big(x(T-\cdot))\\
    =& \int_0^T \Big(\nabla_\lambda V\cdot f + u\cdot \nabla V - \frac{1}{\beta}
    \nabla \cdot u\Big)(x(s),\lambda(T-s), T-s)\, ds \\
    =& -\int_0^T \Big(\nabla_\lambda V\cdot \dot{\lambda}^R + (-u)\cdot \nabla V - \frac{1}{\beta}
    \nabla \cdot (-u)\Big)(x(s),\lambda^R(s), s)\, ds \\
    =& -\overline{W}^R,
  \end{align*} 
    where $\overline{W}^R$ is the modified work of the reversed dynamics
    (\ref{dynamics-1-escorted-reversed}).
   Therefore, (\ref{crooks-eps0-1}) implies 
  \begin{align}
    \frac{\overline{\mathbf{E}}_{\lambda(0), 0} \big(e^{-\beta
    \overline{W}}\phi(\overline{W})\big)}{\overline{\mathbf{E}}^{R}_{\lambda(T),
    0}\,(\phi(-\overline{W}^R))}
    = e^{-\beta \Delta F(T)}\,,\quad \forall~ \phi \in C_b(\mathbb{R})\,.
    \label{crooks-eps0-2}
  \end{align}

  Readers can recognize that the identities (\ref{micro-reversibility}),
  (\ref{crooks-eps0-1}) and (\ref{crooks-eps0-2}) are the counterparts of the microscopic reversibility
  and Crooks's relations in~\cite{crooks-path-ensemble-pre2000,escorted-simulation2011} for 
  (escorted) continuous-time Markovian processes.
  It was already pointed out in~\cite{crooks-path-ensemble-pre2000} that these
  relations (in particular the microscopic reversibility) hold for general
  Markov chains out of equilibrium without reversibility assumption.
  The derivations above show that this is also true for the continuous-time
  process $\bar{x}(s)$ in (\ref{dynamics-1-escorted}) with the control protocol in (\ref{ode-control}). 
  \label{rmk-crook-jarzynski-eps0}
\end{remark}
\subsection{Change of measure and information-theoretic formulation}
\label{subsec-is}
In this subsection, we explore the idea of importance
sampling~\cite{ce_paper2014,Hartmann2016-Nonlinearity} to study the
Jarzynski's equality. We focus on the case when the control protocol $\lambda(s)$
is deterministic and satisfies the ODE (\ref{ode-control}), i.e.  $\epsilon = 0$ in dynamics (\ref{lambda-dynamics}). 
For simplicity, we also assume that the coefficient matrix $\sigma$ in
dynamics (\ref{dynamics-1}) is an invertible $n\times n$ matrix.
Denote $\mathbf{P}$, $\mathbf{E}$ as the probability measure 
and the mathematical expectation on
path space $C\big([0,T], \mathbb{R}^n\big)$ 
with respect to paths of the process (\ref{dynamics-1-q-vector})  starting from
$x(0) \sim \mu_{\lambda(0)}$, where $\lambda(s)$ satisfies (\ref{ode-control})
with fixed $\lambda(0) \in \mathbb{R}^m$.
Then the Jarzynski's equality (\ref{generalized-jarzynski}) reads 
\begin{align}
  \mathbf{E}\Big[ e^{-\beta W}\Big] = e^{-\beta \Delta F}\,,
  \label{jarzynski-repeat}
\end{align}
where $\Delta F = F\big(\lambda(T)\big)-F\big(\lambda(0)\big)$, with 
\begin{align}
  W=\int_0^T \nabla_\lambda V\big(x(s), \lambda(s)\big) \cdot f\big(\lambda(s), s\big) ds\,.
  \label{w-repeat}
\end{align}
See Remark~\ref{rmk-1} for related discussions.

Let $\overline{\mathbf{P}}$ be another probability measure on
the space $C\big([0,T], \mathbb{R}^n\big)$ which is equivalent to
$\mathbf{P}$ and let $\overline{\mathbf{E}}$ be the corresponding expectation.
Applying a change of measure in (\ref{jarzynski-repeat}), together with Jensen's inequality, we can deduce 
\begin{align}
  \begin{split}
  \Delta F =&\, -\beta^{-1} \ln \overline{\mathbf{E}} \Big(e^{-\beta W}
  \frac{d\mathbf{P}}{d\overline{\mathbf{P}}}\Big) \\
  \le&\, \overline{\mathbf{E}} \Big(W + \beta^{-1} \ln
  \frac{d\overline{\mathbf{P}}}{d\mathbf{P}}\Big) \\
  =&\, \overline{\mathbf{E}}(W) + \beta^{-1}
  D_{KL}\big(\overline{\mathbf{P}}\,\|\,\mathbf{P}\big)\,,
\end{split}
\label{df-w-ineq}
\end{align}
where $D_{KL}\big(\cdot\,\|\,\cdot\big)$ denotes the Kullback-Leibler
divergence of two probability measures~\cite{MacKay2002-inf-theory,Bishop2006-pattern-recognition}.  
Notice that the inequality (\ref{df-w-ineq}) can be interpreted as a
generalization of the second law of thermodynamics~\cite{callen1985thermodynamics}.
In particular, under certain conditions on the work $W$, the equality in
(\ref{df-w-ineq}) can be attained by the optimal probability measure
$\mathbf{P}^*$, which is determined by 
\begin{align}
  \frac{d\mathbf{P}^*}{d\mathbf{P}} = e^{-\beta(W-\Delta F)}\,,\qquad
  \mathbf{P}^*-a.s.
  \label{opt-p}
\end{align}
In other words, the optimal change of measure
tilts the original path probabilities exponentially according to the
differences between the work $W$ and the free energy difference $\Delta F$. 
In particular, the probability of paths with smaller work $W$ (compared to
$\Delta F$) increases under the optimal measure.

Meanwhile, the importance sampling Monte Carlo estimator for the free energy
difference $\Delta F$ based on the identity
\begin{align}
\Delta F =-\beta^{-1} \ln\mathbf{E}^* \Big(e^{-\beta W}
\frac{d\mathbf{P}}{\,d\mathbf{P}^*}\Big)
\label{free-energy-optimal-estimator}
\end{align}
will achieve zero variance. 
More generally, inspired by the last line in (\ref{df-w-ineq}), we
define 
\begin{align}
  \Phi(\overline{\mathbf{P}}) := \overline{\mathbf{E}}\big(W) + \beta^{-1}
  D_{KL}\big(\overline{\mathbf{P}}\,\|\,\mathbf{P}\big) \,,
  \label{phi-cost}
\end{align}
for a general probability measure $\overline{\mathbf{P}}$ which is equivalent
to $\mathbf{P}$. Then the above discussions imply the following variational principle 
\begin{align}
  \begin{split}
  \Delta F =& \inf_{\overline{\mathbf{P}}\sim \mathbf{P}}  
    \Big[\overline{\mathbf{E}}\big(W) + \beta^{-1}
      D_{KL}\big(\overline{\mathbf{P}}\,\|\, \mathbf{P}\big)\Big] \\
 =& \inf_{\overline{\mathbf{P}}\sim \mathbf{P}} \Phi(\overline{\mathbf{P}})  =
    \Phi(\mathbf{P}^*)\,,
\end{split}
\label{variation-form}
\end{align}
where `$\sim$' denotes the equivalence relation between two probability measures.
In other words, the optimal probability measure $\mathbf{P}^*$ in (\ref{opt-p}) can be characterized as the minimizer of the
minimization problem (\ref{variation-form}) and the corresponding minimum
equals to $\Delta F$. Furthermore, using (\ref{opt-p})
and (\ref{phi-cost}), we can verify the following simple relation 
\begin{align}
  \begin{split}
  \Phi(\overline{\mathbf{P}}) =& \overline{\mathbf{E}} \bigg(W + \beta^{-1} \ln
  \frac{d\overline{\mathbf{P}}}{d\mathbf{P}}\bigg) \\
  = & \mathbf{E}^*\bigg[\bigg(W + \beta^{-1}
  \ln\frac{d\overline{\mathbf{P}}}{d\mathbf{P}}\bigg)
\frac{d\overline{\mathbf{P}}}{\,d\mathbf{P}^*}\bigg] \\
=&  \mathbf{E}^*\bigg[\bigg(\Delta F +  \beta^{-1} \ln
\frac{d\mathbf{P}}{\,d\mathbf{P}^*} + \beta^{-1} \ln
\frac{d\overline{\mathbf{P}}}{d\mathbf{P}}\bigg)
\frac{d\overline{\mathbf{P}}}{\,d\mathbf{P}^*}\bigg] \\
=& \Delta F + 
\beta^{-1} \mathbf{E}^*\bigg[\bigg( \ln
\frac{d\overline{\mathbf{P}}}{\,d\mathbf{P}^*}\bigg)
\frac{d\overline{\mathbf{P}}}{\,d\mathbf{P}^*}\bigg] \\
=& \Delta F + \beta^{-1} D_{KL}\big(\overline{\mathbf{P}}\,\|\, \mathbf{P}^*\big)\,,
\end{split}
\label{entropy-exp}
\end{align}
for a general probability measure $\overline{\mathbf{P}}$ such that
$\overline{\mathbf{P}} \sim \mathbf{P}$.
It becomes apparent from the last expression in (\ref{entropy-exp}) that $\Delta F$ is
the global minimum of the function $\Phi$ and is attained by the (unique) probability measure
$\mathbf{P}^*$, since $D_{KL}\big(\overline{\mathbf{P}}\,\|\,
\mathbf{P}^*\big)\ge 0$ and the equality is achieved if and only if 
$\overline{\mathbf{P}}=\mathbf{P}^*$. Furthermore, minimizing the function
$\Phi$ is equivalent to minimizing the Kullback-Leibler divergence $D_{KL}\big(\cdot\,\|\, \mathbf{P}^*\big)$. 

In the following, we show that the optimal change of measure $\mathbf{P}^*$ can
be characterized more transparently. To this end, let $\mathbf{P}_{x,t}$,
$\mathbf{E}_{x,t}$ denote the path measure and the conditional expectation of the process
(\ref{dynamics-1-q-vector}) starting from a fixed state $x \in \mathbb{R}^n$ at time $t$. 
Notice that, by the disintegration
theorem~\cite[Theorem~$5.3.1$]{ambrosio2005gradient}, we can write the path
measure $\mathbf{P}$ as 
$$\mathbf{P}=\int_{\mathbb{R}^n} \mathbf{P}_{x,0}\, d\mu_{\lambda(0)}(x).$$
Defining the function 
\begin{align}
  g(x,t) = \mathbf{E}_{x,t} \big(e^{-\beta W_{(t,T)}}\big)\,,
  \label{g-fun-repeat}
\end{align}
analogously to (\ref{g-def}), Jarzynski's equality (\ref{jarzynski-repeat}) implies 
that 
\begin{align}
  \Delta F = -\beta^{-1} \ln \big(\mathbf{E}_{\mu_{\lambda(0)}} g(\cdot,
  0)\big)\,.
  \label{jarzynski-g}
\end{align}
Sampling an expectation value whose form is similar to (\ref{g-fun-repeat}) 
using importance sampling Monte Carlo method has been studied in previous
work~\cite{ip-dupuis-multiscale,ip-kostas1,ip-eric,ce_paper2014,Hartmann2017-ptrf,Hartmann2016-Nonlinearity}.
In particular, we know from the Feynman-Kac formula that $g$ solves the PDE 
\begin{align}
  &\partial_t g + \mathcal{L}_1 g -\beta (f\cdot \nabla_\lambda V)g =
    0\,,\quad  g(\cdot, T) = 1\,,
    \label{g-pde-repeat}
\end{align}
where $\mathcal{L}_1$ is the infinitesimal generator in (\ref{l-1}) with
$\lambda=\lambda(\cdot)$ being dependent on time $t$.
Introducing $U=-\beta^{-1}\ln g$, it follows from (\ref{g-pde-repeat}) that
$U$ satisfies a Hamilton-Jacobi-Bellman equation
\begin{align}
  \begin{split}
  &\partial_t U + \min_{c\in \mathbb{R}^n}\Big\{\mathcal{L}_1 U + \sigma c\cdot \nabla U + \frac{|c|^2}{4} +
  (f \cdot \nabla_\lambda V)\Big\} = 0\,,\\
  & U(\cdot, T) = 0\,,
  \end{split}
\end{align}
and one can show~\cite{fleming2006} that $U$ is the value function of the optimal control problem 
\begin{align}
   U(x,t) = \inf_{u_s} \mathbf{E}^u_{x,t} \bigg[\int_t^T \Big(
   \nabla_\lambda V\big(x^u(s), \lambda(s)\big)
 \cdot f(\lambda(s),s) + \frac{|u_s|^2}{4}\Big) ds\bigg]\,, 
\end{align}
where $u_s \in \mathbb{R}^n$ is the control policy, $x^u(s)$ is the controlled process given by
  \begin{align}
    d x^u(s)  =  b(x^u(s), \lambda(s)) ds + \sigma(x^u(s), \lambda(s))u_s\,ds + \sqrt{2\beta^{-1}} \sigma(x^u(s), \lambda(s)) \,dw^{(1)}(s)\,,
  \label{dynamics-1-u}
\end{align}
and $\mathbf{E}^u_{x,t}$ denotes the corresponding conditional expectation
starting from $x^u(t) = x$ at time $t$. 

In particular, it is well known that the feedback control policy 
\begin{align}
  u^*_s(x) = - 2 \sigma^T(x,\lambda(s)) \nabla U(x,s)
  = 2 \beta^{-1}
  \frac{\sigma^T(x,\lambda(s))\nabla
  g(x,s)}{g(x,s)} \,, \quad (x, s) \in \mathbb{R}^n \times [0, T]
  \label{opt-u}
\end{align}
leads to the zero-variance importance sampling Monte Carlo estimator for the
path ensemble average in (\ref{g-fun-repeat})~\cite{entropy-var-free-energy2017}.
Based on these facts and the equality (\ref{jarzynski-g}), it is not difficult to conclude that the optimal
probability measure to sample the free energy $\Delta F$ in (\ref{free-energy-optimal-estimator})
is given by the disintegration expression
\begin{align}
  \mathbf{P}^*= \int_{\mathbb{R}^n} \mathbf{P}_{x,0}^*\, d\mu_0^*(x) \,,
\label{opt-p-decomp}
\end{align}
where $\mu_0^*$ is the probability measure on $\mathbb{R}^n$ such that 
\begin{align}
  \frac{d\mu^*_0}{dx} \propto e^{-\beta V(x,\lambda(0))} g(x,0) \,,
  \label{opt-mu0}
\end{align}
and $\mathbf{P}^*_{x,0}$ is the probability measure corresponding to the
controlled dynamics (\ref{dynamics-1-u}) starting from $x^u(0) = x$, with $u^*_s
= u^*_s(x^u(s))$ which is defined in (\ref{opt-u}) for $s \in [0,T]$. In other words, 
the importance sampling estimator (\ref{free-energy-optimal-estimator}) for the free energy $\Delta F$ will
achieve zero-variance, if we generate trajectories from
dynamics (\ref{dynamics-1-u}) with the control $u^*_s$ starting from the initial
distribution $x^u(0) \sim \mu_0^*$. 
\begin{remark}
  In the following, we make a comparison with other relevant directions in the literature.
  \begin{enumerate}
    \item
      (Optimal control protocol) In the importance sampling approach above, where the main purpose is to
      improve the numerical efficiency of free energy calculation,
      we assumed that the control protocol $\lambda(s)$ is
      fixed and the dynamics of the original nonequilibrium process
      is modified by adding an extra (additive) control force. 
      In contrast to this, the problem of minimizing either the average work
      or the average heat
      by varying the control protocols has been considered in several recent
      works in the study of thermodynamics for small systems~\cite{optimal-protocol2008, optimal-finite-time-seifert-2007,extracting-work-feedback-2011,optimal-protocols-transport-2011}.
      Motivated by these studies, it may be also interesting to optimize the control protocols in order to minimize the variance of 
      the Monte Carlo estimators. This problem is beyond the scope of the
      current paper but we would like to consider it in the future.
    \item
      (Escorted free energy simulation) The idea of further adding an extra control force to the nonequilibrium processes in
  order to improve the efficiency of free energy calculation has also been
  explored in the escorted free energy simulation
  method~\cite{pre-escort-fe-simulation2008,escorted-simulation2011}.
  In this method~\cite{pre-escort-fe-simulation2008}, the authors derived the identity (\ref{jarzynski-escorted}) for the modified
  dynamics (\ref{dynamics-1-escorted}), and suggested to apply it to compute the free energy
  difference $\Delta F$ by choosing the vector field $u$ in
  (\ref{dynamics-1-escorted}) properly (such that the ``lag'' is reduced).
  There also exists an optimal vector field, at least formally, such
  that the Monte Carlo estimator in the escorted simulation method achieves zero variance.
  Despite of these similarities, we emphasize that the importance sampling
  method in this subsection and the escorted free energy simulation method rely on different identities (of the nonequilibrium processes with extra control).
  In other words, the change of measure identity in the first line of
  (\ref{df-w-ineq}) and the identity (\ref{jarzynski-escorted}) can not be
  derived from one to the other straightforwardly.
  Furthermore, unlike the escorted free energy simulation method where the initial
  distribution is fixed, in importance sampling one has the freedom to change the initial distribution as well. 
  In particular, this is the case for the optimal change of measure, since
  $\mu_0^*$ in (\ref{opt-mu0}) is typically different from the equilibrium distribution $\mu_{\lambda(0)}$.
\item
  (Bidirectional sampling, Bennett's acceptance ratio method)
It is known in the literature~\cite{crooks-path-ensemble-pre2000,compare_free_energy_methods_2006,optimal-estimator-minh-2009,escorted-simulation2011}
      that free energy estimators based on Crooks's relation
      (\ref{crooks-eps0-2}), using trajectories of both the forward and backward processes, perform much better than estimators based on the
      Jarzynski's equality (\ref{jarzynski-repeat}), which only use trajectories of the forward process. The optimal choice of the function $\phi$ in
      (\ref{crooks-eps0-2}) is known~\cite{BENNETT1976}, given the numbers of both forward and backward trajectories.
      It is interesting to consider how one can apply the importance sampling
      idea to further improve the efficiency of estimators which use trajectories of both forward and backward processes.
      We leave this question in future study.
  \end{enumerate}
  \label{rmk-ip-and-escorted}
\end{remark}
\subsection{Cross-entropy method}
\label{subsec-ce}
From the previous subsection, we know that the probability measure
$\mathbf{P}^*$ in (\ref{opt-p}), or equivalently in (\ref{opt-p-decomp}), is optimal in the sense that the importance sampling estimator
(\ref{free-energy-optimal-estimator}) has zero-variance. However, in
practice it is often difficult to compute $\mathbf{P}^*$ or $u_s^*$. 
In this subsection, we briefly outline a numerical approach to sample the
free energy difference $\Delta F$ using the importance sampling Monte Carlo
method~\cite{ce_paper2014,ce_book}. 
The main idea is to approximate the optimal measure $\mathbf{P}^*$ within a family of parameterized probability measures
$\big\{\mathbf{P}_{\boldsymbol{\omega}}\,|\,\boldsymbol{\omega} \in \mathbb{R}^k\big\}$,
with the hope that the closer $\mathbf{P}_{\boldsymbol{\omega}}$ is to
$\mathbf{P}^*$, the more efficient the importance sampling estimator will be (in the sense that variance is small). Different from the importance sampling method studied in~\cite{path_sampling_zuckerman2004,
optimum-bias-2008}
which requires Monte Carlo sampling in path space with an acceptance-rejection
procedure, the method proposed below can be implemented at the SDE level. 

We recall that the probability measure $\mathbf{P}$ corresponds
to the trajectories of processes (\ref{dynamics-1}) and (\ref{ode-control}).
Now let $\bar{\mu}_0$ be the probability measure on $\mathbb{R}^n$, 
possibly different from $\mu_{\lambda(0)}$. Given a parameter $\boldsymbol{\omega}=(\omega_1, \omega_2, \cdots,
\omega_k)^T \in \mathbb{R}^k$, we define $\mathbf{P}_{\boldsymbol{\omega}}$
as the probability measure corresponding to the trajectories of the process 
\begin{align}
  \begin{split}
    d x(s) & =  b\big(x(s), \lambda(s)\big) ds + 
    \sigma\big(x(s), \lambda(s)\big) \Big(\sum_{l=1}^k \omega_l \phi^{(l)}\big(x(s),
    \lambda(s), s\big)\Big)\,ds +
    \sqrt{2\beta^{-1}} \sigma\big(x(s), \lambda(s)\big) \,dw(s)\,,
\end{split}
  \label{dynamics-ce-1}
\end{align}
and the control protocol (\ref{ode-control}), 
starting from $x(0) \sim \bar{\mu}_{0}$, where $\phi^{(l)} : \mathbb{R}^n
\times \mathbb{R}^m \times \mathbb{R}^+ \rightarrow
\mathbb{R}^n$, $1 \le l \le k$, are $k$ ansatz functions. Clearly, we have
$\mathbf{P}_{\boldsymbol{\omega}}=\mathbf{P}$ when
$\boldsymbol{\omega}=\boldsymbol{0}\in \mathbb{R}^k$ and
$\bar{\mu}_0=\mu_{\lambda(0)}$. As a special choice of
ansatz functions, we can take
$\phi^{(l)} =-\sigma^T\nabla V^{(l)}$, where $V^{(l)}: \mathbb{R}^n\times \mathbb{R}^m\rightarrow
\mathbb{R}$, $1 \le l \le k$, are $k$ potential functions. In this case, recalling that dynamics (\ref{dynamics-1})
can be written equivalently as (\ref{dynamics-1-q-vector}), 
we see that dynamics (\ref{dynamics-ce-1}) becomes 
\begin{align*}
  \begin{split}
    d x(s) & =  \bigg[J - a\nabla\Big(V + \sum_{l=1}^k \omega_lV^{(l)}\Big)+
    \frac{1}{\beta} \nabla\cdot a\bigg](x(s), \lambda(s))\, ds + 
     \sqrt{2\beta^{-1}} \sigma\big(x(s), \lambda(s)\big) \,dw(s)\,,
\end{split}
\end{align*}
i.e., probability measure $\mathbf{P}_{\bm{\omega}}$ corresponds to the
dynamics under the modified potential $V + \sum\limits_{l=1}^k \omega_lV^{(l)}$.

The optimal approximation of
the probability measure $\mathbf{P}^*$ within 
the set $\big\{\mathbf{P}_{\boldsymbol{\omega}}\,|\,\boldsymbol{\omega} \in \mathbb{R}^k\big\}$
is defined as the minimizer of the minimization problem 
\begin{align}
\min_{\boldsymbol{\omega} \in \mathbb{R}^k} D_{KL}\big(\mathbf{P}^*\,\|\, \mathbf{P}_{\boldsymbol{\omega}}\big)\,.
  \label{mini-omega-problem}
\end{align}
Note that, comparing to the minimization of the function $\Phi$ in (\ref{phi-cost}), which is
equivalent to minimizing $D_{KL}(\cdot\,\|\,\mathbf{P}^*)$ by (\ref{entropy-exp}), 
approximations have been introduced in (\ref{mini-omega-problem}), i.e.,
we have first
switched the order of the two arguments in $D_{KL}(\cdot\,\|\,\cdot)$ and then 
confined ourselves on a parameterized subset of probability measures with fixed
starting distribution $\bar{\mu}_0$. Using
(\ref{opt-p}), we can write the objective function in
(\ref{mini-omega-problem}) more explicitly as
\begin{align}
  D_{KL}\big(\mathbf{P}^*\,\|\, \mathbf{P}_{\boldsymbol{\omega}}\big) = 
  D_{KL}\big(\mathbf{P}^*\,\|\, \mathbf{P}\big) - e^{\beta\Delta F} \mathbf{E}\Big(e^{-\beta W}
  \ln\frac{d\mathbf{P}_{\boldsymbol{\omega}}}{d\mathbf{P}} \Big) \,,
  \label{omega-problem-objective-explicit}
\end{align}
where the parameter $\boldsymbol{\omega}$ only appears in the second term on the right hand side of the above equality. 
Applying Girsanov's theorem~\cite{oksendalSDE}, we have 
\begin{align}
  \frac{d\mathbf{P}_{\boldsymbol{\omega}}}{d\mathbf{P}} = 
\frac{d\bar{\mu}_0}{d\mu_{\lambda(0)}}\big(x(0)\big) \times 
  \exp\bigg[\frac{\beta}{2} \int_0^T \Big(\sum_{l=1}^k
    \omega_l\phi^{(l)}\Big)\cdot \sigma^{-1} \big(dx(s) - b\, ds\big) - \frac{\beta}{4} \int_0^T 
  \Big|\sum_{l=1}^k \omega_l\phi^{(l)}\Big|^2\, ds\bigg]\,,
  \label{girsanov-ce}
\end{align}
where the dependence of the functions $b, \sigma, \phi^{(l)}$ on $x(s),
\lambda(s), s$ is omitted for simplicity.
Substituting (\ref{girsanov-ce}) into equality (\ref{omega-problem-objective-explicit}), 
we can observe that the objective function in (\ref{mini-omega-problem}) is
in fact quadratic with respect to the parameter $\boldsymbol{\omega} \in \mathbb{R}^k$.
Taking derivatives, we conclude that the minimizer of (\ref{mini-omega-problem}) is determined by the linear equation
$A\boldsymbol{\omega}^* = R$, where 
\begin{align}
  \begin{split}
    A_{ll'} = \mathbf{E} \bigg[e^{-\beta W} \int_0^T \phi^{(l)}\cdot
      \phi^{(l')}\,
  ds\bigg]\,, 
  \quad R_{l} = \mathbf{E} \bigg[e^{-\beta W} \int_0^T \phi^{(l)} \cdot
\sigma^{-1} \big(dx(s) - b\, ds\big) \bigg]\,,
\end{split}
\label{a-r-coeff}
\end{align}
for $1 \le l, l' \le k$. 

In practice, we can estimate entries of $A$ and $R$ in
(\ref{a-r-coeff}) by simulating a relatively small number of
trajectories, and compute $\bm{\omega}^*$ by solving the linear equation
$A\bm{\omega}^*=R$.
After this, the free energy difference $\Delta F$ can be estimated using importance sampling
by simulating a large number of trajectories corresponding to
$\mathbf{P}_{\bm{\omega}^*}$. Also notice that, instead of computing $A$ and $R$
using the original dynamics and solving $\bm{\omega}^*$ directly, it is
helpful to solve $\bm{\omega}^*$ in an iterative manner starting from a higher
temperature (small $\beta$) or running a different dynamics (importance sampling). We refer readers to the previous
studies~\cite{ce_book,ce_paper2014} for more
algorithmic details.
\begin{remark}
  More generally, instead of keeping the starting distribution
  $\bar{\mu}_0$ fixed, we could also optimize $\bar{\mu}_0$ within a parameterized set
  of probability measures on $\mathbb{R}^n$ by solving an optimization
  problem which is similar to (\ref{mini-omega-problem}). 
  In this case, while the optimal parameter $\bm{\omega}^*$ can still be obtained
  from the same linear equation $A\bm{\omega}^*=R$, a nonlinear equation needs
  to be solved in order to get the optimal $\bar{\mu}_0$.
  We expect to develop algorithms which adaptively optimize $\bm{\omega}^*$ and $\bar{\mu}_0$ in an alternative manner.
  This will be considered in future work.
  \label{rmk-3}
\end{remark}

\textbf{Choices of ansatz functions}.
  Clearly, the efficiency of the importance sampling Monte Carlo method crucially depends on the
  choices of ansatz functions used in the cross-entropy method.
  From Jarzynski's equality (\ref{jarzynski-repeat}) and the optimal change of measure
  (\ref{opt-p}), we can expect that an importance sampling estimator will have 
  better performance if paths with smaller work $W$ (comparing to
  $\Delta F$) are sampled more frequently. Accordingly, the ansatz functions
  used in the cross-entropy method should be chosen such that the work $W$ can be
  decreased by the control forces.
  A similar idea has been used in the previous work~\cite{Hartmann2016-Nonlinearity}, where several ways of choosing ansatz functions have been proposed.  
  
  In the current situation where the work $W$ is given in (\ref{w-repeat}), 
  we can see that $W$ will be large if the potential increases along the
  movement of the parameter $\lambda$.
  Actually, this already explains the reason why a standard Monte Carlo
  simulation of fast-switching dynamics based on Jarzynski's
  equality is likely to have poor efficiency.
To elucidate this point more clearly, we consider a special situation when the expression of the
  work $W$ becomes simpler and allows us to have some insights on how to
  choose ansatz functions.
  Specifically, let $\lambda \in [0,1]$ and suppose that we are interested in the
  free energy differences corresponding to potentials $V(x,0)$ and
  $V(x,1)$, $x \in \mathbb{R}^n$. Then a simple way is to consider the linear
  interpolation~\cite{path_sampling_zuckerman2004}
  \begin{align}
    V(x,\lambda) = (1-\lambda)V(x,0) + \lambda V(x,1)\,,\quad \lambda \in
    [0,1]\,,
  \end{align}
  and the control protocol $\lambda(s) = s$ on the time interval $s \in [0,1]$.
  In this case, the expression of work in (\ref{w-repeat}) as a path
  functional becomes as simple as 
  \begin{align}
    W = \int_0^1 \Big(V(x(s), 1) - V(x(s), 0)\Big) ds\,.
    \label{work-w-special}
  \end{align}
   It is not difficult to see that paths simulated by a standard Monte Carlo method will typically have large work due to the fact that,
  starting from the Boltzmann distribution of the potential $V(x,0)$ and on
  the finite time interval $[0,1]$,
  the nonequilibrium process $x(s)$ is likely to stay within the
  region where potential $V(x,1)$ is large, in particular when the low potential regions of
  $V(x,0)$ and $V(x,1)$ do not overlap (see \cite{jarzynski-rare2006} for more detailed discussions). 
  Accordingly, the importance sampling can improve the efficiency of the
  standard Monte Carlo estimator if we place ansatz functions in a way such that, after
  optimization using the cross-entropy method, 
  transitions of the controlled dynamics (\ref{dynamics-ce-1})
  from low energy regions of $V(x,0)$ to low energy region of
  $V(x,1)$ within time $[0,1]$ become easier. Similar idea (i.e., to reduce the ``lag'') has been used to guide the choice of the vector field 
  in the escorted free energy simulation method~\cite{pre-escort-fe-simulation2008,escorted-simulation2011}.
  Readers are referred to Subsection~\ref{subsec-ex1} for numerical study of the ideas discussed above. 

\section{Jarzynski-like equality and fluctuation theorem : reaction coordinate case}
\label{sec-coordinate}
Different from the situation in Section~\ref{sec-alchemical} where the free energy
in (\ref{free-energy}) is defined as a function of the parameter $\lambda$
through the invariant measure $\mu_\lambda$ on $\mathbb{R}^n$, in this
section we assume a function $\xi : \mathbb{R}^n \rightarrow \mathbb{R}^d$ is
given and the free energy is defined as a function of $z \in \mathbb{R}^d$
through the invariant measure $\mu_z$ on the level set $\xi^{-1}(z)$. 
In the literature, such a function $\xi$ is often termed as \textit{reaction
coordinate function} or \textit{collective variable}~\cite{givon2004emd,mimick_sde,effective_dynamics,blue-moon,tony-free-energy-compuation,Maragliano2006}. 

In this context, we point out that a Jarzynski-like equality has been obtained
in the previous work~\cite{LELIEVRE2007}, and a Jarzynski-Crooks fluctuation
identity has been derived for the constrained Langevin dynamics in~\cite{Tony-constrained-langevin2012}.
In this section, following the analysis in Section~\ref{sec-alchemical}, we will 
prove a fluctuation theorem (Theorem~\ref{thm-fluct-relation-coordinate}) which is similar to Theorem~\ref{thm-fluct-relation}, and then we obtain
the Jarzynski-like equality (Theorem~\ref{thm-jarzynski-coordinate}) by
applying the fluctuation theorem. Importance sampling and
variance reduction issues will be discussed in Subsection~\ref{subsec-info-the-ce-coordinate}.
\subsection{Mathematical setup}
\label{sec-coordinate-setup}
First of all, we recall some notations as well as some results from the 
work~\cite{effective_dyn_2017,zhang2017} in order to introduce the problem under investigation. 

Let $ \xi : \mathbb{R}^n \rightarrow \mathbb{R}^d$ be a $C^2$ function with components $\xi
= (\xi_1, \xi_2, \cdots, \xi_d)^T \in \mathbb{R}^d$, where $1 \le d < n$. Given $z \in
\mbox{Im}\,\xi \subseteq \mathbb{R}^d$, which is a regular value of the map $\xi$,
we define the level set 
\begin{align}
  \Sigma_{z} =\xi^{-1}(z)=\Big\{ y \in \mathbb{R}^n \,\Big|\, \xi(y) = z \in \mathbb{R}^d\Big\}\,.
\label{submanifold-n}
\end{align}
It is known from the regular value theorem~\cite{banyaga2004lectures} that $\Sigma_{z}$ is a
smooth $(n-d)$-dimensional submanifold of $\mathbb{R}^n$.
Let $\nu_z$ denote the surface measure on $\Sigma_{z}$ which is induced from
the Euclidean metric on $\mathbb{R}^n$, and $\nabla \xi$ denote the $n \times d$ matrix whose entries are $(\nabla\xi)_{i\gamma}
=\frac{\partial\xi_\gamma}{\partial y_i}$,  $1 \le i \le n$, $1 \le \gamma \le d$.

Given a smooth function $V : \mathbb{R}^n \rightarrow \mathbb{R}$, we consider
the probability measure on the submanifold $\Sigma_z$ defined as 
\begin{align}
  d \mu_z = \frac{1}{Q(z)} e^{-\beta V} \Big[\mbox{det}\big(\nabla \xi^T \nabla \xi\big)\Big]^{-\frac{1}{2}} d\nu_z\,,
\label{mu-z}
\end{align}
where $Q(z)$ is the normalization constant. 
The probability measure $\mu_z$ arises in many
situations and plays an important role in the free energy calculation along a
reaction coordinate~\cite{blue-moon,projection_diffusion,effective_dynamics,effective_dyn_2017,tony-free-energy-compuation,zhang2017}.
The free energy for fixed $z \in
\mbox{Im}\,\xi \subseteq \mathbb{R}^d$ is
defined as 
\begin{align}
  \begin{split}
    F(z) =& -\beta^{-1} \ln Q(z) \\
    =& -\beta^{-1} \ln 
  \int_{\Sigma_z} e^{-\beta V} \Big[\mbox{det}\big(\nabla\xi^T \nabla \xi\big)\Big]^{-\frac{1}{2}} d\nu_z \\
    =& -\beta^{-1} \ln 
  \int_{\mathbb{R}^n} e^{-\beta V(y)} \delta\big(\xi(y) - z\big)\,dy \,,
  \end{split}
  \label{free-energy-coordinate}
\end{align}
where the last equality follows from the co-area formula~\cite{evans1991measure,krantz2008geometric}. 
Let $\sigma : \mathbb{R}^n \rightarrow \mathbb{R}^{n \times n}$ be an $n \times n$
matrix valued function such that the function $a(\cdot) :=
(\sigma\sigma^T)(\cdot)$ is uniformly elliptic on
$\mathbb{R}^n$. Let $\Psi =\nabla\xi^T a \nabla \xi$ be the invertible $d\times d$ matrix whose entries are
\begin{align}
  \Psi_{\gamma\gamma'} 
  =(\nabla \xi_\gamma)^T a\nabla \xi_{\gamma'} \,,\quad 1 \le \gamma, \gamma' \le
  d\,,
  \label{psi-ij}
\end{align}
where $\nabla\xi_\gamma$ is the usual gradient of the function $\xi_\gamma$. 
Let $P=\mbox{id}-a\nabla\xi\Psi^{-1}\nabla\xi^T$ be the projection matrix, with entries
\begin{align}
  P_{ij} =& \delta_{ij} -
(\Psi^{-1})_{\gamma\gamma'}
	a_{il}\partial_l\xi_\gamma\,\partial_j\xi_{\gamma'}\,,  \quad 1 \le i,j \le n\,.
\label{p-ij}
\end{align}
Notice that in the above $\delta_{ij}$ is the Kronecker delta function and 
Einstein's summation convention is used here and in the following.
From (\ref{p-ij}), we can directly verify that 
\begin{align}
  \begin{split}
    &P^2=P\,,\quad P^T\nabla \xi_\gamma = 0\,,\quad 1 \le \gamma \le d\,, \\
    &(aP^T)_{ij}=(Pa)_{ij} = a_{ij} -
    (\Psi^{-1})_{\gamma\gamma'} (a\nabla \xi_\gamma)_i (a \nabla \xi_{\gamma'})_j\,, \quad  1\le i,j \le n\, ,
  \end{split}
\label{p-i-j}
  \end{align}
  i.e., $P$ is the orthogonal projection w.r.t. the scalar product
  $\langle u, v\rangle_{a^{-1}} = u^Ta^{-1} v$, for $u, v \in \mathbb{R}^n$.

  It is shown in~\cite{zhang2017} that, starting from $y(0) \in \Sigma_z$, the process 
\begin{align}
  \begin{split}
    dy_i(s) 
= & -(Pa)_{ij} 
\frac{\partial V}{\partial y_j}\,ds + \frac{1}{\beta} \frac{\partial
    (Pa)_{ij}}{\partial y_j}\,ds + \sqrt{2\beta^{-1}}\, (P\sigma)_{ij}\,
    dw_{j}(s)\,,\quad  1\le i \le n\,,
  \end{split}
\label{dynamics-submanifold}
\end{align}
where $w(s)$ is an $n$-dimensional Brownian motion, will remain on the
submanifold $\Sigma_z$ and has a unique invariant measure $\mu_z$ which is
defined in (\ref{mu-z}). In particular, 
denoting by $\mathcal{L}^{\perp}$ the infinitesimal generator of the process (\ref{dynamics-submanifold}), i.e., 
\begin{align}
\mathcal{L}^{\perp} = 
  -(Pa)_{ij} \frac{\partial V}{\partial y_j}\frac{\partial}{\partial y_i} +
  \frac{1}{\beta} \frac{\partial (Pa)_{ij}}{\partial y_j}
  \frac{\partial}{\partial y_i} + \frac{1}{\beta} (Pa)_{ij}
\frac{\partial^2}{\partial y_i\partial y_j}\,,
  \label{generator-l-perp}
\end{align}
it is easy to verify that $\mathcal{L}^{\perp} \xi_\gamma \equiv
0$, for $1 \le \gamma \le d$. 
\subsection{Fluctuation theorem}
\label{subsec-fluct-thm-coordinate}
In order to state the fluctuation theorem, we further introduce a ``controlled'' process
as well as its time-reversed counterpart based on the process (\ref{dynamics-submanifold}). 
Specifically, we let $f = (f_1, f_2,\cdots, f_d)^T : \mathbb{R}^n \times
[0, T] \rightarrow \mathbb{R}^d$ be a bounded smooth function and consider the
process
\begin{align}
  \begin{split}
    dy_i(s) 
= & - (Pa)_{ij} 
\frac{\partial V}{\partial y_j}\,ds + \frac{1}{\beta} \frac{\partial
    (Pa)_{ij}}{\partial y_j}\,ds + (\Psi^{-1})_{\gamma\gamma'}
    (a\nabla\xi_\gamma)_i\,f_{\gamma'}\,ds + \sqrt{2\beta^{-1}}\,
    (P\sigma)_{ij}\, dw_{j}(s)\,,
  \end{split}
\label{dynamics-f}
\end{align}
for $1 \le i \le n$ on the time interval $[0,T]$.
The infinitesimal generator of the process (\ref{dynamics-f}) is given by
\begin{align}
  \mathcal{L} = \mathcal{L}^\perp + (\Psi^{-1})_{\gamma\gamma'}
  (a\nabla\xi_\gamma)_if_{\gamma'} \frac{\partial}{\partial y_i} \,,
  \label{l-coordinate}
\end{align}
where the operator $\mathcal{L}^\perp$ is defined in (\ref{generator-l-perp}), 
and a simple application of Ito's formula implies that 
\begin{align}
  d\xi(y(s)) = f(y(s), s)\,ds \,.
  \label{xi-dt-coordinate}
\end{align}
Similarly, the time-reversed process of the dynamics (\ref{dynamics-f}) on the time
interval $[0, T]$ is defined as 
\begin{align}
  \begin{split}
    dy_i^{R}(s) 
= & - (Pa)_{ij} 
\frac{\partial V}{\partial y_j}\,ds + \frac{1}{\beta} \frac{\partial
    (Pa)_{ij}}{\partial y_j}\,ds - (\Psi^{-1})_{\gamma\gamma'}
    (a\nabla\xi_\gamma)_i\,f^-_{\gamma'}\,ds + \sqrt{2\beta^{-1}}\,
    (P\sigma)_{ij}\, dw_{j}(s)\,,
  \end{split}
\label{dynamics-f-reversed}
\end{align}
where $1 \le i \le n$, $f^-_{\gamma'}(\cdot, s) = f_{\gamma'}(\cdot, T-s)$, and the infinitesimal
generator is
\begin{align}
  \mathcal{L}^R = \mathcal{L}^\perp - (\Psi^{-1})_{\gamma\gamma'}
  (a\nabla\xi_\gamma)_if^-_{\gamma'} \frac{\partial}{\partial y_i}\,.
  \label{l-reversed-coordinate}
\end{align}

Using a similar argument as in the proof of Theorem~\ref{thm-fluct-relation}, we
obtain the following fluctuation theorem which concerns the relation
between the dynamics
(\ref{dynamics-f}) and the time-reversed one (\ref{dynamics-f-reversed}).
\begin{theorem}
  Let $0 \le t' < t \le T$ and $y,y' \in \mathbb{R}^n$. For any continuous
  function $\eta \in C\big(\mathbb{R}^n \times [0,T]\big)$ with compact
  support, we have 
\begin{align}
  \begin{split}
  &e^{-\beta V(y')}\,
    \mathbf{E}^R_{y',t'}\bigg[\exp\bigg(\int_{t'}^t \eta(y^R(s),
      T-s) ds\bigg) \delta\big(y^R(t)-y\big)\,\bigg]\\
    =&e^{-\beta V(y)}\,\mathbf{E}_{y,T-t}\bigg[e^{-\beta
    \mathcal{W}} \exp\bigg(\int_{T-t}^{T-t'} \eta(y(s), s) ds\bigg) 
\delta\big(y(T-t')-y'\big)\bigg]\,,
  \end{split}
  \label{fluct-relation-coordinate}
\end{align}
where 
\begin{align}
  \mathcal{W} = \int_{T-t}^{T-t'} \Big[
    (\Psi^{-1})_{\gamma\gamma'} (a\nabla \xi_\gamma)_i f_{\gamma'} \frac{\partial
    V}{\partial y_i} - \frac{1}{\beta} \frac{\partial}{\partial y_i}
    \Big((\Psi^{-1})_{\gamma\gamma'} (a\nabla \xi_\gamma)_i
    f_{\gamma'}\Big)\Big] ds\,,
  \label{w-coordinate}
\end{align}
  $y^R(\cdot)$, $y(\cdot)$ satisfy the dynamics (\ref{dynamics-f-reversed}) and (\ref{dynamics-f}), respectively.
$\mathbf{E}^R_{y',t'}$ is the conditional expectation with respect to the path
ensemble of the dynamics (\ref{dynamics-f-reversed}) starting from $y^R(t') = y'$ at time $t'$. 
And $\mathbf{E}_{y,T-t}$ is the conditional expectation with respect to the dynamics
  (\ref{dynamics-f}) starting from $y(T-t) = y$ at time $T-t$.
\label{thm-fluct-relation-coordinate}
\end{theorem}
The proof of Theorem~\ref{thm-fluct-relation-coordinate} can be found in Appendix~\ref{app-4}.
Similar to Theorem~\ref{thm-fluct-relation}, the
identity~(\ref{fluct-relation-coordinate}) should be
understood in the sense of distributions. We refer to Remark~\ref{rmk-delta} for further discussions. 
\subsection{Jarzynski-like equality}
\label{subsec-jarzynski-like-coordinate}
In this subsection, we assume that there is a function $\widetilde{f} = (\widetilde{f}_1,
\widetilde{f}_2, \cdots, \widetilde{f}_d)^T : \mathbb{R}^d \times [0, T]
\rightarrow \mathbb{R}^d$, such that 
\begin{align}
f(y,s) = \widetilde{f}(\xi(y), s), \quad \forall (y, s) \in \mathbb{R}^n
  \times [0, T]\,. 
  \label{f-f-tilde}
\end{align} 
Fix $t\in [0, T]$ and suppose that both the ODE 
\begin{align}
  \dot{\zeta}(s\,;z) = \widetilde{f}(\zeta(s\,;z), s), \quad s \in [0, t]\,,
  \label{zt-ode}
\end{align}
starting from $\zeta(0\,;z)=z$, and the ODE
\begin{align}
  \dot{\zeta}^R(s\,;z) = -\widetilde{f}(\zeta^R(s\,;z), T-s), \quad s \in [T-t, T]\,,
  \label{zt-ode-r}
\end{align}
starting from $\zeta^R(T-t\,;z) = z$, have a unique solution for any $z \in \mathbb{R}^d$. 
Under this assumption, it is not difficult to conclude that
\begin{align*}
  \zeta^R(s\,;\zeta(t\,;z)) = \zeta(T-s\,;z)\,,\quad \zeta(T-s\,;
  \zeta^R(T\,;z)) = \zeta^R(s\,;z)\,,\quad s \in [T-t,T]\,,
\end{align*}
which in turn implies that the map $\zeta^R(T\,;\cdot) : \mathbb{R}^d
\rightarrow \mathbb{R}^d$ is invertible and its
inverse is given by $\zeta(t\,;\cdot)$.

Consider the process $y(s)$ in (\ref{dynamics-f}) on the time interval $[0,t]$, and 
process $y^R(s)$ in (\ref{dynamics-f-reversed}) on the time interval $[T-t, T]$, respectively.
Assume that $\xi(y(0))=z$ and $\xi(y^R(T-t)) = z'$, where $z, z'\in \mathbb{R}^d$. Similar to (\ref{xi-dt-coordinate}), we can obtain
\begin{align*}
  d\xi(y(s)) = \widetilde{f}\big(\xi(y(s)), s\big)\, ds,\qquad d\xi(y^R(s)) =
  -\widetilde{f}\big(\xi(y^R(s)), T-s\big)\, ds\,,
\end{align*}
which imply that 
\begin{align}
  \xi(y(s)) = \zeta(s\,;z)\,,\quad \xi(y^R(T-s)) = \zeta^R(T-s\,;z')\,, \qquad \forall\,s \in [0,t]\,.
  \label{xi-yt-coordinate}
\end{align}

Applying Theorem~\ref{thm-fluct-relation-coordinate}, we can obtain the following
Jarzynski-like equality for the free energy difference in the reaction coordinate case.
\begin{theorem}[Jarzynski-like equality]
  Let $y(s)$ be the dynamics in (\ref{dynamics-f}) with the function $f$ in
  (\ref{f-f-tilde}) and $z(s)$ solve the ODE
  (\ref{zt-ode}). For any smooth and bounded test function $\varphi : \mathbb{R}^n \rightarrow
  \mathbb{R}$ and $t \in [0, T]$, we have 
\begin{align}
  \mathbf{E}_{z(0),0} \Big[\varphi(y(t))\,e^{-\beta W(t)}\Big] = e^{-\beta \big(F(z(t)) - F(z(0))\big)} \int_{\Sigma_{z(t)}} \varphi\, d\mu_{z(t)} 
  \,,
  \label{generalized-jarzynski-coordinate-varphi}
\end{align}
   where $F(\cdot)$ is the free energy in (\ref{free-energy-coordinate}) and
  $W(t)$ is defined as
  \begin{align}
    W(t) = \int_{0}^{t} \Big[
    (\Psi^{-1})_{\gamma\gamma'} (a\nabla \xi_\gamma)_i  \frac{\partial
    V}{\partial y_i} - \frac{1}{\beta} \frac{\partial}{\partial y_i} \Big((\Psi^{-1})_{\gamma\gamma'} (a\nabla \xi_\gamma)_i
  \Big)\Big]\, \dot{z}_{\gamma'}(s) \, ds\,.
  \label{w-coordinate-jarzynski}
\end{align}
  $\mathbf{E}_{z(0),0}$ denotes the conditional expectation with respect to the dynamics
$y(s)$, starting from the initial distribution $y(0) \sim \mu_{z(0)}$ on
  $\Sigma_{z(0)}$.
  In particular, taking $\varphi\equiv 1$, we have 
\begin{align}
  \mathbf{E}_{z(0),0}\Big[e^{-\beta W(t)}\Big] = e^{-\beta \big(F(z(t)) - F(z(0))\big)} \,.
  \label{generalized-jarzynski-coordinate}
\end{align}
\label{thm-jarzynski-coordinate}
\end{theorem}
\begin{proof}
  Let $\mbox{div}_z$ denote the divergence operator with respect to $z \in
  \mathbb{R}^d$. Notice that from the definitions of $\Psi$ in (\ref{psi-ij})
  and the function $f$ in (\ref{f-f-tilde}) we can compute 
  \begin{align*}
 (\Psi^{-1})_{\gamma\gamma'} (a\nabla \xi_\gamma)_i
    \frac{\partial f_{\gamma'}}{\partial y_i} = 
 (\Psi^{-1})_{\gamma\gamma'} (a\nabla \xi_\gamma)_i
    \frac{\partial \widetilde{f}_{\gamma'}}{\partial z_j} \frac{\partial
    \xi_j}{\partial y_i} = (\mbox{div}_z \widetilde{f}\,)(\xi(y), s)\,.
  \end{align*}
  Choosing $\eta(y,s) =-(\mbox{div}_z\widetilde{f}\,)(\xi(y),s)$ in
  the equality
  (\ref{fluct-relation-coordinate}) of Theorem~\ref{thm-fluct-relation-coordinate},
  we obtain
  \begin{align}
  \begin{split}
  &e^{-\beta V(y')}\,
    \mathbf{E}^R_{y',T-t}\bigg[\exp\bigg(-\int_{T-t}^T
    (\mbox{div}_z\widetilde{f}\,)\big(\xi(y^R(s)),T-s\big)ds\bigg)\delta\big(y^R(T)-y\big)\,\bigg]\\
    =&e^{-\beta V(y)}\,\mathbf{E}_{y,0}\Big[e^{-\beta W(t)}
\delta\big(y(t)-y'\big)\Big]\,.
  \end{split}
    \label{thm-jarzynski-coordinate-eqn1}
  \end{align}
  Let $\tau>0$ and multiply both sides of
  (\ref{thm-jarzynski-coordinate-eqn1}) by $\varphi(y')e^{-\beta\frac{|\xi(y)-z(0)|^2}{\tau}}$. 
Integrating with respect to $y,y'$, yields
  \begin{align}
  \begin{split}
    &\int_{\mathbb{R}^n} e^{-\beta \big(V(y)
    +\frac{|\zeta^R(T\,;\,\xi(y))-z(0)|^2}{\tau}\big)}\,
    \exp\bigg(-\int_{T-t}^T(\mbox{div}_z\widetilde{f}\,)\big(\zeta^R(s\,;\xi(y)),T-s\big)ds\bigg)
    \,\varphi(y)\, dy\\
    =&\int_{\mathbb{R}^n}\,e^{-\beta \big(V(y) + \frac{|\xi(y)-z(0)|^2}{\tau}\big)}\,\mathbf{E}_{y,0}\Big[e^{-\beta W(t)} \varphi(y(t))
    \Big]\, dy\,.
  \end{split}
    \label{thm-jarzynski-coordinate-eqn2-simplify}
  \end{align}
  Notice that, on the left hand side above, we have used the fact that 
$\xi(y^R(s))$ under the conditional expectation is deterministic 
and is given by (\ref{xi-yt-coordinate}). 

   We can rewrite the left hand side of (\ref{thm-jarzynski-coordinate-eqn2-simplify}) by applying the co-area formula
  \begin{align}
    \begin{split}
      &\int_{\mathbb{R}^n} e^{-\beta \big(V(y)+ \frac{|\zeta^R(T\,;\,\xi(y))-z(0)|^2}{\tau}\big)}\,
    \exp\bigg(-\int_{T-t}^T(\mbox{div}_z\widetilde{f}\,)\big(\zeta^R(s\,;\xi(y)),T-s\big)ds\bigg) \,\varphi(y)\, dy\\
=&\int_{\mathbb{R}^d} 
      e^{-\beta\frac{|z'-z(0)|^2}{\tau}} 
      \bigg[
\int_{\{y\,|\,\zeta^R(T\,;\,\xi(y))=z'\}} e^{-\beta V(y)}\,
 \,\varphi(y)\, 
    \exp\bigg(-\int_{T-t}^T(\mbox{div}_z\widetilde{f}\,)\big(\zeta^R(s\,;\xi(y)),T-s\big)ds\bigg)\\
      &\hspace{3cm} \times \Big[\det\Big(\big(\nabla
      \zeta^R(T\,;\xi(y))\big)^T\nabla\zeta^R(T\,;\xi(y))\Big)\Big]^{-\frac{1}{2}}
      \nu^R_{z'}(dy)\bigg]\,dz'\,,
    \end{split}
    \label{thm-jarzynski-coordinate-eqn2-lhs-coarea}
  \end{align}
  where $\nu^R_{z'}$ is the volume measure on the level set
  $\big\{y\in\mathbb{R}^n\,|\,\zeta^R(T\,;\xi(y))=z'\big\}$, 
  $\nabla\zeta^R(s\,;\xi(y))$ denotes the $n\times d$ matrix with
  components $\big(\nabla\zeta^R(s\,;\xi(y))\big)_{i\gamma}=\frac{\partial
  \zeta^R_\gamma(s\,;\,\xi(y))}{\partial y_i}$,
  for $s \in [T-t,T]$, $1 \le \gamma \le d$ and $1\le i \le n$. 
  
To simplify the above expressions, let
  $\nabla_z\zeta^R(s\,;z)$ denote the 
  $d\times d$ matrix with components $(\nabla_z\zeta^R(s\,;z))_{ij} =
  \frac{\partial\zeta^R_i(s\,;\,z)}{\partial z_j}$ for $1 \le i, j \le d$,
  i.e., the differentiations with
  respect to the initial value at time $T-t$. Furthermore, since $\zeta^R(T\,;\cdot)$ is
  invertible, we can deduce that $\zeta^R(s\,;\cdot)$ is invertible for all
  $s\in [T-t,T]$, which then implies that the matrix $\nabla_z \zeta^R(s\,;z)$ has full rank for
  $s\in [T-t,T]$. Applying chain rule, we have $\nabla \zeta^R(s\,;\xi(y)) =
  \nabla\xi\nabla_z\zeta^R(s\,;\xi(y))$ 
  and therefore 
  \begin{align*}
\Big[
  \det\Big(\big(\nabla \zeta^R(T\,;\xi(y))\big)^T\nabla\zeta^R(T\,;\xi(y))\Big)\Big]^{-\frac{1}{2}} =
\Big[\det \Big(\nabla_z\zeta^R(T\,;\xi(y))\Big)\Big]^{-1}
    \Big[\det\big(\nabla\xi^T\nabla\xi)(y)\Big]^{-\frac{1}{2}}\,.
  \end{align*}
  Combining the above identity, the equation (\ref{thm-jarzynski-coordinate-eqn2-lhs-coarea}), and applying Lemma~\ref{lemma-det} below,
  we know that equation (\ref{thm-jarzynski-coordinate-eqn2-simplify}) can be
  simplified as 
  \begin{align}
    \begin{split}
      &\frac{1}{Z_\tau}\int_{\mathbb{R}^n}\,e^{-\beta
      \big(V(y) + \frac{|\xi(y)-z(0)|^2}{\tau}\big)}\,\mathbf{E}_{y,0}\Big[e^{-\beta W(t)} \varphi(y(t))
      \Big]\,  dy\\
      =&\frac{\Big(\frac{\pi\tau}{\beta}\Big)^{\frac{d}{2}}}{Z_\tau}\Big(\frac{\beta}{\pi\tau}\Big)^{\frac{d}{2}}\int_{\mathbb{R}^d} 
      e^{-\beta\frac{|z'-z(0)|^2}{\tau}} 
      \bigg[
\int_{\{y\,|\,\zeta^R(T\,;\,\xi(y))=z'\}} e^{-\beta V(y)}\,
 \,\varphi(y)\Big[\det\big(\nabla \xi^T\nabla\xi)\Big]^{-\frac{1}{2}}
      \nu^R_{z'}(dy)\bigg]\,dz'\,,
    \end{split}
  \end{align}
  where $Z_\tau=\int_{\mathbb{R}^n} e^{-\beta
  \big(V(y)+\frac{|\xi(y)-z(0)|^2}{\tau}\big)} dy$ is the normalization
  constant.
  Letting $\tau\rightarrow 0$ and applying~\cite[Proposition $3$]{zhang2017}, we obtain
  \begin{align}
    \begin{split}
      &\int_{\Sigma_{z(0)}}\,\mathbf{E}_{y,0}\Big[e^{-\beta W(t)} \varphi(y(t)) \Big]\, \mu_{z(0)}(dy)\\
      =&  \frac{1}{Q(z(0))} \int_{\big\{y\,\big|\,\zeta^R(T\,;\,\xi(y))=z(0)\big\}}  
      e^{-\beta V(y)}\, \,\varphi(y)\Big[\det\big(\nabla \xi^T\nabla\xi)\Big]^{-\frac{1}{2}}
      \nu^R_{z(0)}(dy)\,,
    \end{split}
    \label{tau-limit-equality}
  \end{align}
  where $Q(\cdot)$ is the normalization constant in (\ref{mu-z}).
  Since the inverse of the map $\zeta^R(T\,;\cdot)$ is $\zeta(t\,;\cdot)$, we
  know $$\big\{y\in\mathbb{R}^n\,\big|\, \zeta^R(T\,;\xi(y))=z(0)\big\}=
  \big\{y\in\mathbb{R}^n\,\big|\,\xi(y)=\zeta(t\,;z(0)) =
  z(t)\big\}=\Sigma_{z(t)}\,,$$
  and therefore (\ref{tau-limit-equality}) becomes 
  \begin{align}
    \begin{split}
      \int_{\Sigma_{z(0)}}\,\mathbf{E}_{y,0}\Big[e^{-\beta W(t)} \varphi(y(t)) \Big]\, \mu_{z(0)}(dy)
      = \frac{Q(z(t))}{Q(z(0))} 
      \int_{\Sigma_{z(t)}} \varphi(y) \mu_{z(t)}(dy)\,,
    \end{split}
  \end{align}
  which is equivalent to the identity
  (\ref{generalized-jarzynski-coordinate-varphi}).
\end{proof}
 We have used the following result in the above proof.
\begin{lemma}
  Let $\zeta^R(s\,;z)$ be the solution of the ODE (\ref{zt-ode-r})
  for $s\in [T-t,T]$, starting from $z\in\mathbb{R}^d$ at time $s=T-t$. 
  $\nabla_z\zeta^R(s\,;z)$ denotes the $d\times d$ matrix where
  $(\nabla_z\zeta^R(s\,;z))_{ij} = \frac{\partial\zeta^R_i(s\,;\,z)}{\partial z_j}$
  for $1 \le i,j \le d$ and $T-t \le s \le T$.
  Suppose that $\nabla_z\zeta^R(s\,;z)$ is invertible for $T-t \le s \le T$, then we have 
  \begin{align}
    \det\Big(\nabla_z\zeta^R(s\,;z)\Big) = e^{-\int_{T-t}^s
    (\mbox{\normalfont{div}}_z\widetilde{f}\,)(\zeta^R(s'\,;\,z),T-s')\,ds'}\,,\quad s \in [T-t,T]\,.
    \label{lemma-det-formula}
  \end{align}
  \label{lemma-det}
\end{lemma}
\begin{proof}
  Differentiating both sides of the ODE (\ref{zt-ode-r}) with respect to $z$, we obtain
  the matrix equation
  \begin{align}
    \frac{d\big(\nabla_z\zeta^R(s\,;z)\big)}{ds} = - \nabla_z\zeta^R(s\,;z)\,\nabla_z
    \widetilde{f}(\zeta^R(s\,;z),T-s)\,,\quad s \in [T-t,
    T]\,,
  \end{align}
  with the initial condition $\nabla_z\zeta^R(T-t\,;z) = \mbox{id}$. Applying Jacobi's formula, we
  know that the determinant of $\nabla_z\zeta^R(s\,;z)$ satisfies 
  \begin{align*}
    & \frac{d\big[\det\big(\nabla_z \zeta^R(s\,;z)\big)\big]}{ds}\\
    =&
    \det\big(\nabla_z\zeta^R(s\,;z)\big)\,\mbox{tr}\bigg(\Big(\nabla_z\zeta^R(s\,;z)\Big)^{-1}
    \frac{d\big(\nabla_z\zeta^R(s\,;z)\big)}{ds}\bigg) \\
    =&
    -\det\big(\nabla_z\zeta^R(s\,;z)\big)\,\mbox{tr}\Big(
     \nabla_z \widetilde{f}(\zeta^R(s\,;z),T-s) \Big) \\
     =&
     -\det\big(\nabla_z\zeta^R(s\,;z)\big)\,\big(\mbox{div}_z\widetilde{f}\,\big)(\zeta^R(s\,;z),T-s)\,.
  \end{align*}
  The expression (\ref{lemma-det-formula}) is obtained by integrating the above equation. 
\end{proof}
\begin{remark}
  \begin{enumerate}
    \item
      In the special case when the reaction coordinate $\xi \in \mathbb{R}$ is scalar, matrix
$a=\sigma=\mbox{id}$, we have $\Psi = |\nabla \xi|^2$ and it can be checked that the work (\ref{w-coordinate-jarzynski}) becomes 
\begin{align}
  \begin{split}
  W(t) =& \int_{0}^{t} \bigg[
      \frac{\nabla \xi}{|\nabla \xi|^2} \cdot \nabla V 
      - \frac{1}{\beta} \mbox{div}\Big(\frac{\nabla
    \xi}{|\nabla\xi|^2}\Big)\bigg] \dot{z}(s)\,ds \\
    =&\int_{0}^{t} 
    \frac{\nabla \xi}{|\nabla \xi|^2} \cdot \Big[\nabla \Big(V  + \frac{1}{\beta} \ln
    |\nabla \xi|\Big) + \frac{1}{\beta} H\Big]\,\dot{z}(s)\,ds \,,
  \end{split}
  \label{w-coordinate-jarzynski-special}
\end{align}
where $H=-\mbox{div}\Big(\frac{\nabla \xi}{|\nabla
      \xi|}\Big)\frac{\nabla\xi}{|\nabla \xi|}$ is the mean curvature vector (field) of the surface $\Sigma_{z}$~\cite{LELIEVRE2007}.

  Notice that the free energy (\ref{free-energy-coordinate}) is different from the one considered in~\cite{LELIEVRE2007}.
      In fact, from the second expression in
      (\ref{w-coordinate-jarzynski-special}), we see that
      Theorem~\ref{thm-jarzynski-coordinate} is identical to the Feynman-Kac
  fluctuation equality Theorem of~\cite{LELIEVRE2007} for the potential $V +
  \frac{1}{2\beta} \ln (\mbox{\textnormal{det}}\,\Psi)$.
\item
  As in the alchemical transition case, one can also study the escorted
      dynamics and Crooks's relations in the reaction coordinate case. For
      simplicity, we will omit the discussions on the escorted dynamics and only
      briefly summarize the Crooks's relations. In fact, by modifying the proof of 
      Theorem~\ref{thm-jarzynski-coordinate}, we can show that 
      \begin{align}
	\frac{\mathbf{E}(e^{-\beta W} \mathcal{G})}{\mathbf{E}^R(\mathcal{G}^R)} = e^{-\beta \Delta F(T)}\,,
      \end{align}
      for any bounded smooth function $\mathcal{G}$ on the path space,
      where $W=W(T)$ is the work in (\ref{w-coordinate-jarzynski}), $\mathcal{G}^R(y(\cdot)) = \mathcal{G}(y(T-\cdot))$ for any path
      $y(\cdot)$, $\mathbf{E}$ and $\mathbf{E}^R$ are the expectation with
      respect to the process $y(\cdot)$ in (\ref{dynamics-f}) starting from
      $y(0)\sim \mu_{z(0)}$ on $\Sigma_{z(0)}$, and the expectation with respect to the process
      $y^R(\cdot)$ in (\ref{dynamics-f-reversed}) starting from $y^R(0) \sim
      \mu_{z(T)}$ on $\Sigma_{z(T)}$, respectively. In particular, this implies
  \begin{align}
    \frac{\mathbf{E} \big(e^{-\beta
    W}\phi(W)\big)}{\mathbf{E}^{R}(\phi(-W^R))}
    = e^{-\beta \Delta F(T)}\,,\quad \forall~ \phi \in C_b(\mathbb{R})\,,
    \label{crooks-eps0-2-coordinate}
  \end{align}
      where $W^R$ is the work for the time-reversed process $y^R(\cdot)$ in (\ref{dynamics-f-reversed}).
      We refer to Remark~\ref{rmk-crook-jarzynski-eps0} for comparisons. 
\item
  Similarly as in the alchemical transition case, by considering the Jarzynski-like
      equality~(\ref{generalized-jarzynski-coordinate}) for the dynamics 
\begin{align}
  \begin{split}
    dy_i(s) 
    = & - \frac{1}{\tau} (Pa)_{ij} 
    \frac{\partial V}{\partial y_j}\,ds + \frac{1}{\beta\tau} \frac{\partial
    (Pa)_{ij}}{\partial y_j}\,ds + (\Psi^{-1})_{\gamma\gamma'}
    (a\nabla\xi_\gamma)_i\,f_{\gamma'}\,ds \\
    &+ \sqrt{\frac{2\beta^{-1}}{\tau}}\, (P\sigma)_{ij}\, dw_{j}(s)\,,
  \end{split}
\label{dynamics-f-tau}
\end{align}
as $\tau\rightarrow 0$, we can recover the thermodynamic integration identity
      in the reaction coordinate case. See Appendix~\ref{app-1} and \ref{app-2} for details.
  \end{enumerate}
  \label{rmk-2}
\end{remark}
\subsection{Information-theoretic formulation and numerical considerations}
\label{subsec-info-the-ce-coordinate}
In this subsection, we study the information-theoretic formulation of the
Jarzynski-like equality (\ref{generalized-jarzynski-coordinate}) in the
reaction coordinate setting. Numerical issues related to computing free
energy differences will be discussed as well. Since the
analysis is similar to Subsection~\ref{subsec-is} and
Subsection~\ref{subsec-ce}, the discussion in this subsection will be brief
and mainly focus on the changes.

First of all, let $\mathbf{P}$, $\mathbf{E}$ denote the probability measure and the expectation of the path ensemble corresponding to
the dynamics (\ref{dynamics-f}) starting from $y(0)\sim \mu_{z(0)}$, with the
function $f$ given in~(\ref{f-f-tilde}). 
We can rewrite the equality (\ref{generalized-jarzynski-coordinate})
as 
\begin{align}
  \Delta F = - \beta^{-1} \ln \mathbf{E}\,\big( e^{-\beta W}\big)\,,
  \label{df-jarzynski-coordinate}
\end{align}
where $\Delta F = F(z(T)) - F(z(0))$ is the free energy difference and
$W=W(T)$ is defined in (\ref{w-coordinate-jarzynski}).
Let $\overline{\mathbf{P}}$ be another probability measure on the path space which is
equivalent to $\mathbf{P}$ and $\overline{\mathbf{E}}$ denote the corresponding expectation.  
Applying a change of measure in (\ref{df-jarzynski-coordinate}), we have  
\begin{align}
  \Delta F = - \beta^{-1} \ln \overline{\mathbf{E}}\,\Big( e^{-\beta
  W}\frac{d\mathbf{P}}{d\overline{\mathbf{P}}}\Big)\,.
  \label{df-jarzynski-coordinate-change-measure}
\end{align}
Following the same argument in Subsection~\ref{subsec-is}, we can deduce exactly the
same inequality (\ref{df-w-ineq}), as well as the expression for the optimal measure
$\mathbf{P}^*$, which is characterized by (\ref{opt-p}), such that the Monte Carlo estimator based on
(\ref{free-energy-optimal-estimator}) will achieve zero variance.
The derivations (\ref{phi-cost}), (\ref{variation-form}), (\ref{entropy-exp})
in Subsection~\ref{subsec-is} carry over to the current setting as well. 

On the other hand,  since the trajectories of
the dynamics (\ref{dynamics-f}) satisfy $\xi(y(t)) = z(t)$
for $t \in [0,T]$, it is important to notice that the probability measure $\mathbf{P}$ concentrates on the set of paths 
\begin{align}
  \Big\{y(\cdot)\,\Big|\, y(\cdot) \in C([0,T], \mathbb{R}^n), ~ y(t) \in
  \Sigma_{z(t)}, ~0 \le t \le T\Big\}\,.
  \label{path-subset-coordinate}
\end{align}
Accordingly, the probability measure $\overline{\mathbf{P}}$ used to perform
the change of measure in (\ref{df-jarzynski-coordinate-change-measure}) should
also concentrate on the set (\ref{path-subset-coordinate}) in order to assure that
it is equivalent to $\mathbf{P}$.

The optimal measure $\mathbf{P}^*$ can be characterized more transparently by
considering the HJB equation. Specifically, define 
\begin{align}
  g(y,t) =
\mathbf{E}\Big(e^{-\beta W_{(t,T)}}~\Big|~y(t) = y\Big)\,,\quad  \forall\,y \in \Sigma_{z(t)}\,,
\end{align}
where $y(\cdot)$ satisfies (\ref{dynamics-f}) and $W_{(t,T)}$ is similarly defined as in (\ref{w-coordinate-jarzynski})
except that the integration is from $t$ to $T$. It follows from the Feynman-Kac
formula that $g$ satisfies 
\begin{align}
  \begin{split}
  &\partial_t g + \mathcal{L} g -\beta \Big[(\Psi^{-1})_{\gamma\gamma'} (a\nabla \xi_\gamma)_i  \frac{\partial
    V}{\partial y_i} - \frac{1}{\beta} \frac{\partial}{\partial y_i} \Big((\Psi^{-1})_{\gamma\gamma'} (a\nabla \xi_\gamma)_i
  \Big)\Big]\, f_{\gamma'} g = 0 \,, \\
  &g(\cdot, T) = 1\,.
  \end{split}
\end{align}
where $\mathcal{L}$ is the infinitesimal generator defined in (\ref{l-coordinate}) for the process $y(\cdot)$.
And a simple calculation shows that $U=-\beta^{-1}\ln g$ satisfies the HJB
equation
\begin{align}
  \begin{split}
  & \partial_t U + \min_{c \in \mathbb{R}^n} \Big\{\mathcal{L}U + (P\sigma
  c)\cdot \nabla U + \frac{|c|^2}{4} \\
  &\hspace{2.2cm}+ \Big[(\Psi^{-1})_{\gamma\gamma'} (a\nabla \xi_\gamma)_i  \frac{\partial
    V}{\partial y_i} - \frac{1}{\beta} \frac{\partial}{\partial y_i} \Big((\Psi^{-1})_{\gamma\gamma'} (a\nabla \xi_\gamma)_i
  \Big)\Big]\, f_{\gamma'}\Big\} = 0\,, \\
  & U(\cdot, T) = 0\,,
  \end{split}
\end{align}
from which we conclude that the optimally controlled dynamics satisfies
\begin{align}
  \begin{split}
    dy_i(s) 
= & - (Pa)_{ij} 
\frac{\partial V}{\partial y_j}\,ds + \frac{1}{\beta} \frac{\partial
    (Pa)_{ij}}{\partial y_j}\,ds + (\Psi^{-1})_{\gamma\gamma'}
    (a\nabla\xi_\gamma)_i\,f_{\gamma'}\,ds \\
    & + \big[P\sigma u^*_s(y(s))\big]_i\,ds + \sqrt{2\beta^{-1}}\, (P\sigma)_{ij}\, dw_{j}(s)\,, \quad 1 \le i \le
    n\,,
  \end{split}
  \label{dynamics-f-optimal-controlled}
\end{align}
where the optimal feedback control $u^*_s(y) = -2(P\sigma)^T\nabla U$, starting
from the distribution $\mu_0^*$ which is determined by
$\frac{d\mu_0^*}{d\mu_{z(0)}} \propto g(\cdot, 0)$.

\textbf{Cross-entropy method.}
In the following, we briefly discuss the cross-entropy method following Subsection~\ref{subsec-ce}.
Consider a family of parameterized probability measures
$\{\mathbf{P}_{\bm{\omega}}\,|\,\bm{\omega} \in \mathbb{R}^k\}$, where, for given
$\bm{\omega} = (\omega_1, \omega_2, \cdots, \omega_k)^T \in \mathbb{R}^k$,
$\mathbf{P}_{\bm{\omega}}$ is the probability measure of paths
corresponding to the dynamics 
\begin{align}
  \begin{split}
    dy_i(s) 
= & - (Pa)_{ij} 
\frac{\partial V}{\partial y_j}\,ds + \frac{1}{\beta} \frac{\partial
    (Pa)_{ij}}{\partial y_j}\,ds + (\Psi^{-1})_{\gamma\gamma'}
    (a\nabla\xi_\gamma)_i\,f_{\gamma'}\,ds \\
    & + (P\sigma)_{ij} \Big(\sum_{l=1}^k \omega_l \phi^{(l)}_j\Big) ds +
    \sqrt{2\beta^{-1}}\, (P\sigma)_{ij}\, dw_{j}(s)\,, \quad 1 \le i \le
    n\,,
  \end{split}
  \label{dynamics-f-omega}
\end{align}
where $\phi^{(l)} = (\phi^{(l)}_1, \phi^{(l)}_2, \cdots, \phi^{(l)}_n)^T :
\mathbb{R}^n \times [0, T]\rightarrow \mathbb{R}^n$ are $k$ ansatz 
functions, $1 \le l \le k$.
As a special choice, we consider $\phi^{(l)}=-\sigma^T\nabla V^{(l)}$ where $V^{(l)} :
\mathbb{R}^n \rightarrow \mathbb{R}$, $1 \le l \le k$,  are smooth and
linearly independent potential functions, by which (\ref{dynamics-f-omega}) becomes 
\begin{align}
  \begin{split}
    dy_i(s) 
= & - (Pa)_{ij} 
    \frac{\partial \big(V + \sum_{l=1}^k \omega_l V^{(l)}\big)}{\partial y_j}\,ds + \frac{1}{\beta} \frac{\partial
    (Pa)_{ij}}{\partial y_j}\,ds \\
    & + (\Psi^{-1})_{\gamma\gamma'} (a\nabla\xi_\gamma)_i\,f_{\gamma'}\,ds + \sqrt{2\beta^{-1}}\, (P\sigma)_{ij}\,
    dw_{j}(s)\,, \quad 1 \le i \le n\,,
  \end{split}
  \label{dynamics-f-omega-vl}
\end{align}
i.e., paths are sampled with the modified potential function $V +
\sum\limits_{l=1}^k \omega_l V^{(l)}$.

Applying Ito's formula as in (\ref{xi-dt-coordinate}), we can verify that trajectories of
the dynamics (\ref{dynamics-f-omega}), starting from $y(0) \in \Sigma_{z(0)}$, satisfy $\xi(y(t)) = z(t)$ for $t \in [0,
T]$ as well. Therefore, the probability measures $\mathbf{P}_{\bm{\omega}}$
indeed concentrate on the set (\ref{path-subset-coordinate}). Applying Girsanov's theorem, we obtain 
\begin{align}
  \frac{d\mathbf{P}_{\boldsymbol{\omega}}}{d\mathbf{P}} = 
  \exp\bigg[\sqrt{\frac{\beta}{2}} \int_0^T \Big(\sum_{l=1}^k
  \omega_l\phi^{(l)}\Big)\cdot dw(s)  - \frac{\beta}{4} \int_0^T 
  \Big|\sum_{l=1}^k \omega_l\phi^{(l)}\Big|^2\, ds\bigg]\,,
  \label{girsanov-ce-coordinate}
\end{align}
where $w(s)$ is the Brownian motion in the original dynamics
(\ref{dynamics-f}) (i.e., under the probability measure $\mathbf{P}$). 
Following the same argument as in Subsection~\ref{subsec-ce}, we know that the
minimizer of the optimization problem (\ref{mini-omega-problem}) is given by
the unique solution of the linear equation $A\bm{\omega}^*=R$, where 
\begin{align}
  \begin{split}
    A_{ll'} = \mathbf{E} \bigg(e^{-\beta W} \int_0^T \phi^{(l)}\cdot \phi^{(l')}\,
  ds\bigg)\,, 
    \quad R_{l} = \sqrt{2\beta^{-1}}\mathbf{E} \bigg[e^{-\beta W} \int_0^T \phi^{(l)} \cdot dw(s) \bigg]\,,
\end{split}
\label{a-r-coeff-coordinate}
\end{align}
for $1 \le l, l' \le k$.

\textbf{Variance reduction by increasing mixing.}
In practice, however, due to the complicate expressions of work $W$ in (\ref{w-coordinate-jarzynski}) or
(\ref{w-coordinate-jarzynski-special}), it becomes difficult to have an
intuitive idea to guide the choices
of ansatz functions, which play a crucial role in the cross-entropy method
above. In the following, we briefly discuss another idea that can be explored
in order to reduce the variance in the free energy calculation based on
Jarzynski-like identity. 

Different from the importance sampling method which improves the efficiency of
Monte Carlo method by 
increasing the sampling frequency of paths with small work, the idea 
here, which is inspired by the analysis in Appendix~\ref{app-1} and
Appendix~\ref{app-2}, is to compute free energy differences based on trajectories of the dynamics
(\ref{dynamics-f-tau}) with a small $\tau$ (similar idea has also been
investigated in~\cite{efficient-free-energy-calculation-2000,fast-growth-method-jarzynski2001}).  
The observation is that the standard Monte Carlo estimator based on
Jarzynski-like identity typically sample trajectories with large work (therefore low
efficiency) because the nonequilibrium dynamics do not have enough time to
equilibrate under nonequilibrium force. Therefore, by decreasing
$\tau$ in (\ref{dynamics-f-tau}), the mixing of the ``equilibrium part'' of the nonequilibrium system becomes
faster at each fixed nonequilibrium force. 
Numerically, the work $W$ of the sampled trajectories is likely to be both smaller and more concentrated. 
From the analysis in Appendix~\ref{app-1} and Appendix~\ref{app-2}, we know
that the free energy calculation method based on Jarzynski-like identity
(\ref{generalized-jarzynski-coordinate}) reduces to the thermodynamic integration method when $\tau\rightarrow 0$.
In practice, $\tau$ should be chosen not very small since otherwise the system
will become more stiff and a smaller time step-size has to be used in numerical integration.
Readers are referred to Subsection~\ref{subsec-ex2} for numerical study of free energy
calculation using different $\tau$. 
\section{Numerical examples}
\label{sec-examples}
We consider two simple examples 
and study the efficiency of Monte Carlo methods for 
free energy computation. 
\subsection{Example $1$: 1D example in alchemical transition case}
\label{subsec-ex1}
In this example, we consider one-dimensional potentials 
\begin{align}
  V(x, \lambda) = (1-\lambda)\frac{(x + 1)^2}{2} + \lambda \Big(\frac{(x^2-1)^2}{4}-0.4x\Big)\,,
  \label{pot-v-ex1}
\end{align}
where $x\in \mathbb{R}$ and $\lambda \in [0,1]$. As $\lambda$ increases from $0$ to $1$, 
$V(\cdot, \lambda)$ varies from a quadratic potential centered at $x=-1$ to a
tilted double well potential (Figure~\ref{sub-fig-pot-all}).
Recalling the free energy $F$ defined in (\ref{free-energy}), (\ref{normal-const}),
we will compute free energy differences $\Delta F(\lambda) = F(\lambda) -
F(0)$, using Monte Carlo based on Jarzynski's identity (\ref{jarzynski-repeat}).
We fix $\beta = 5.0$ and the SDE 
\begin{align}
  dx(s) = -\frac{\partial V}{\partial x}(x(s), \lambda(s))\, ds + \sqrt{2\beta^{-1}} dw(s)\,,
  \label{sde-ex1}
\end{align}
with control protocol $\lambda(s) = s$, $s \in [0,1]$, will be considered in the
Monte Carlo simulations. Clearly, for the initial distribution
$\mu_0=\mu_{\lambda(0)}$, we have $\frac{d\mu_0}{dx} \propto \exp\big(-\beta
\frac{(x+1)^2}{2}\big)$.

In fact, since the problem is one dimensional in space, we can directly compute 
the normalization constant $Z(\lambda)$
by numerically integrating (\ref{normal-const}) and therefore
obtain the free energy differences $\Delta F(\lambda)$, which are shown in Figure~\ref{subfig-df-cureve}. 
In particular, we obtain $\Delta F(1) = F(1) - F(0)= -3.44 \times 10^{-1}$
and this will be our reference solution. 
Furthermore, we can also approximate the optimal change of measure
$\mathbf{P}^*$ in (\ref{opt-p-decomp}) by computing the optimal control force $u^*$ and
the optimal initial distribution $\mu_0^*$ according to (\ref{opt-u}), (\ref{opt-mu0}), respectively. 
 For this purpose, we need to compute the function 
  $g(x,t) = \mathbf{E}_{x,t} \big(e^{-\beta W_{(t,T)}}\big)$
 in (\ref{g-fun-repeat}) which satisfies (\ref{g-pde-repeat}). 
 Notice that, in the current setting, we have $T=1$ and (\ref{g-pde-repeat}) becomes 
 \begin{align}
   \begin{split}
   &\frac{\partial g}{\partial t} - \frac{\partial V}{\partial x} \frac{\partial g}{\partial x}
   + \frac{1}{\beta} \frac{\partial^2 g}{\partial x^2} -\beta \big(V(x,1) -
   V(x,0)\big) g = 0\,,\quad  0 \le t < 1\,,\\
   & g(\cdot, 1) = 1\,.
 \end{split}
 \label{g-fun-ex1}
 \end{align}
 To compute $g$, we truncate the space of $(x,t)$ 
 to $[-5.0, 5.0] \times [0,1]$ and discretize the PDE (\ref{g-fun-ex1}) 
on a uniform grid of size $10000\times 10000$, following a similar way that was described in \cite{Hartmann2017-ptrf,zhangs-schuette-entropy-2017}. The
solution $g$ is obtained by solving the discretized system backwardly from
$t=1$ to $0$. The function $U=-\beta^{-1}\ln g$ is displayed in Figure~\ref{subfig-log-g-2d}
and the profile of $g(\cdot, 0)$ at $t=0$ is shown in Figure~\ref{subfig-g-t0}.
Based on these results, we can obtain the optimal control 
potentials (which is $V+2U$ according to (\ref{dynamics-1-u}) and (\ref{opt-u})) and the optimal initial distribution $\mu_0^*$. These results
are shown in Figure~\ref{sub-fig-opt-pot-lam}, Figure~\ref{subfig-u-t0} and
Figure~\ref{subfig-start-mu}, respectively. In particular, combining the 
expression (\ref{opt-mu0}) with Figure~\ref{subfig-g-t0} and Figure~\ref{subfig-start-mu}, it can be observed 
that, due to the strong
inhomogeneity of $g(\cdot, 0)$, the high probability density region of the
optimal initial distribution $\mu_0^*$ is
shifted along the positive $x$ axis and has little overlap with that of the
distribution $\mu_0$.

Now we turn to discuss the performance of Monte Carlo methods. First of all, we
apply the standard Monte Carlo method to estimate free energy differences. SDE
(\ref{sde-ex1}) is discretized with time step-size $\Delta s = 5 \times
10^{-4}$ and we repeat the simulation $10$ times. 
For each independent run, the estimator 
\begin{align}
  \mathcal{I}(\lambda) = \frac{1}{N} \sum_{i=1}^N e^{-\beta W_i(\lambda)}
  \label{estimator-stdmc-ex1}
\end{align}
is computed by generating $N = 5 \times 10^{5}$ trajectories of dynamics (\ref{sde-ex1}) starting
from $\mu_0$, where $W_i(\lambda)$ is the numerical approximation of 
(\ref{work-w-special}) on $[0, \lambda]$ for the $i$th trajectory.
The free energy differences are then estimated by 
\begin{align}
  \Delta F(\lambda) \approx -\beta^{-1}\ln \mathcal{I}(\lambda)\,,
  \label{df-ex1}
\end{align}
which is asymptotically unbiased when $N \rightarrow +\infty$.
The results are summarized in Figure~\ref{subfig-df-cureve},
Figure~\ref{subfig-multi-run} as well as 
in the last row of Table~\ref{tab-1}. We can observe that the estimations of
free energy differences have
very large fluctuations within the $10$ runs and the standard Monte Carlo
estimator (\ref{estimator-stdmc-ex1}) has a very large (sample) standard deviation. 

Noticing that the initial distribution $\mu_0$ in fact is very different from
the optimal initial distribution $\mu_0^*$, 
we have also used the probability measure $\bar{\mu}_0$, which is given by 
$\frac{d\bar{\mu}_0}{dx} \propto \exp\big(-\beta
\frac{(x-0.5)^2}{2}\big)$, 
as the initial distribution in importance sampling Monte Carlo methods. 
From the profiles of their probability density functions in
Figure~\ref{subfig-start-mu}, we expect that the importance sampling Monte
Carlo estimators using $\bar{\mu}_0$ will have better performance than 
estimators using $\mu_0$. 
Besides the change of measure in the initial distribution, the controlled dynamics
\begin{align}
  dx(s) = -\frac{\partial V}{\partial x}(x(s), \lambda(s))\, ds +
  \sum_{l=1}^k\omega_l\phi^{(l)}(x(s), s)\, ds + \sqrt{2\beta^{-1}}\, dw(s)
  \label{sde-ex1-control}
\end{align}
is used to generate trajectories instead of dynamics (\ref{sde-ex1}), which leads to a further change of measure on
path space. In (\ref{sde-ex1-control}), $\phi^{(l)}$ are ansatz functions which
we choose to be either piecewise linear functions or Gaussian
functions~\cite{Hartmann2016-Nonlinearity}.  
In the case of piecewise linear ansatz function, we divide the domain $[-1.3,
1.3]$ uniformly into $30$ Voronoi cells $\mathcal{C}_l$ and the ansatz
functions are defined as $\phi^{(l)}(x,t) = (1-t)\mathbf{1}_{\mathcal{C}_l}(x)$, $1 \le l \le
30$, where $\mathbf{1}_{\mathcal{C}_l}$ denotes the characteristic function of
cell $\mathcal{C}_l$. In the case of Gaussian ansatz function, we choose
two functions $\phi^{(l)}(x,t) = \frac{\partial V^{(l)}}{\partial x}(x,t)$, where $l=1,2$ and 
\begin{align}
  V^{(1)}(x,t) = (1-t) \exp\Big(-\frac{x^2}{2}\Big)\,,\quad V^{(2)}(x,t) =
  (1-t)\exp\Big(-\frac{(x-1.2)^2}{4.5}\Big)\,.
  \label{gauss-ansatz-ex1}
\end{align}
In both cases, the ansatz functions are chosen based on the idea discussed in
Subsection~\ref{subsec-ce} and the dependence on time $t$ is included since we know
that the optimal control force, which is proportional to $\frac{\partial
g}{\partial x}$, vanishes at time $t=1$, due to the Dirichlet boundary condition in (\ref{g-fun-ex1}).

After these preparations, we apply the cross-entropy method discussed in
Subsection~\ref{subsec-ce} to optimize the coefficients $\omega_l$ in
(\ref{sde-ex1-control}) by simulating $10^5$ trajectories. 
The control forces at time $t=0$, as well as the control potentials in
Gaussian ansatz case are shown Figure~\ref{subfig-u-t0} and Figure~\ref{subfig-gauss-pot}, respectively. 
Apparently, although the control forces are different from the optimal one, all of them  
can help drive the system along the positive $x$ axis. 
Similarly as in the standard Monte Carlo case, we estimate the free energy differences 
using importance sampling Monte Carlo method for $10$ times where $N= 5\times
10^5$ trajectories of the controlled dynamics (\ref{sde-ex1-control}) are
simulated for each run. 
Instead of (\ref{estimator-stdmc-ex1}), estimator 
\begin{align}
  \mathcal{I}(\lambda) = \frac{1}{N} \sum_{i=1}^N e^{-\beta W_i(\lambda)}\, r_i
\label{estimator-ipmc-ex1}
\end{align}
is computed, where $r_i$ is the likelihood ratio given by Girsanov's theorem 
(see (\ref{girsanov-ce})).  The results are shown in
Figure~\ref{subfig-df-cureve}, Figure~\ref{subfig-multi-run}, as well as in Table~\ref{tab-1}.
Comparing to the standard deviation of the standard Monte Carlo estimator
(\ref{estimator-stdmc-ex1}),
we observe that the standard deviations of the importance sampling Monte Carlo
estimators $\mathcal{I}(\lambda)$ in (\ref{estimator-ipmc-ex1}) are significantly reduced 
when we applied a change of measure both in the initial distribution and in the
dynamics, i.e., when the controlled dynamics (\ref{sde-ex1-control}) with
initial distribution $\bar{\mu}_0$ is used. And both types of ansatz functions
exhibit comparable performances. 
To better understand the efficiency of Monte Carlo methods, 
the probability density functions and the mean values of work within the $10$ runs of simulations 
are shown in Figure~\ref{subfig-work-1}, Figure~\ref{subfig-work-2} and
Table~\ref{tab-1} for each Monte Carlo estimators.
 Clearly, by applying importance sampling both in the initial distribution and
 in the dynamics,
trajectories with low work value are more efficiently sampled, leading to a
much better efficiency of the Monte Carlo estimators. 

\begin{figure}[htpb]
  \subfigure[$V(x,\lambda)$]{\includegraphics[width=0.48\textwidth]{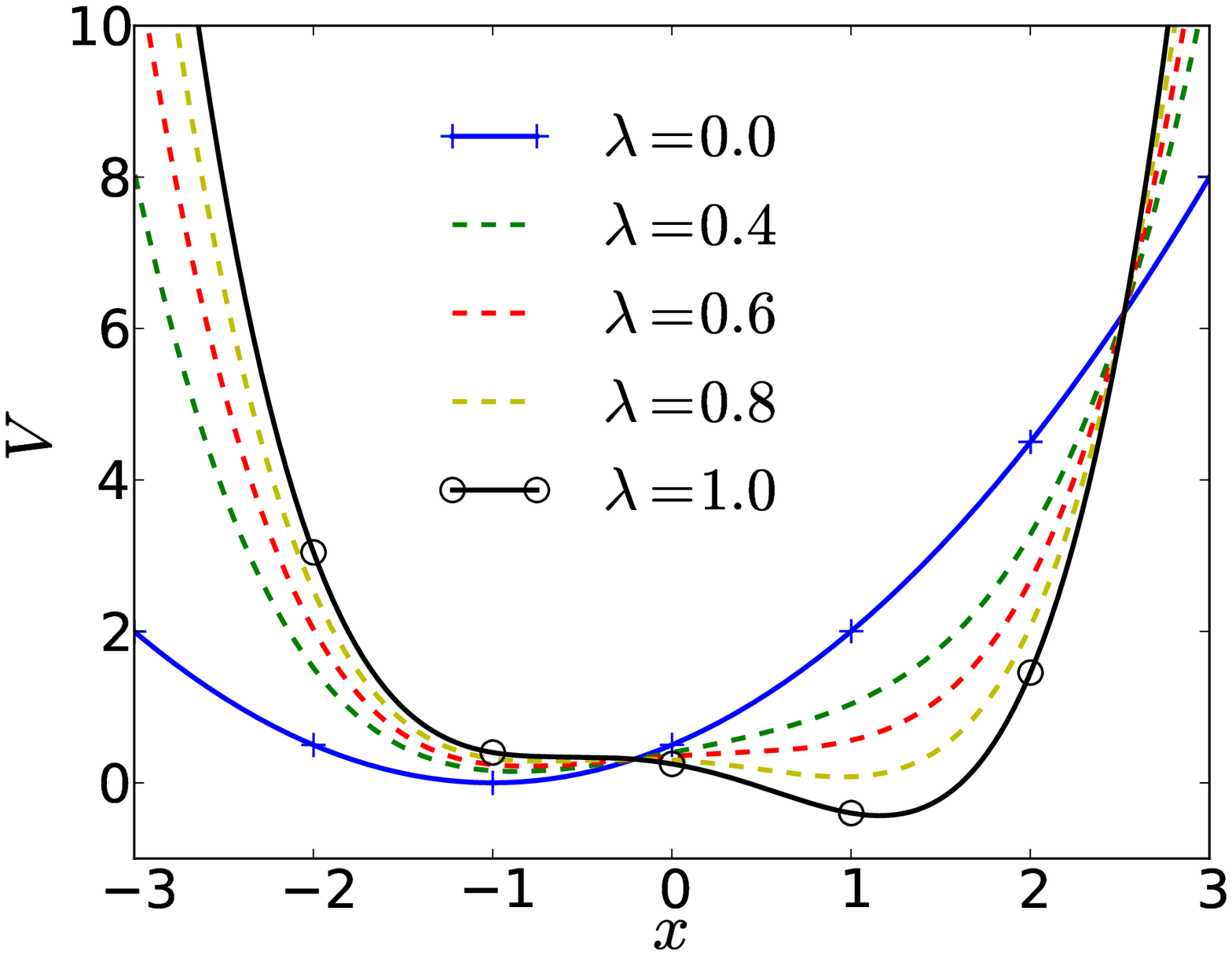}
  \label{sub-fig-pot-all}}
  \subfigure[$U=-\beta^{-1}\ln g$]{\includegraphics[width=.48\textwidth]{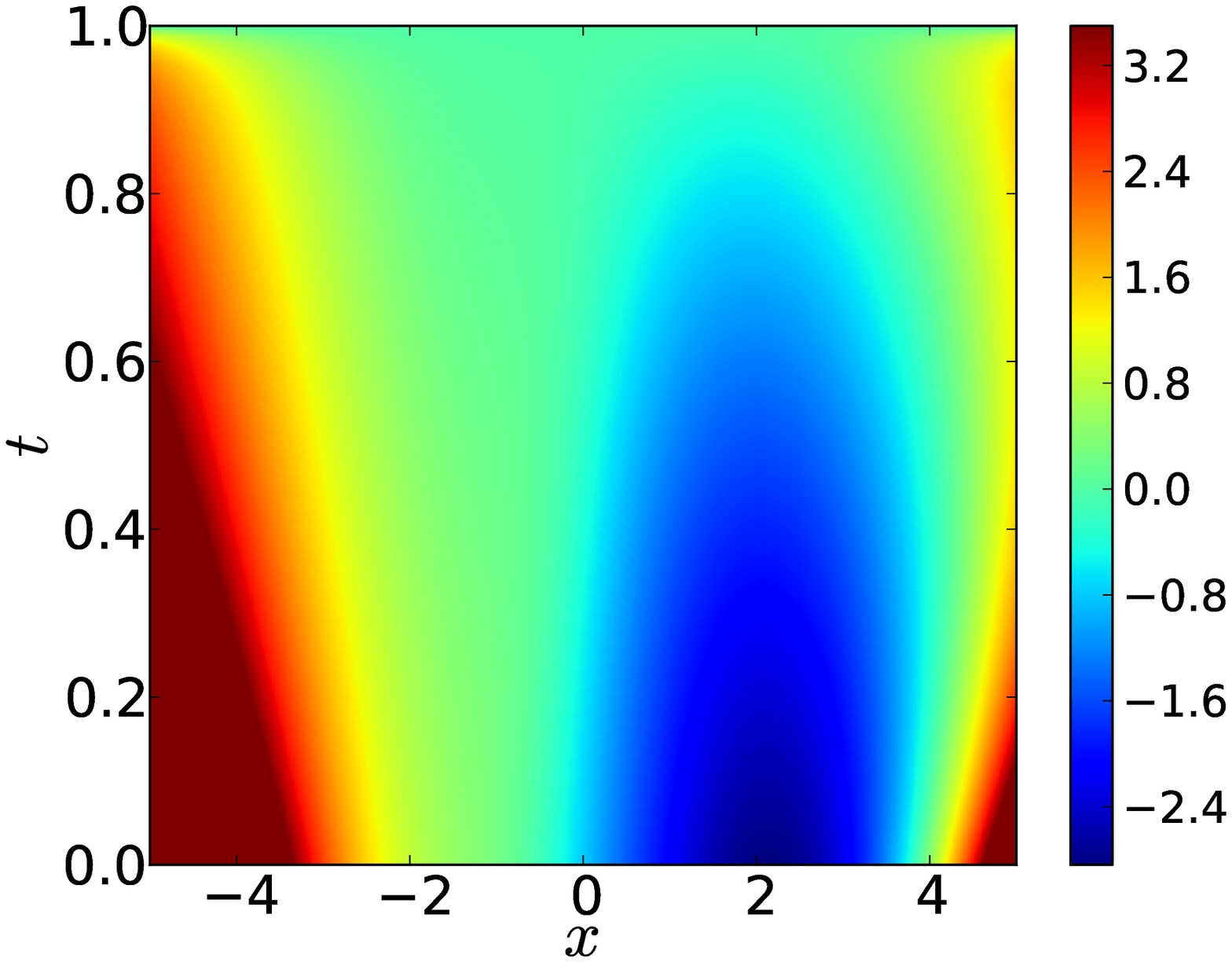}\label{subfig-log-g-2d}}
  \caption{Example $1$. (a) Potential $V(x,\lambda)$ in (\ref{pot-v-ex1}). (b) 
  Function $U=-\beta^{-1} \ln g$, where $\beta = 5.0$ and $g$ solves PDE
(\ref{g-fun-ex1}).}
\end{figure}

\begin{figure}[htpb]
  \subfigure[Optimally biased potentials]{\includegraphics[width=0.48\textwidth]{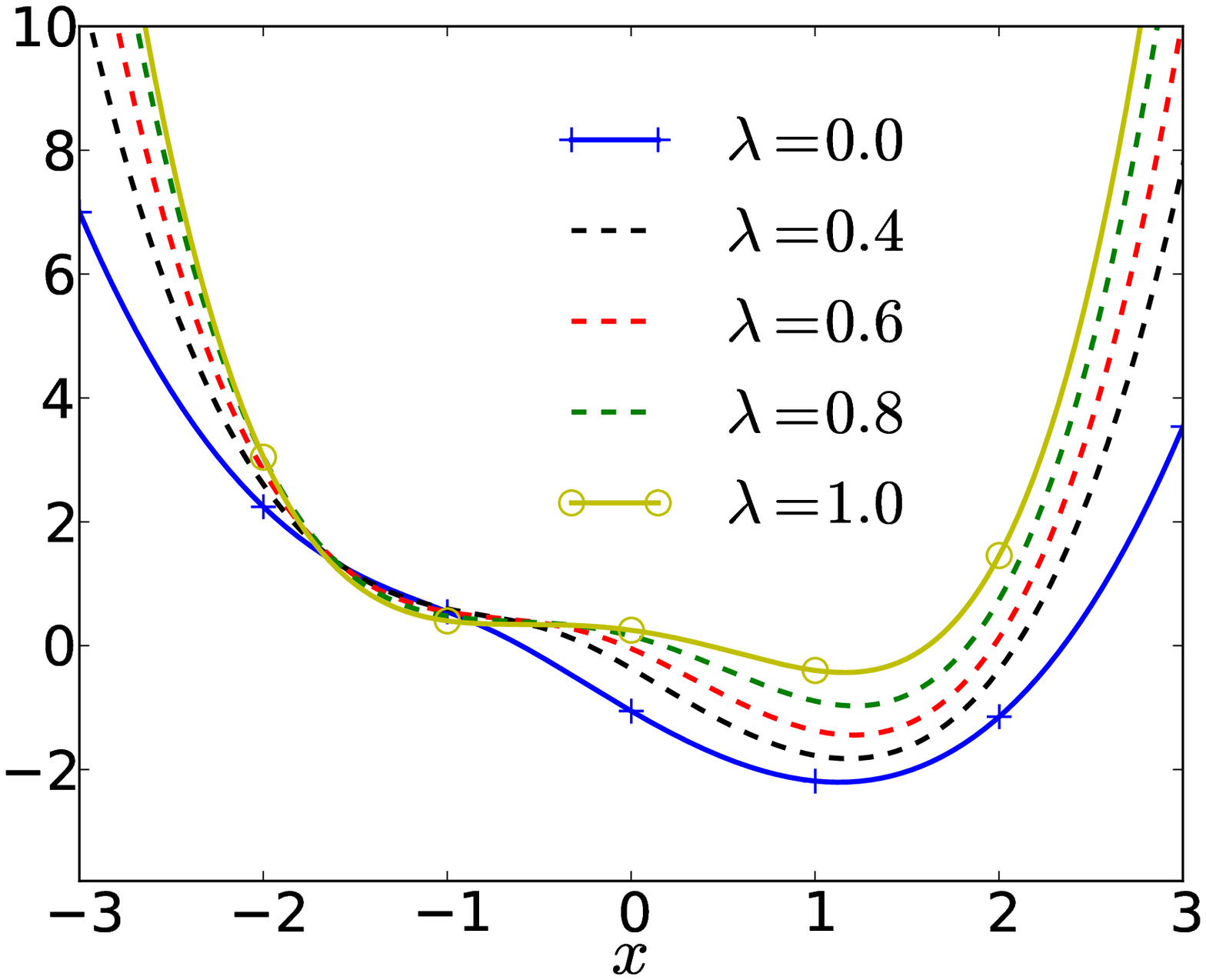}\label{sub-fig-opt-pot-lam}}
  \subfigure[Biased potentials using Gaussian ansatz]{\includegraphics[width=0.48\textwidth]{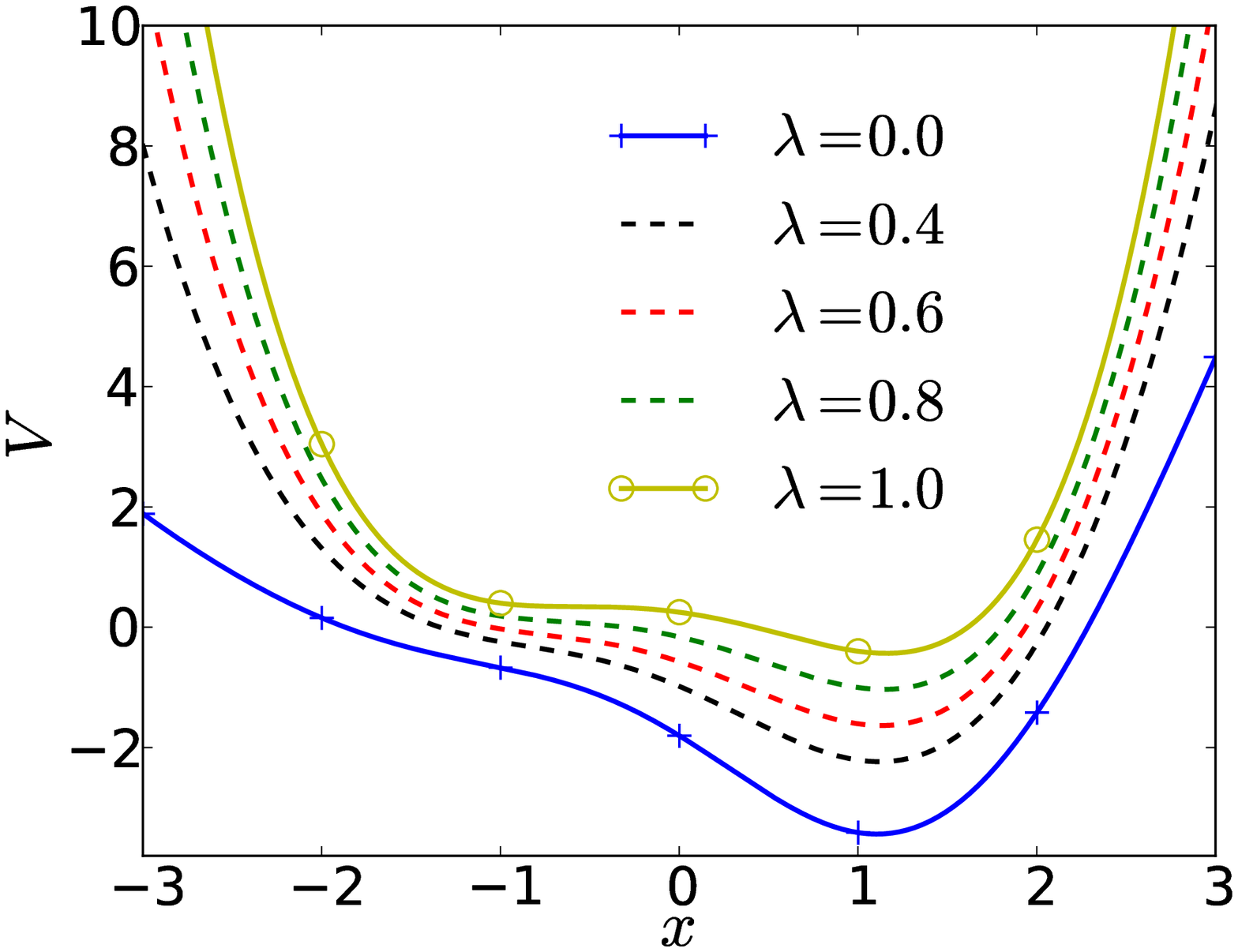}\label{subfig-gauss-pot}}
  \caption{Example $1$ with the control protocol $\lambda(s) = s$, for $s\in [0, 1]$. (a) Optimally biased potential ($V + 2U)$. (b) Biased potentials
  computed from cross-entropy method with Gaussian ansatz functions (\ref{gauss-ansatz-ex1}). }
\end{figure}
\begin{figure}[htpb]
  \subfigure[$g(x,0)$]{\includegraphics[width=0.48\textwidth]{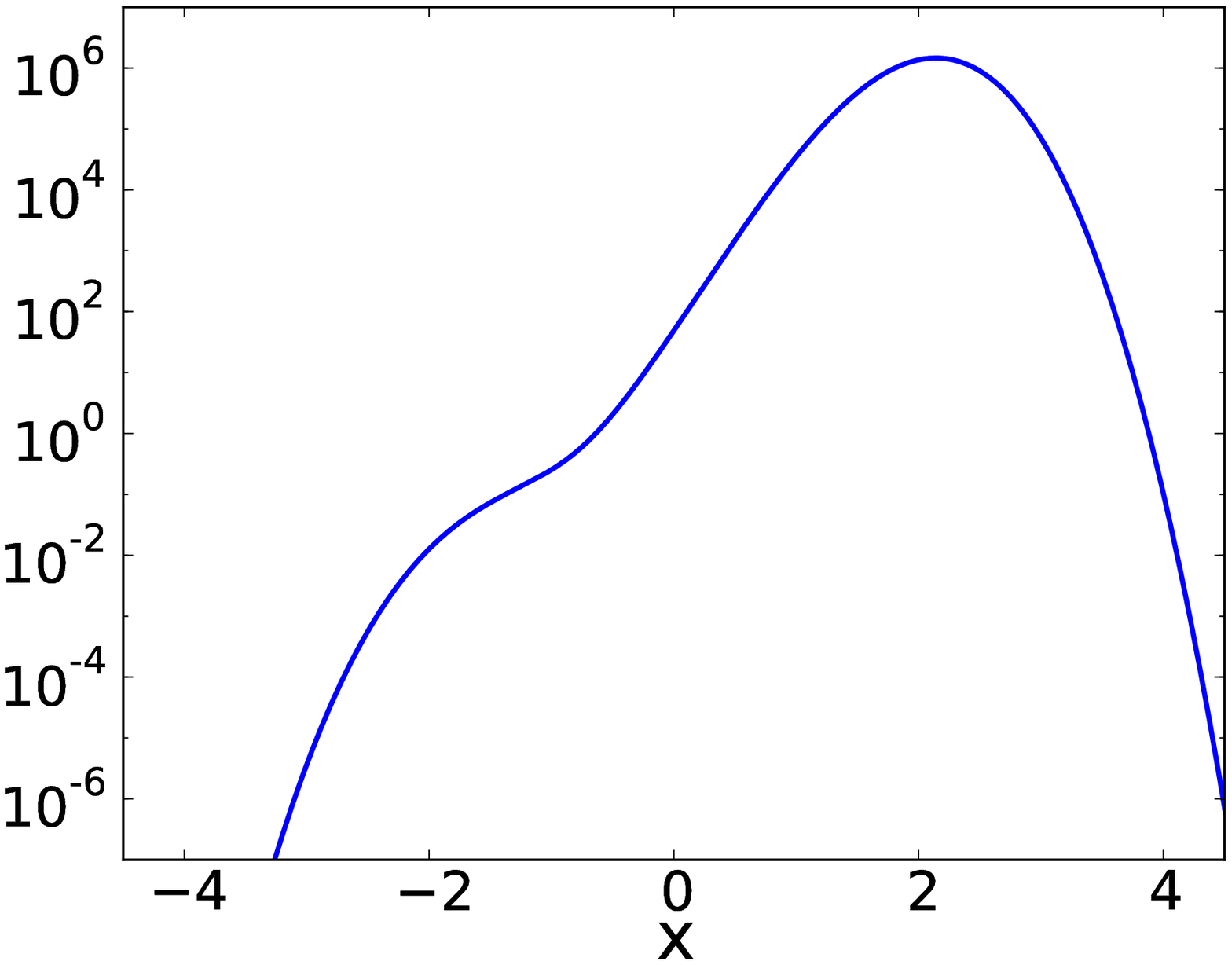}\label{subfig-g-t0}}
\subfigure[Control forces at $t=0$]{
  \includegraphics[width=0.48\textwidth]{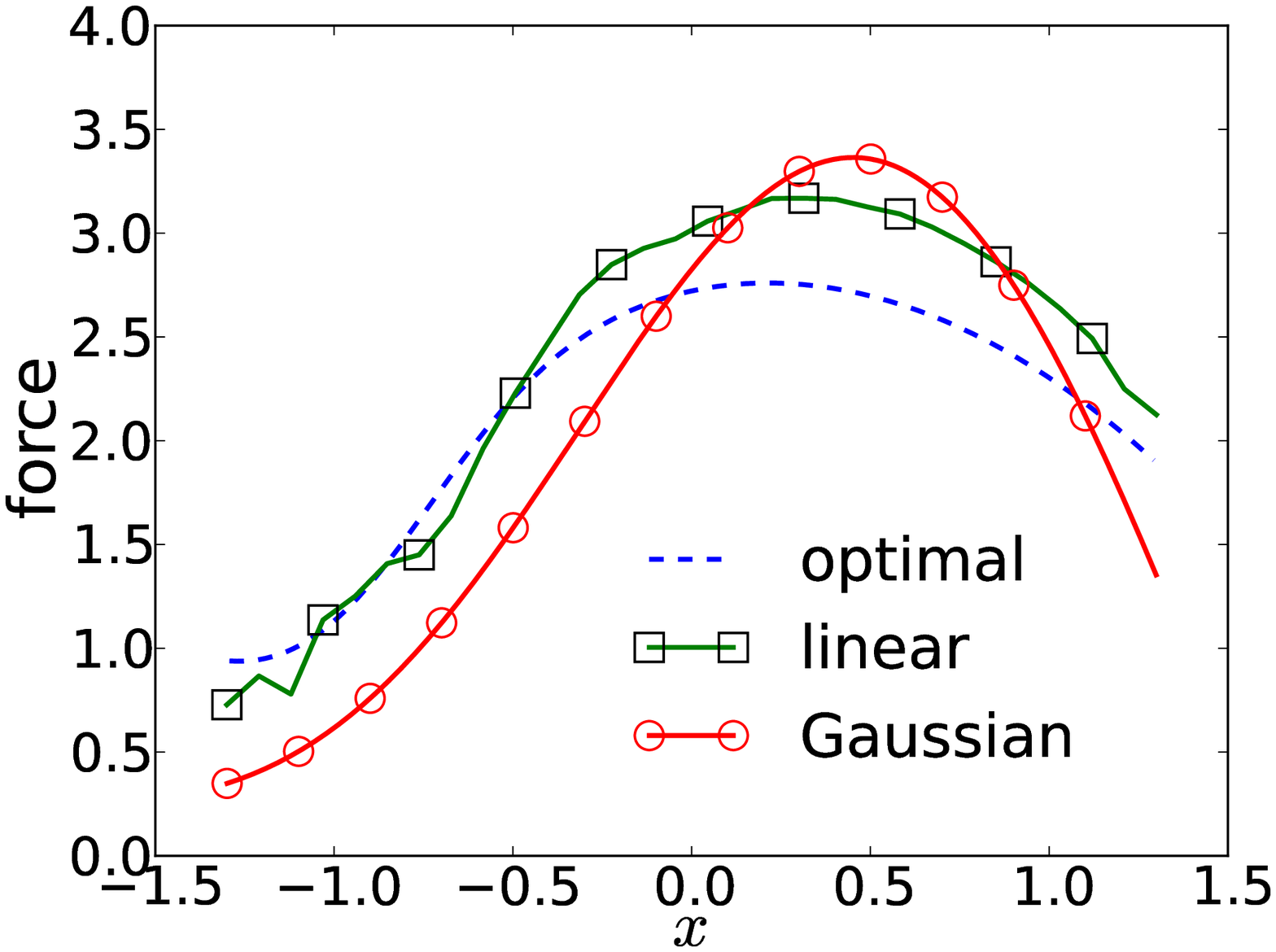}
  \label{subfig-u-t0}}
  \caption{Example $1$ with the control protocol $\lambda(s) = s$, for $s\in [0, 1]$. (a) Profile of the function $g(x, 0) =
  \mathbf{E}_{x,0}(e^{-\beta W})$ where $\beta = 5.0$ and $g$ solves PDE (\ref{g-fun-ex1}). (b)
Profiles of control forces at time $t=0$. Curves with Labels ``optimal'',
``linear'' and ``Gaussian'' correspond to the optimal control $u^*$, 
the control forces obtained from the cross-entropy method using piecewise linear and
Gaussian ansatz functions.}
\end{figure}
\begin{figure}[htpb]
  \centering 
  \includegraphics[width=7cm]{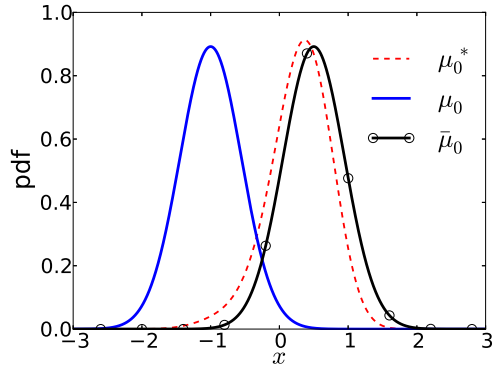}
  \caption{Example $1$ with the control protocol $\lambda(s) = s$, for $s\in [0, 1]$. Probability density functions of different initial distributions
  used in Monte Carlo methods for $\beta = 5.0$. The corresponding densities
  are $\frac{d\mu_0}{dx} \propto 
  \exp\big(-\beta \frac{(x+1)^2}{2}\big)$, $\frac{d\bar{\mu}_0}{dx} \propto
  \exp\big(-\beta \frac{(x-0.5)^2}{2}\big)$, and $\frac{d\mu_0^*}{dx}\propto 
\exp\big(-\beta \frac{(x+1)^2}{2}\big) g(x,0)$, which is given by (\ref{opt-mu0}). \label{subfig-start-mu}}
\end{figure}
\begin{figure}[htpb]
  \subfigure[]{\includegraphics[width=7cm]{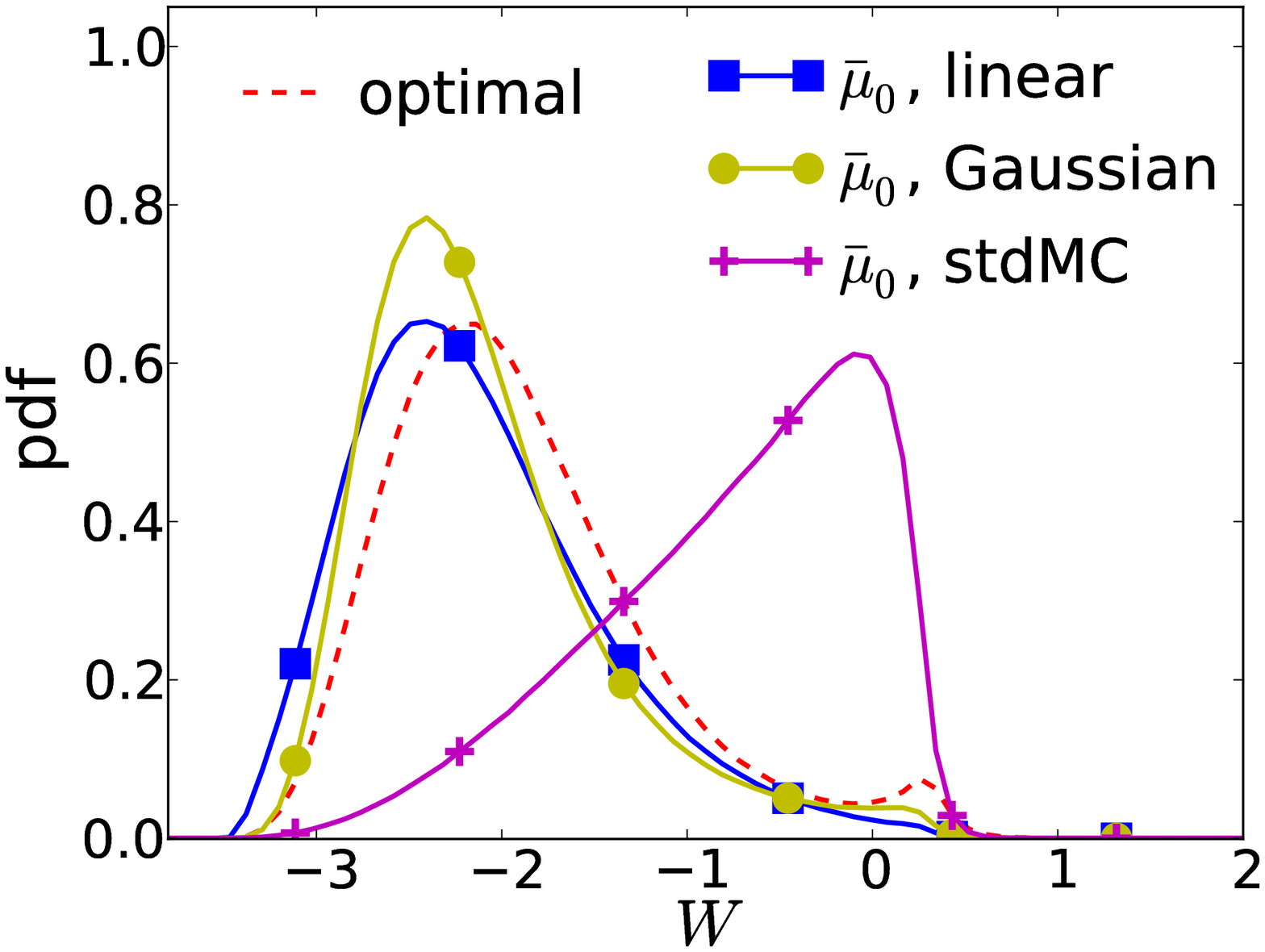}\label{subfig-work-1}}
  \subfigure[]{\includegraphics[width=7cm]{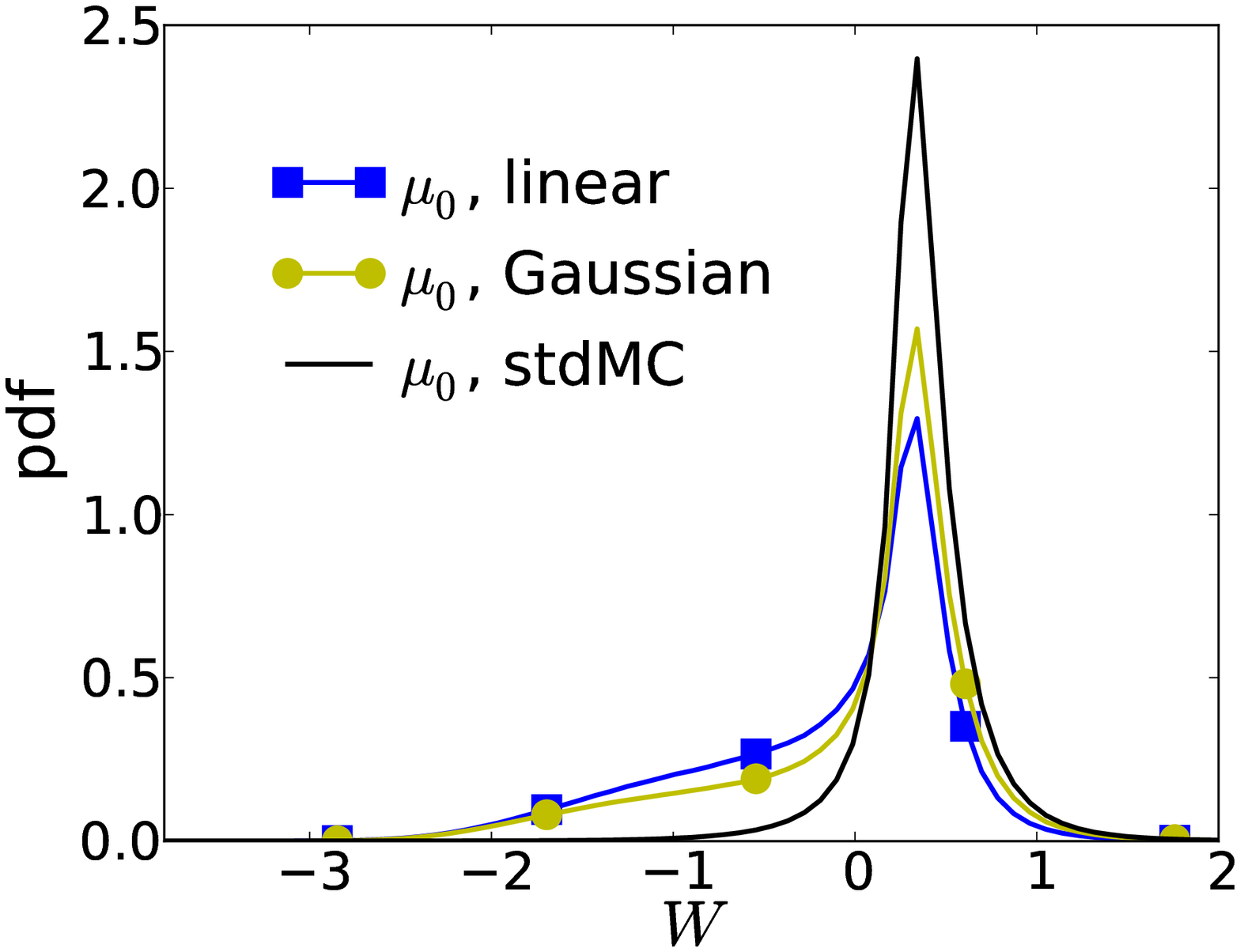}\label{subfig-work-2}}
  \caption{Example $1$ with the control protocol $\lambda(s) = s$, for $s\in
    [0, 1]$. Probability density functions of work along trajectories
    estimated from $10$ independent runs of Monte Carlo simulations where $5\times 10^5$ trajectories are simulated for
each run. (a) ``optimal'' corresponds to the importance sampling estimator with control
$u^*$ starting from the distribution $\mu^*_0$. The other three curves correspond to 
Monte Carlo estimators with initial distribution $\bar{\mu}_0$, 
using either the controlled dynamics (\ref{sde-ex1-control}) with piecewise linear ansatz
functions (Label ``$\bar{\mu}_0$, linear''), Gaussian ansatz
functions (Label ``$\bar{\mu}_0$, Gaussian''), or the uncontrolled dynamics
(\ref{sde-ex1}) (Label ``$\bar{\mu}_0$, stdMC''). (b)
Results correspond to Monte Carlo estimators with initial distribution $\mu_0$, 
using either the controlled dynamics (\ref{sde-ex1-control}) with piecewise linear ansatz
functions (Label ``$\mu_0$, linear''), Gaussian ansatz
functions (Label ``$\mu_0$, Gaussian''), or the uncontrolled dynamics (\ref{sde-ex1}) (Label ``$\mu_0$, stdMC''). 
\label{fig-work}
}
\end{figure}
\begin{figure}[htpb]
  \subfigure[$\Delta F(\lambda)$]{\includegraphics[width=7cm]{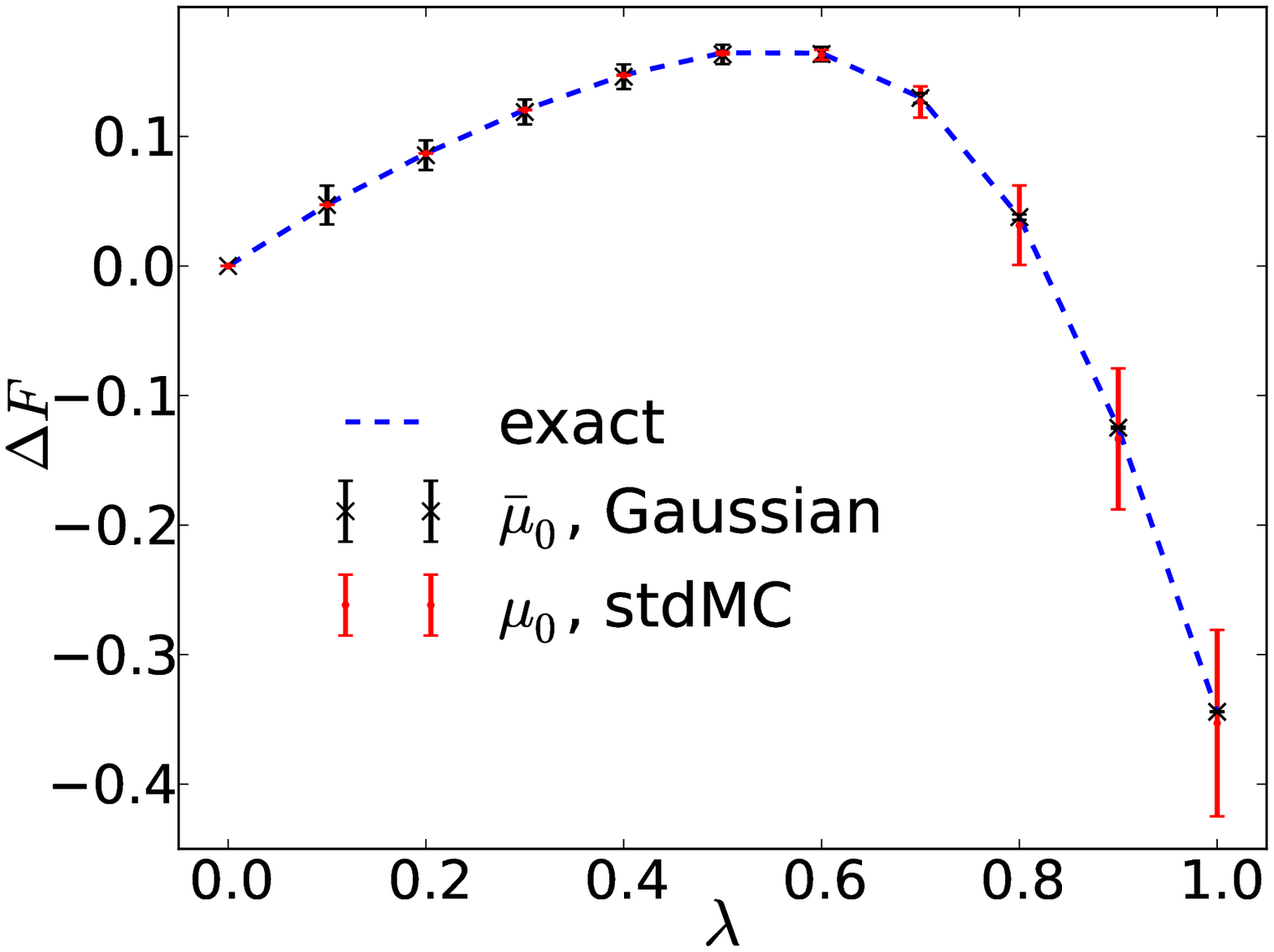}\label{subfig-df-cureve}}
  \subfigure[$\Delta F(1)$]{\includegraphics[width=7cm]{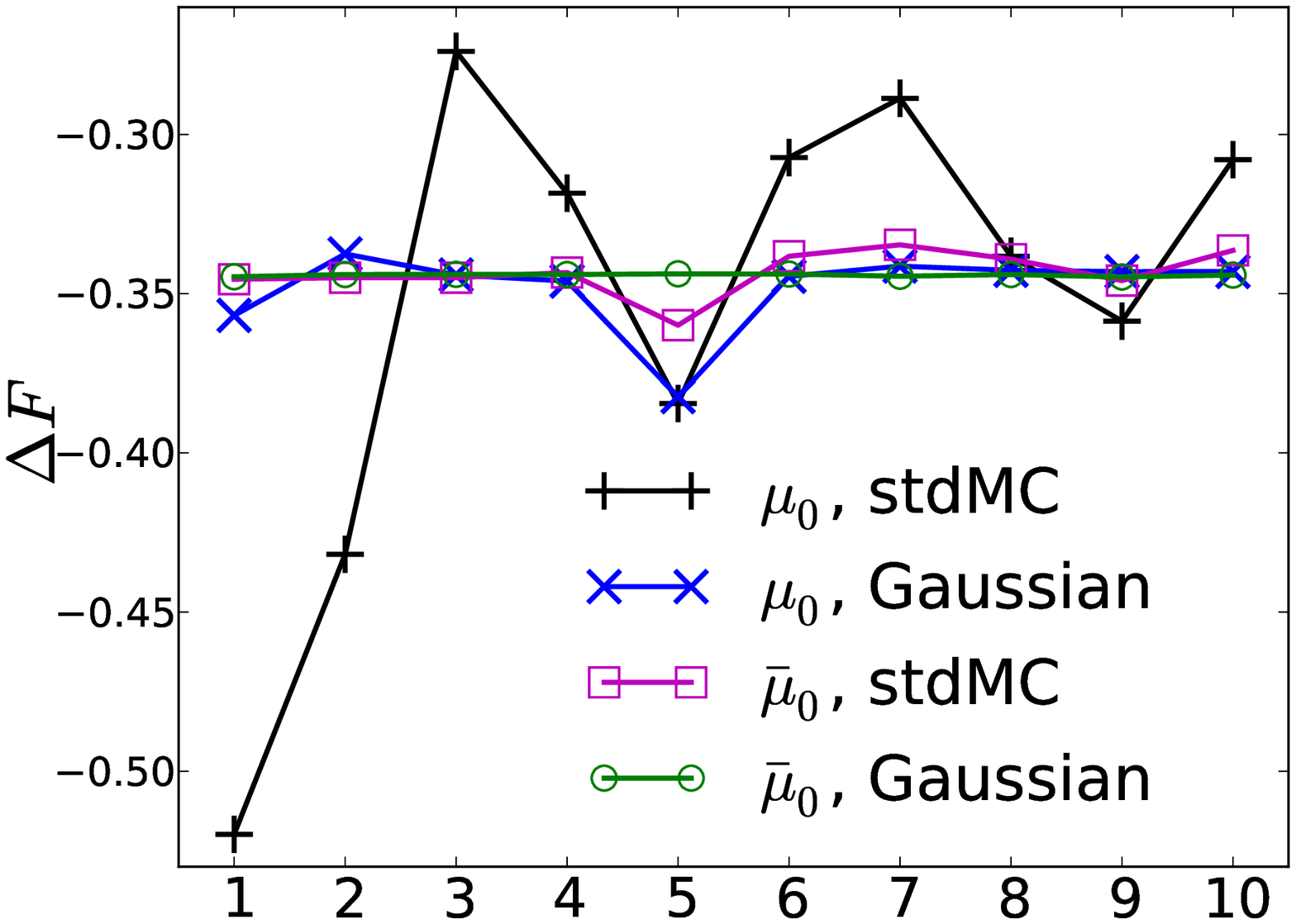}\label{subfig-multi-run}}
  \caption{Example $1$ with the control protocol $\lambda(s) = s$, for $s\in [0, 1]$. Labels of different curves have the same meaning as those appeared
in Figure~\ref{fig-work}.  (a) Profiles of free energy differences $\Delta F(\lambda)$ for
  $\lambda \in [0,1]$. Standard deviations of the free energy difference
  estimations for $10$ independent runs are shown in vertical error bar for different
  $\lambda$. ``exact'' corresponds to results obtained by directly
integrating the normalization constant $Z(\lambda)$ from (\ref{normal-const}).
(b) Mean values of free energy differences at $\lambda = 1$ for $10$
independent runs using different (importance sampling) Monte Carlo methods. For each run,
$5 \times 10^5$ trajectories of either SDE (\ref{sde-ex1}) or the controlled SDE
(\ref{sde-ex1-control}) are generated with time step-size $\Delta t = 5\times
10^{-4}$. Results corresponding to piecewise linear ansatz functions are
not shown here since they are very similar to those corresponding Gaussian ansatz
functions.}
\end{figure}
\begin{table}[htpb]
  \centering
  \begin{tabular}{cc|ccccc}
    \hline
    \hline
    initial & control & mean $\mathcal{I}$ & SD $\mathcal{I}$ & mean
 $\Delta F$ & SD $\Delta F$& mean $W$ \\
    \hline
    $\mu^*_0$ &  optimal & $5.58$ & $8.4\times 10^{-2}$ & $-3.44 \times
     10^{-1}$ & $2.4\times 10^{-4}$ & $-1.85$  \\
    \hline
    \multirow{3}{*}{$\bar{\mu}_0$}&  linear & $5.59$ &$6.0\times 10^{0}$
							      &$-3.44\times
						     10^{-1}$ &$3.4
						\times 10^{-4}$ & $-2.08$ \\
					   & Gaussian & $5.59$ &$7.1\times
		     10^{0}$ &$-3.44\times 10^{-1}$ &$3.5\times 10^{-4}$ & $-2.05$\\
      & stdMC & $5.51$ &$9.8\times 10^{1}$ &$-3.41\times 10^{-1}$ &$5.4\times 10^{-3}$ & $-0.71$\\
    \hline
    \multirow{3}{*}{$\mu_0$}& linear & $5.74$ &$2.2\times 10^{2}$
						 &$-3.49\times 10^{-1}$
						 &$1.0\times 10^{-2}$ & $-0.08$ \\
			    & Gaussian & $5.71$ &$2.6\times 10^{2}$ &$-3.48\times
	    10^{-1}$ &$1.3\times 10^{-2}$& $0.06$ \\
		     & stdMC & $6.28$ &$1.7\times 10^{3}$ &$-3.53\times 10^{-1}$&$7.2\times 10^{-2}$ & $0.40$\\
    \hline
    \hline
  \end{tabular}
  \caption{Example $1$ with the control protocol $\lambda(s) = s$, for $s\in [0, 1]$. Estimations of free energy difference for $\lambda = 1$
  using different (importance sampling) Monte Carlo methods. 
  Direct calculation of (\ref{normal-const}) and (\ref{free-energy}) gives
  the reference value $\Delta F = -3.44\times 10^{-1}$.
Column ``initial'' specifies the initial distribution that are used to generate
trajectories in Monte Carlo simulations.
Column ``control'' specifies the different dynamics (different control forces) and
the meaning of each name is the same as those appeared in Figure~\ref{fig-work}.
Columns ``mean $\mathcal{I}$'', ``SD $\mathcal{I}$'' show the mean and the sample standard
deviation of estimators (\ref{estimator-stdmc-ex1}) or (\ref{estimator-ipmc-ex1}).
Columns ``mean $\Delta F$'', ``SD $\Delta F$'' show 
the mean and the sample standard deviation of 
$10$ independent runs of the free energy difference estimations $\Delta F(1)$ using
(\ref{df-ex1}). The mean values of work $W$ for different Monte Carlo methods
are shown in Column ``mean $W$''. \label{tab-1} }
\end{table}
\subsection{Example $2$: reaction coordinate case}
\label{subsec-ex2}
In the second example, we study free energy calculation in the reaction
coordinate case considered in Section~\ref{sec-coordinate}. A similar example has
been considered in~\cite{effective_dynamics}, where the main focus was the approximation quality of effective dynamics. 
The system consists of three two-dimensional particles $A, B, C$ whose positions are at $x_A, x_B, x_C$, with
potential 
\begin{align}
  V(x_A, x_B, x_C) = \frac{1}{2\epsilon}\bigg\{r_{BC}-\Big[1 + \kappa
  \Big(\sin(\theta_{ABC})-\frac{1}{2}\Big)\Big] l_{eq}\bigg\}^2 +
  \frac{1}{2\epsilon}\big(r_{AB}-l_{eq}\big)^2 + V_3(\theta_{ABC})\,,
  \label{pot-v-ex2}
\end{align}
where $r_{AB}$, $r_{BC}$ are the distances between particles $A$ and $B$,
$B$ and $C$, respectively.  $\theta_{ABC}$ is the angle spanned by the bonds
$AB$ and $BC$, and $V_3$ is the potential of angle given by 
\begin{align}
  V_3(\theta) = \frac{k_\theta}{2}\Big((\theta-\theta_0)^2 - (\delta
  \theta)^2\Big)^2- k_{\theta, 1} (\theta-\theta_0) \,,
\end{align}
with $k_{\theta}>0$. 
Furthermore, in order to remove rigid body motion invariance, we
fix the position of particle $B$ ($x_B=0$) and particle $A$ is only allowed to move along horizontal axis.
For parameters, we take $\theta_0=\frac{\pi}{3}$, $\delta \theta = \frac{\pi}{6}$, 
$\epsilon = 0.1$, $k_\theta = 20$, $k_{\theta, 1} = 0.3$, and $l_{eq} = 5.0$. 

The system essentially has three degree of freedom, i.e., the position of
$x_C=(y_1, y_2)$ and the position of $x_A = (y_3,0)$ on the $x$-axis. The free energy
is defined according to (\ref{free-energy-coordinate}), where we take 
\begin{align}
  \xi(y_1,y_2,y_3) = \theta_{ABC}=\arctan\frac{y_2}{y_1} 
  \label{xi-ex2}
\end{align}
as the reaction coordinate function and $\beta = 5.0$. In order to calculate free energy differences, 
we consider the dynamics $y(s)=(y_1(s), y_2(s), y_3(s))$ in (\ref{dynamics-f-tau}) during the time interval $[0, 1]$ with
$a=\sigma=\mbox{id}$, and $f\equiv \frac{\pi}{3}$, starting from $\theta(y(0)) =
\frac{\pi}{6}$ at time $s=0$. 
In this case, the projection matrix in (\ref{p-ij}) can be directly computed as 
\begin{align}
  P = 
  \begin{pmatrix}
    \frac{y_1^2}{y_1^2+y_2^2} & \frac{y_1y_2}{y_1^2+y_2^2} & 0 \\
    \frac{y_1y_2}{y_1^2+y_2^2} & \frac{y_2^2}{y_1^2+y_2^2} & 0 \\
    0 & 0 & 1
  \end{pmatrix}
\end{align}
and we have $\Psi=|\nabla \xi|^2 = \frac{1}{y_1^2+y_2^2}$ in (\ref{psi-ij}). 
The angle $\theta_{ABC}$ of the system $y(s)$ evolves
uniformly during time $s \in [0,1]$ from $\frac{\pi}{6}$ to $\frac{\pi}{2}$ and the
free energy at $\theta_{ABC} = \frac{\pi}{6}$ is taken as reference. 
The free energy differences are calculated based on the Jarzynski-like identity (\ref{generalized-jarzynski-coordinate}),  
where the work $W$ is given in (\ref{w-coordinate-jarzynski-special})
and becomes as simple as 
\begin{align}
  W(t) = \int_0^t \Big(-y_2 \frac{\partial V}{\partial y_1} + y_1
  \frac{\partial V}{\partial y_2}\Big)(y(s))\, ds \,.
  \label{w-ex2}
\end{align}
In the numerical experiment below, we take $\kappa=0.3,\,0.6$ in the
potential $V$ in (\ref{pot-v-ex2}) and the
performance of the Monte Carlo estimator is tested using different values $\tau
= 1.0, \,0.6,\,0.3$ in dynamics (\ref{dynamics-f-tau}). 
In each case, we estimate the free energy differences based on $10$
independent runs of Monte Carlo sampling of 
\begin{align}
  \Delta F(\theta(t)) \approx -\beta^{-1} \ln \mathcal{I}(\theta(t)) = -\beta^{-1}
  \ln \Big(\frac{1}{N} \sum_{i=1}^N e^{-\beta W_i(t)}\Big) \,,
  \label{estimator-ex2}
\end{align}
 where $\theta(t) = \frac{\pi}{6} + \frac{\pi}{3} t$. In each run, $N=5 \times 10^5$ trajectories of dynamics (\ref{dynamics-f-tau})
are simulated using time step-size $\Delta t = 10^{-4}$, where $W_i$ denotes the
work (\ref{w-ex2}) of the $i$th trajectory.

The numerical results are shown in Figure~\ref{fig-ex2-1}
, Figure~\ref{fig-ex2-work-pdf} (results for $\kappa=0.3$ are similar and
therefore are not displayed) and Table~\ref{tab-ex2}.
From both Figure~\ref{fig-ex2-1} and Table~\ref{tab-ex2}, we can observe that the
free energy calculation using $\tau=1.0$ lead to large fluctuations and inaccurate estimations. On the
other hand, by decreasing $\tau$ to $0.3$, the variance of $10$ independent runs of free energy calculation decreases 
significantly and the results become stable. Based on the $10$ runs of
Monte Carlo simulations of the nonequilibrium dynamics, 
we can also estimate the probability density functions
of the work (\ref{w-ex2}) and the results are shown in Figure~\ref{fig-ex2-work-pdf}.
It can be seen that, as $\tau$ decreases, the probability density functions
shift along the negative horizontal axis and become more concentrated. 
This indicates that the work of the sampled paths becomes smaller on average and the variance decreases. 
All these results confirm that variance of the Monte Carlo estimator can
be reduced by decreasing the value of $\tau$ (see discussions at the end of
Subsection~\ref{subsec-info-the-ce-coordinate}). 
\begin{figure}[htpb]
  \centering 
  \subfigure[$\Delta F(\theta)$]{\includegraphics[width=7cm]{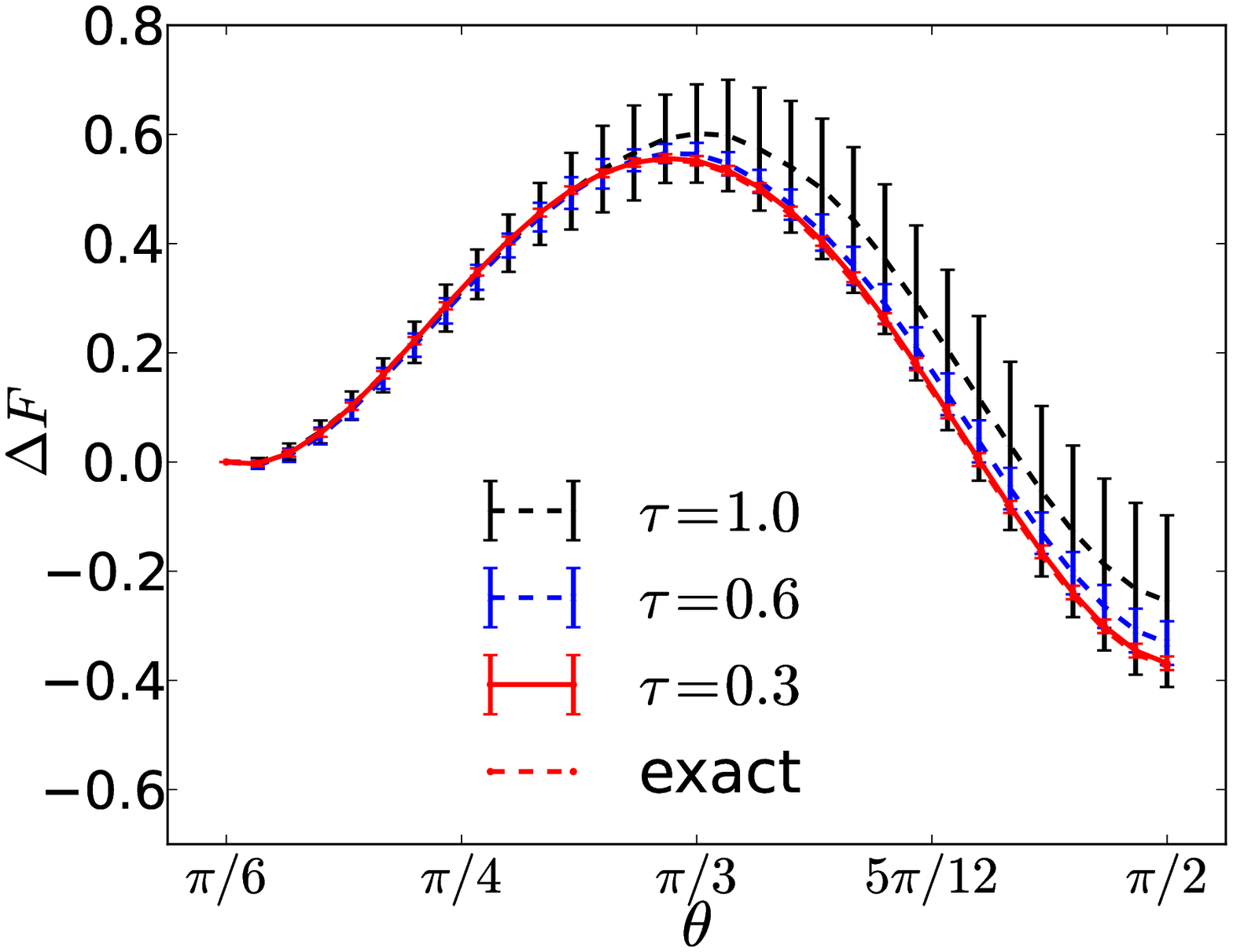}\label{fig-ex2-df-curve}}
  \subfigure[$\Delta
  F(\frac{\pi}{2})$]{\includegraphics[width=7cm]{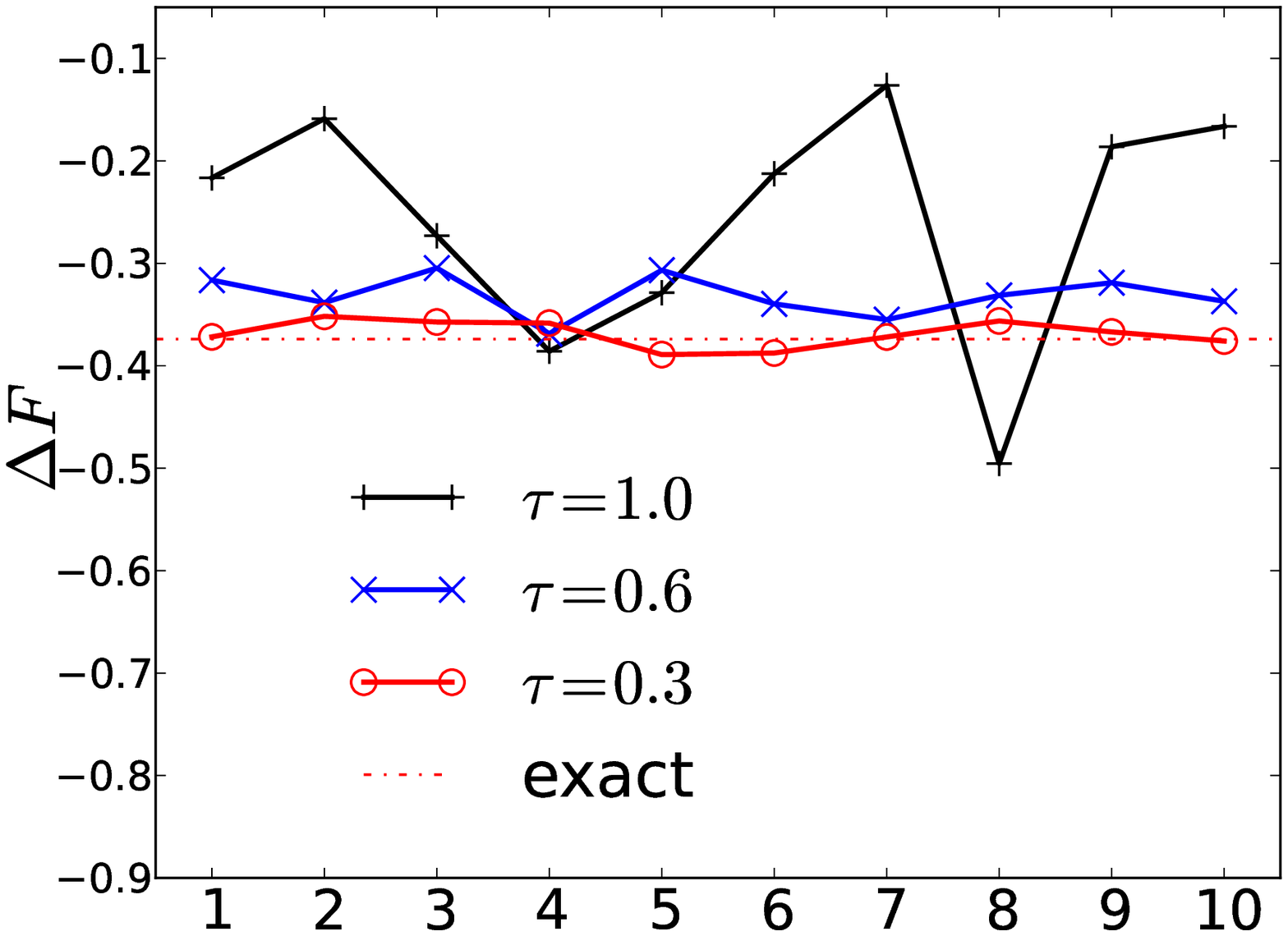}\label{fig-ex2-df-10runs}}
  \caption{Example $2$ for $\kappa=0.6$. (a) Profiles of free energy differences $\Delta
  F(\theta)$ for $\theta=\theta_{ABC} \in [\frac{\pi}{6},\frac{\pi}{2}]$
  computed using different $\tau$ in (\ref{dynamics-f-tau}). Standard deviations of the free energy difference
  estimations for $10$ independent runs are shown in vertical error bar for different
  $\theta$. ``exact'' corresponds to the reference results obtained by directly
integrating the normalization constants $Q(\cdot)$ appeared in (\ref{mu-z}).
Curves with Label ``$\tau=0.3$'' and Label ``exact'' almost coincide.
  (b) Mean values of free energy differences at $\theta = \frac{\pi}{2}$ for $10$ runs
  of Monte Carlo simulations using different values of $\tau$ in (\ref{dynamics-f-tau}). 
  The horizontal line with Label ``exact'' corresponds to the reference value $\Delta
  F(\frac{\pi}{2}) = -3.74\times 10^{-1}$.
  For each run, $5 \times 10^5$ trajectories of SDE (\ref{dynamics-f-tau}) are generated with
  time step size $\Delta t = 10^{-4}$.\label{fig-ex2-1}}
\end{figure}
\begin{figure}[htpb]
  \centering 
  \includegraphics[width=7cm]{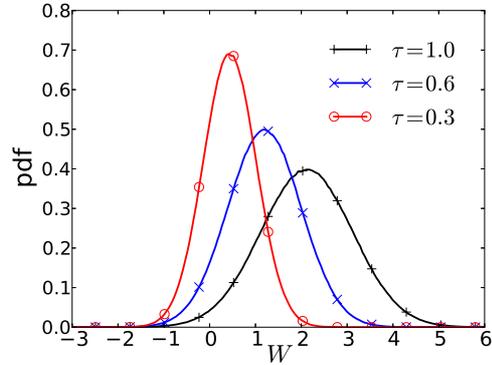}
  \caption{
    Example $2$ for $\kappa=0.6$.
  Probability density functions of the work $W$ (\ref{w-ex2}) along
  trajectories of (\ref{dynamics-f-tau})
  for different values $\tau=1.0, \,
  0.6,\, 0.3$. For each $\tau$, the probability density function is estimated from $10$ runs of
  Monte Carlo simulations where $5\times 10^5$ trajectories are simulated in
each run. \label{fig-ex2-work-pdf}
}
\end{figure}
\begin{table}[htpb]
  \centering
  \begin{tabular}{cc|ccccc}
    \hline
    \hline
    $\kappa$ & $\tau$ & mean $\mathcal{I}$ & SD $\mathcal{I}$ & mean $\Delta F$ & SD $\Delta F$& mean $W$ \\
    \hline
\multirow{3}{*}{$0.3$}
    & $1.0$ & $5.46$ & $9.4\times 10^{1}$ & $-3.39 \times 10^{-1}$ & $1.4\times 10^{-2}$ & $0.29$  \\
    &  $0.6$ & $5.67$ &$4.0\times 10^{1}$ &$-3.47\times 10^{-1}$ &$1.1 \times 10^{-2}$ & $0.05$ \\
     & $0.3$ & $5.52$ &$1.5\times 10^{1}$ &$-3.42\times 10^{-1}$ &$2.9\times 10^{-3}$ & $-0.13$\\
    \hline
\multirow{3}{*}{$0.6$}
    & $1.0$ & $4.27$ & $2.0\times 10^{3}$ & $-2.55 \times 10^{-1}$ & $1.6\times 10^{-1}$ & $2.14$  \\
    &  $0.6$ & $5.28$ &$5.1\times 10^{2}$ &$-3.32\times 10^{-1}$ &$4.0 \times 10^{-2}$ & $1.22$ \\
     & $0.3$ & $6.33$ &$2.3\times 10^{2}$ &$-3.69\times 10^{-1}$ &$1.3\times 10^{-2}$ & $0.46$\\
    \hline
    \hline
  \end{tabular}
  \caption{Example $2$. Estimations of free energy difference for $\theta= \frac{\pi}{2}$
  using Monte Carlo methods for different values $\kappa$ and $\tau$. 
  Direct calculation of (\ref{normal-const}) and (\ref{free-energy}) gives
  the reference value $\Delta F(\frac{\pi}{2}) = -3.42\times 10^{-1}$ and $\Delta F(\frac{\pi}{2}) = -3.74\times 10^{-1}$ for $\kappa=0.3$ and $0.6$,
  respectively.
Columns ``mean $\mathcal{I}$'', ``SD $\mathcal{I}$'' show the mean and the sample standard
  deviation of the estimator $\mathcal{I}$ in (\ref{estimator-ex2}). 
Columns ``mean $\Delta F$'', ``SD $\Delta F$'' show 
the mean and the sample standard deviation of 
  $10$ runs of free energy difference estimations $\Delta F(\frac{\pi}{2})$ using
(\ref{estimator-ex2}). The mean values of the work $W$ for Monte
  Carlo simulations using different $\kappa$ and $\tau$ are shown in the Column ``mean $W$''. \label{tab-ex2} }
\end{table}

\section{Conclusions}
\label{sec-conclusion}
In this work, we have studied nonequilibrium theorems for diffusion
processes. Jarzynski's equalities and fluctuation theorems are proved for 
quite general types of diffusion processes in both the alchemical transition case and
the reaction coordinate case. The information-theoretic formulation of
the Jarzynski's equality, as well as variance reduction approaches are discussed in both
cases. Our mathematical tools to derive these nonequilibrium relations are
from the theory of stochastic differential equation, in particular the Feynman-Kac formula and
the Girsanov's theorem. An advantage of the approach is that, it enables us to elucidate
the connections between Jarzynski's equality and the thermodynamic integration
identity, which were often treated as two distinct free energy calculation methods. 

Two variance reduction approaches for Monte Carlo methods have been studied 
in order to compute free energy differences using Jarzynski's equality. 
As demonstrated by simple examples, 
these approaches can largely improve the efficiency of Monte Carlo estimators
in both the alchemical transition case and the reaction coordinate case. 
One of the key findings is that variance
reduction by a change of measure requires to change both the initial distribution and 
the equation of the dynamics. 
We expect that our simple numerical studies can provide some insights into the source of sampling variances.

While the current work focuses on diffusion processes, the
mathematical tools may be applicable to other types of stochastic
processes, such as Markov chains, particle systems or networks, whose
evolution depends on external parameters. In future work, we will also
investigate free energy calculation for high-dimensional
applications using the variance reduction approaches proposed in this work,
together with the recent techniques of solving high-dimensional
PDEs~\cite{Darbon2016,deep-relaxation-osher2017,deep-pde-weinan2017}. 
\section*{Acknowledgement}
The authors acknowledge financial support by the Einstein Center of Mathematics (ECMath) through project CH21. 

\appendix 
\section{Connections with thermodynamic integration and adiabatic switching : Alchemical transition case}
\label{app-1}
In this appendix, we study two (essentially equivalent) asymptotic regimes of nonequilibrium processes
using formal arguments. 
In particular, we will derive the thermodynamic integration identity from Jarzynski's
identity, therefore bridging these two different free energy calculation
methods. Let us point out that such a connection is indeed known in physics
community~\cite{Crooks1998}, although we are not aware of its mathematical
derivation in the literature. For simplicity, we only consider the alchemical transition case studied in Section~\ref{sec-alchemical} and assume
the protocol $\lambda(\cdot)$ is deterministic with $\epsilon = 0$. 

\textbf{From Jarzynski's equality to thermodynamic integration}
Thermodynamic integration is a well known method and has been widely
used to compute free energy differences~\cite{frenkel2001understanding}. From the definition of the normalization constant $Z(\cdot)$ in
(\ref{normal-const}), we can derive the thermodynamic integration identity
by the simple argument 
\begin{align}
  \Delta F(T) = & F(\lambda(T)) - F(\lambda(0))\notag \\
  =& -\beta^{-1} \ln
  \frac{Z(\lambda(T))}{Z(\lambda(0))} \notag \\
  =& 
  -\beta^{-1} \int_0^T \frac{d}{ds}\Big( \ln
  \frac{Z(\lambda(s))}{Z(\lambda(0))}\Big)\, ds \notag \\ = &
  \int_0^T \Big( 
  \frac{\int_{\mathbb{R}^n} e^{-\beta V(x,\lambda(s))} \nabla_\lambda V(x,\lambda(s))\,
  dx)}{Z(\lambda(s))}\Big)\,\cdot f(\lambda(s), s)\, ds \notag \\
  =& \int_0^T \big(\mathbf{E}_{\mu_{\lambda(s)}} (\nabla_{\lambda} V)\big) \cdot
  f(\lambda(s), s)\, ds\,.
  \label{ti-identity}
\end{align}
In the following, using a formal argument, we show that the identity
(\ref{ti-identity}) corresponds to the Jarzynski's equality (\ref{jarzynski-1}) in certain asymptotic limit. For this purpose, 
we consider the dynamics 
\begin{align}
  \begin{split}
    d x(s) & =  \frac{1}{\tau} b(x(s), \lambda(s))\,ds +
    \sqrt{\frac{2\beta^{-1}}{\tau}} \sigma(x(s),
    \lambda(s)) \,dw^{(1)}(s)\,,
\end{split}
  \label{dynamics-time-rescaled}
\end{align}
on $s \in [0, T]$, where $0 < \tau \ll 1$ and $\lambda(s)$ satisfies the ODE
\begin{align}
  \dot{\lambda}(s) = f(\lambda(s), s)\,.
  \label{lambda-ode}
\end{align}
Clearly, dynamics (\ref{dynamics-time-rescaled}) is related to
(\ref{dynamics-1}) by rescaling time with the parameter $0 < \tau \ll 1$, and its
infinitesimal generator is $\frac{1}{\tau} \mathcal{L}_1$, where
$\mathcal{L}_1$ is defined in (\ref{l-1}) with $\lambda(\cdot)$ being time dependent. The main
observation is that, 
repeating the argument from Subsection~\ref{sub-sec-jarzynski}, the Jarzynski's equality (\ref{jarzynski-1}) holds for (\ref{dynamics-time-rescaled}) and (\ref{lambda-ode}) for any
$\tau>0$. As a consequence, 
\begin{align}
 e^{-\beta \Delta F(T)}
  = \mathbf{E}_{\mu(\lambda(0))} \Big(g(\cdot, \lambda(0), 0)\Big)  \,,
  \label{jarzynski-g-rescaled}
\end{align}
where the function $g$ now satisfies 
\begin{align}
  \begin{split}
    &\partial_t g + \frac{1}{\tau}\mathcal{L}_{1} g + f\cdot \nabla_\lambda g - \beta
    \big(f\cdot \nabla_\lambda V \big) g = 0\,, \quad 0 \le t < T\,,\\
    &g(\cdot,\cdot, T) = 1\,.
\end{split}
\label{g-pde-time-rescaled-eps0}
\end{align}
To show that (\ref{jarzynski-g-rescaled}) reduces to the thermodynamic
integration identity (\ref{ti-identity}) as $\tau \rightarrow 0$, it is enough
to study the asymptotic limit of (\ref{g-pde-time-rescaled-eps0}). 
To this end, we consider the formal asymptotic expansion 
\begin{align*}
g = g_0 + \tau g_1+\tau^2g_2 + \cdots
\end{align*}
as $\tau\rightarrow 0$, where $g_0, g_1, \cdots$ are functions independent of $\tau$. 
Substituting this expansion into (\ref{g-pde-time-rescaled-eps0}) and comparing
terms of different powers of $\tau$, we can conclude that $g_0=g_0(\lambda, t)$ is independent of $x$ and satisfies 
\begin{align}
  \begin{split}
  &\partial_t g_0 + \mathcal{L}_1 g_1 + f\cdot \nabla_\lambda g_0 - \beta
  (f\cdot \nabla_\lambda V)g_0 = 0\,, \quad 0 \le t < T\\
  & g_0(\cdot, T) = 1\,.
\end{split}
\label{g0-pde-before-closure}
\end{align}
Taking the expectation with respect to $\mu_{\lambda}$ on both sides of
(\ref{g0-pde-before-closure}) and noticing that $\mathbf{E}_{\mu_\lambda}(\mathcal{L}_1 g_1) = 0$, we obtain 
\begin{align}
  \begin{split}
  &\partial_t g_0 + f\cdot \nabla_\lambda g_0 - \beta
    \big(f \cdot \mathbf{E}_{\mu_{\lambda}}(\nabla_\lambda V)\big) g_0 = 0\,, \quad 0 \le t < T\\
  & g_0(\cdot, T) = 1\,.
\end{split}
\label{g0-pde-after-closure}
\end{align}
It is easy to verify that the solution of (\ref{g0-pde-after-closure}) is given by 
\begin{align}
g_0(\lambda, t) = e^{-\beta \int_t^T 
 \big(\mathbf{E}_{\mu_{\lambda(s)}} (\nabla_{\lambda} V)\big) \cdot
 f(\lambda(s), s)\, ds}\,,
\end{align}
where $\lambda(s)$ satisfies (\ref{lambda-ode}) with initial value $\lambda(t) = \lambda$.
Taking the limit $\tau\rightarrow 0$ in (\ref{jarzynski-g-rescaled}) then yields
\begin{align}
  e^{-\beta \Delta F(T)} = \lim_{\tau \rightarrow 0} 
  \mathbf{E}_{\mu(\lambda(0))} \Big(g(\cdot, \lambda(0), 0)\Big)  
  = g_0(\lambda(0), 0) =   e^{-\beta \int_0^T 
 \big(\mathbf{E}_{\mu_{\lambda(s)}} (\nabla_{\lambda} V)\big) \cdot
 f(\lambda(s), s)\, ds}\,,
\end{align}
which is equivalent to the thermodynamic integration identity
(\ref{ti-identity}).

\textbf{Adiabatic switching}
Now we turn to another (equivalent) asymptotic regime where the protocol
$\lambda(\cdot)$ is switched infinitely slowly. 
Specifically, given $\lambda_0, \lambda_1 \in \mathbb{R}^m$, 
the protocol $\lambda(\cdot)$ satisfying $\lambda(0) = \lambda_0$ and $\lambda(T) =
\lambda_1$ as $T\rightarrow +\infty$ is called adiabatic switching. For the
nonequilibrium process $x(\cdot)$ in (\ref{dynamics-1}) under adiabatic switching, it is well known
that we have 
\begin{align}
  F(\lambda_1) - F(\lambda_0) = \lim_{T\rightarrow +\infty}
  \mathbf{E}_{\lambda_0,
  0} \big(W(T)\big) = \lim_{T\rightarrow +\infty} \mathbf{E}_{\lambda_0, 0} \Big(\int_0^T
  \nabla_\lambda V(x(s), \lambda(s)) \cdot f(\lambda(s), s)\, ds\Big) \,,
  \label{df-work-adiabatic}
\end{align}
i.e., the free energy difference equals to the average work performed during the
switching. 
In the following we provide a formal mathematical argument to derive the
above identity. For this purpose, we define 
\begin{align}
  u(x,\lambda, t) = \mathbf{E}\Big(\int_t^T
  \nabla_\lambda V(x(s), \lambda(s)) \cdot f(\lambda(s), s)\, ds~\Big|~x(t) = x,
  \lambda(t) = \lambda\Big) \,,
\end{align}
which, by the Feynman-Kac formula, satisfies 
\begin{align}
  \begin{split}
  &\partial_t u + \mathcal{L}_1 u + f\cdot \nabla_\lambda u + f\cdot \nabla_\lambda V = 0\,, \\
  & u(\cdot, \cdot, T) = 0\,.
\end{split}
\label{adiabatic-work-pde}
\end{align}
Notice that,  as $T\rightarrow +\infty$, the switching becomes infinitely slow
and $\dot{\lambda}(t) = f$ goes to zero. Instead, we rescale the time by
$\bar{t} = \frac{t}{T} \in [0, 1]$ and define $\bar{\lambda}(\bar{t}\,) =
\lambda(\frac{\bar{t}}{\tau})$, where $\tau = \frac{1}{T} \rightarrow 0$.
$\bar{\lambda}(\cdot)$ satisfies $\bar{\lambda}(0) = \lambda_0, \bar{\lambda}(1) =
\lambda_1$ and 
\begin{align}
  \frac{d\bar{\lambda}}{d\bar{t}} = \bar{f}(\bar{\lambda}(\bar{t}\,),
  \bar{t}\,)\,,
  \label{ode-scaled}
\end{align}
where $\bar{f}(\cdot, \bar{t}\,) = \frac{1}{\tau}f(\cdot, \frac{\bar{t}}{\tau})$
is a function of $\mathcal{O}(1)$. Under this time scaling, PDE (\ref{adiabatic-work-pde})
becomes 
\begin{align}
  \begin{split}
    &\partial_{\bar{t}} u + \frac{1}{\tau} \mathcal{L}_1 u + \bar{f}\cdot \nabla_\lambda u + 
    \bar{f}\cdot \nabla_\lambda V = 0\,, \quad 0 \le \bar{t} < 1\,,\\
    & u \equiv 0\,, \quad \bar{t} = 1\,.
\end{split}
\label{adiabatic-work-pde-tbar}
\end{align}
Consider the expansion $u = u_0 + \tau u_1 + \tau^2 u_2 + \cdots$, then the same
argument as above yields that the function $u_0$ is independent of $x$ and
satisfies 
\begin{align}
  \begin{split}
    &\partial_{\bar{t}} u_0 + \bar{f}\cdot \nabla_\lambda u_0 + 
    \bar{f}\cdot \mathbf{E}_{\lambda} \big(\nabla_\lambda V\big) = 0\,, \quad 0 \le \bar{t} < 1\,,\\
    & u_0 \equiv 0\,, \quad \bar{t} = 1\,.
\end{split}
\label{adiabatic-work-pde-tbar-u0}
\end{align}
The solution of (\ref{adiabatic-work-pde-tbar-u0}) can be directly computed:
\begin{align}
  u_0(\lambda, \bar{t}\,) = \int_{\bar{t}}^1 \mathbf{E}_{\bar{\lambda}(s)}
  \big(\nabla_\lambda V\big)\cdot \bar{f}(\bar{\lambda}(s), s)\,ds \,,
\end{align}
where $\bar{\lambda}(\cdot)$ satisfies (\ref{ode-scaled}) on $[\bar{t},
1]$ with $\bar{\lambda}(\bar{t}\,) = \lambda$. 
In particular, taking $\bar{t} = 0$ and applying the thermodynamic integration
identity (\ref{ti-identity}), gives
\begin{align}
  u_0(\lambda_0, 0) = \int_{0}^1 \mathbf{E}_{\bar{\lambda}(s)}
  \big(\nabla_\lambda V\big)\cdot \bar{f}(\bar{\lambda}(s), s)\,ds 
  = F(\lambda_1) - F(\lambda_0)\,.
\end{align}
Therefore, 
\begin{align*}
  &\lim_{T\rightarrow +\infty} \mathbf{E}_{\lambda_0, 0} \Big(\int_0^T
  \nabla_\lambda V(x, \lambda(s)) \cdot f(\lambda(s), s)\, ds\Big)  \\
  =& 
  \lim_{\tau\rightarrow 0} \mathbf{E}_{\mu_{\lambda_0}} \big(u(\cdot,
  \lambda_0, 0)\big)\\
  =& u_0(\lambda_0, 0) = F(\lambda_1) - F(\lambda_0)\,,
\end{align*}
which concludes the proof of (\ref{df-work-adiabatic}).
\section{Thermodynamic integration identity in the reaction coordinate case} 
\label{app-2}
In the reaction coordinate case considered in Section~\ref{sec-coordinate},
connections between the thermodynamic integration identity and the Jarzynski's equality as
well as the adiabatic switching regime can be studied using the same
asymptotic argument as in Appendix~\ref{app-1}.
In this section, we omit the derivation and only provide the thermodynamic integration identity. 
We emphasize that both the identity and its proof can be found in the
literature, e.g.,~\cite{tony-free-energy-compuation,effective_dynamics}. The
result is included for readers' convenience.

Recall the definition of the probability measure $\mu_z$ in (\ref{mu-z}), where the normalization constant 
is given by 
  \begin{align}
    Q(z) = \int_{\mathbb{R}^n} e^{-\beta V(y)} \delta\big(\xi(y) - z\big)\,dy
    \,,\quad z \in \mathbb{R}^d\,,
    \label{q-z-repeat}
  \end{align}
and the free energy is defined in~(\ref{free-energy-coordinate}).
Let $z(s) \in \mathbb{R}^d$ satisfy the ODE~(\ref{zt-ode}) on $[0,T]$.
Similar to the derivations in (\ref{ti-identity}), and using Lemma~\ref{lemma-q-derivative-coordinate} below, we can compute
\begin{align}
  \begin{split}
  & F(z(T)) - F(z(0)) \\
  =& -\beta^{-1} \ln \frac{Q(z(T))}{Q(z(0))}\\
  =& -\beta^{-1} \int_0^T \frac{d}{ds}\Big(
  \ln \frac{Q(z(s))}{Q(z(0))}\Big)\, ds\\
    =& -\beta^{-1} \int_0^T \Big(\frac{1}{Q}\frac{\partial Q}{\partial
    z_{\gamma}}\Big)\big(z(s)\big)\,\dot{z}_{\gamma}(s)\, ds\\
    =& \int_0^T  
    \mathbf{E}_{\mu_{z(s)}}\Big[(a\nabla \xi_{\gamma'})_i (\Psi^{-1})_{\gamma'\gamma} 
    \frac{\partial V}{\partial y_i}
    - \frac{1}{\beta} 
    \frac{\partial}{\partial y_i}\Big(
    (a\nabla\xi_{\gamma'})_i (\Psi^{-1})_{\gamma\gamma'}\Big)\Big]\,
    \dot{z}_\gamma(s)\,ds\,,
  \end{split}
\end{align}
where Einstein's summation convention has been used.
\begin{lemma}
  Let the function $Q$ be defined in (\ref{q-z-repeat}). For $1 \le \gamma \le
  d$, we have 
  \begin{align*}
\frac{\partial Q}{\partial z_{\gamma}}(z)
    = -\beta Q(z) \int_{\Sigma_z} \Big[(a\nabla \xi_{\gamma'})_i (\Psi^{-1})_{\gamma\gamma'} 
    \frac{\partial V}{\partial y_i}
    - \frac{1}{\beta} 
    \frac{\partial}{\partial y_i}\Big(
    (a\nabla\xi_{\gamma'})_i (\Psi^{-1})_{\gamma\gamma'}\Big)\Big]\,\mu_z(dy)\,.
  \end{align*}
  \label{lemma-q-derivative-coordinate}
\end{lemma}
\begin{proof}
  Let $\varphi : \mathbb{R}^d \rightarrow \mathbb{R}$ be a smooth test
  function with compact support. For $1 \le \gamma \le d$, integrating by
  parts and using (\ref{q-z-repeat}), we have 
  \begin{align}
    \int_{\mathbb{R}^d} \varphi(z) \frac{\partial Q}{\partial z_{\gamma}}(z) \, dz 
    =
    -\int_{\mathbb{R}^d}  \frac{\partial \varphi}{\partial z_{\gamma}}(z) Q(z)\, dz 
  =
    -\int_{\mathbb{R}^n}  \frac{\partial \varphi}{\partial
    z_{\gamma}}(\xi(y)) e^{-\beta V(y)}\, dy \,.
    \label{q-diff-lemma-eqn1}
  \end{align}
  On the other hand, from the relation $$\frac{\partial
  \big(\varphi\circ\xi\big)}{\partial y_i}(y) = \frac{\partial \varphi}{\partial
  z_{\gamma'}}(\xi(y))\frac{\partial \xi_{\gamma'}}{\partial y_i}(y), \quad 1
  \le i \le n\,,$$ 
  and the definition of the $d \times d$ matrix $\Psi$ in (\ref{psi-ij}), we obtain
  \begin{align}
\frac{\partial \varphi}{\partial
    z_{\gamma}}(\xi(y))  = 
    \Big[\frac{\partial
    \big(\varphi\circ\xi\big)}{\partial y_i} a_{ij} 
    \frac{\partial \xi_{\gamma'}}{\partial y_j}
    (\Psi^{-1})_{\gamma\gamma'}\,\Big](y)\,.
  \end{align}
Therefore, integrating by parts, (\ref{q-diff-lemma-eqn1}) simplifies to 
\begin{align*}
  &\int_{\mathbb{R}^d} \varphi(z) \frac{\partial Q}{\partial z_{\gamma}}(z)
    \, dz \\
    =& \int_{\mathbb{R}^n} 
    \varphi(\xi(y)) 
  \frac{\partial}{\partial y_i}\Big(a_{ij} \frac{\partial
  \xi_{\gamma'}}{\partial y_j} (\Psi^{-1})_{\gamma\gamma'} e^{-\beta
  V(y)}\Big)\, dy \, \\
  = &  \int_{\mathbb{R}^d} \varphi(z) 
  \Big[
  \int_{\mathbb{R}^n}
  \frac{\partial}{\partial y_i}\Big(a_{ij} \frac{\partial
  \xi_{\gamma'}}{\partial y_j} (\Psi^{-1})_{\gamma\gamma'} e^{-\beta
  V(y)}\Big)\, \delta(\xi(y) - z)dy \Big] dz\, ,
\end{align*}
from which we can conclude after simplification.
\end{proof}

\section{An alternative proof of Theorem~\ref{thm-fluct-relation}} 
\label{app-3}
In this appendix, we provide an alternative proof of Theorem~\ref{thm-fluct-relation}. 
Different from the proof in Subsection~\ref{subsec-fluct-thm} where only the Feynman-Kac formula has been used, the proof below relies on the combination
of both the Feynman-Kac formula and Girsanov's Theorem. While the idea is inspired by the
derivations in~\cite{Chetrite2008}, the proof below is shorter.
\begin{proof}[Alternative proof of Theorem~\ref{thm-fluct-relation}]
  First of all, we recall the definition of $u$ in (\ref{u-exp-form}) as
  well as the equations 
  (\ref{l-reversed}), (\ref{pde-u-forward}), (\ref{l-r-trans}) used in the
  proof of Theorem~\ref{thm-fluct-relation} in Subsection~\ref{subsec-fluct-thm}. 
  In accordance with (\ref{l-r-trans}), we define 
\begin{align}
  \overline{\mathcal{L}} = \Big(J + a\nabla V +
  \frac{1}{\beta} \nabla\cdot a\Big) \cdot \nabla + \frac{1}{\beta} a :
  \nabla^2 + f \cdot \nabla_\lambda + \epsilon\,\alpha\alpha^T:\nabla^2_\lambda\,,
  \label{l-bar}
\end{align}
and consider the function $\omega(x, \lambda,t) = u\big( x,
\lambda, T-t\,;x',\lambda',t'\big)$. From (\ref{pde-u-forward}) and (\ref{l-r-trans}), we know
that $\omega$ satisfies 
\begin{align}
  \begin{split}
    &\frac{\partial \omega}{\partial t} +
    \overline{\mathcal{L}}_{(x,\lambda,t)}
    \omega + \Big[\mbox{div}(J + a\nabla V) + \mbox{div}_\lambda 
      \Big(f -\epsilon \nabla_\lambda\cdot (\alpha\alpha^T)\Big) + \eta
    \Big]\omega = 0\,, \quad \forall t \in [0, T-t')\,,\\
      & \omega(x,\lambda,t) = \delta(x'-x)\delta(\lambda'-\lambda)\,,\quad t = T-t'\,,
  \end{split}
  \label{eqn-omega}
\end{align}
where $(x,\lambda) \in \mathbb{R}^n \times \mathbb{R}^m$ and
$\overline{\mathcal{L}}_{(x,\lambda,t)}$ is the operator (\ref{l-bar}) evaluated at
$(x,\lambda,t)$. 
On the other hand, applying the Feynman-Kac formula to (\ref{eqn-omega}), we observe
that 
\begin{align}
  \begin{split}
   \omega(x, \lambda,t) =&
    \overline{\mathbf{E}}_{x,\lambda,t}\bigg[\exp\bigg(\int_{t}^{T-t'}
      \Big(\mbox{div}\big(J + a\nabla
      V\big) + \mbox{div}_\lambda \big(f -\epsilon \nabla_\lambda\cdot (\alpha\alpha^T)\big) + \eta\Big)\big(\bar{x}(s), \bar{\lambda}(s),s\big)
      ds\bigg)\\
    & \times
  \delta\big(\bar{x}(T-t')-x'\big)\delta\big(\bar{\lambda}(T-t')-\lambda'\big)\bigg]\,,
\end{split}
\label{omega-exp-exp}
\end{align}
where $\overline{\mathbf{E}}_{x,\lambda,t}$ denotes the conditional expectation under
the path ensemble of the dynamics 
\begin{align}
  d \bar{x}(s)  =& \Big(J+a\nabla V + \frac{1}{\beta}\nabla \cdot a\Big)\big(\bar{x}(s),
  \bar{\lambda}(s)\big)\, ds + \sqrt{2\beta^{-1}} \sigma\big(\bar{x}(s),
  \bar{\lambda}(s)\big)\,dw^{(1)}(s)
  \label{aux-y}
\end{align}
and the control protocol
\begin{align}
  d\bar{\lambda}(s) =& f(\bar{x}(s), \bar{\lambda}(s), s)\, ds +
  \sqrt{2\epsilon}\, \alpha\big(\bar{x}(s), \bar{\lambda}(s),s\big) dw^{(2)}(s)\,, \label{aux-lambda}
\end{align}
starting from $\bar{x}(t) = x$ and $\bar{\lambda}(t) = \lambda$ at
time $t$.
Note that the infinitesimal generator of the dynamics (\ref{aux-y}) and
(\ref{aux-lambda}) is given by the operator $\overline{\mathcal{L}}$ in (\ref{l-bar}).

Now we apply Girsanov's theorem to change the probability measure in (\ref{omega-exp-exp}) from the path ensemble
of the dynamics (\ref{aux-y}), (\ref{aux-lambda}) to the path ensemble of the dynamics
(\ref{dynamics-1-q-vector}), (\ref{lambda-dynamics-full}). Specifically,
starting from $(x, \lambda)$ at time $t$,
let $\mathbf{P}_{x, \lambda}$ and $\overline{\mathbf{P}}_{x,\lambda}$ denote the path
measures on the time interval $[t,T-t']$ corresponding to (\ref{dynamics-1-q-vector}),
(\ref{lambda-dynamics-full}) and (\ref{aux-y}), (\ref{aux-lambda}), respectively.
Applying Girsanov's theorem, we obtain after some straightforward calculations 
\begin{align}
  \begin{split}
    \frac{d\mathbf{P}_{x,\lambda}}{d \overline{\mathbf{P}}_{x,\lambda}}
    \big(x(\cdot), \lambda(\cdot)\big)
  = & \exp\bigg[-\beta \int_{t}^{T-t'} \nabla V\big(x(s), \lambda(s)\big) \cdot dx(s) \\
    & + \beta \int_{t}^{T-t'}
\Big(\nabla V \cdot \big(J + \frac{1}{\beta}\nabla \cdot
a\big)\Big)\big(x(s),\lambda(s)\big)\, ds \bigg]\,.
\end{split}
\label{density-ratio}
\end{align}
Therefore, changing the probability measure in (\ref{omega-exp-exp}) from
$\overline{\mathbf{P}}_{x,\lambda}$ to $\mathbf{P}_{x,\lambda}$, using (\ref{density-ratio}), (\ref{div-j-zero}), we find
\begin{align*}
  & u( x, \lambda,T-t) =  \omega( x, \lambda,t) \\
  =& \mathbf{E}_{x,\lambda,t}\bigg[\exp\bigg(\int_{t}^{T-t'}
  \Big(\mbox{div}\big(J + a\nabla
  V\big) + \mbox{div}_\lambda \big(f -\epsilon \nabla_\lambda\cdot (\alpha\alpha^T)\big)
  + \eta \Big)\big(x(s), \lambda(s),s\big) ds\bigg) \\
  & \times \delta\big(x(T-t')-x'\big) \delta\big(\lambda(T-t')-\lambda'\big)
  \frac{d\overline{\mathbf{P}}_{x,\lambda}}{d\mathbf{P}_{x,\lambda}}
  \big(x(\cdot), \lambda(\cdot)\big) \bigg] \\
  =& \mathbf{E}_{x,\lambda,t}\bigg[\exp\bigg(
  \beta \int_{t}^{T-t'} \nabla V(x(s), \lambda(s)) \cdot dx(s) +  
  \int_{t}^{T-t'} \big(a : \nabla^2 V\big)\big(x(s), \lambda(s)\big) ds\\
& 
  + \int_{t}^{T-t'} \Big(\mbox{div}_\lambda \big(f -\epsilon
  \nabla_\lambda\cdot (\alpha\alpha^T)\big) + \eta\Big) \big(x(s), \lambda(s),s\big) ds\bigg)
\delta\big(x(T-t')-x'\big)\,\delta\big(\lambda(T-t')-\lambda'\big) \bigg]\\
  =& \mathbf{E}_{x,\lambda,t}\bigg[\exp\bigg(
  \beta \int_{t}^{T-t'} \nabla V\big(x(s), \lambda(s)\big) \circ dx(s)
  + \int_{t}^{T-t'} \Big(\mbox{div}_\lambda \big(f -\epsilon \nabla_\lambda\cdot
  (\alpha\alpha^T)\big) + 
  \eta\Big) \big(x(s), \lambda(s),s\big)
  ds\bigg)\\
&\times \delta\big(x(T-t')-x'\big)\,\delta\big(\lambda(T-t')-\lambda'\big)
\bigg]\,.
\end{align*}
Note that in the last equality above, 
we have converted Ito integration to Stratonovich integration according to (\ref{dynamics-1-q-vector}). 
Substituting $t$ by $T-t$, integrating by parts, and recalling the expression (\ref{u-exp-form}), we obtain 
\begin{align*}
  \begin{split}
  &e^{-\beta V(x',\lambda')}\,
    \mathbf{E}^R_{x',\lambda',t'}\bigg[
    \exp\bigg(\int_{t'}^t \eta(x^R(s), \lambda^R(s), T-s) ds\bigg) 
    \delta\big(x^R(t)-x\big)\delta\big(\lambda^{R}(t)-\lambda\big)\bigg]\\
    =&e^{-\beta V(x,\lambda)}\,\mathbf{E}_{x,\lambda,T-t}\bigg[e^{-\beta
  \mathcal{W}}
  \exp\bigg(\int_{T-t}^{T-t'} \eta(x(s), \lambda(s), s) ds\bigg) 
    \delta\big(x(T-t')-x'\big)\delta\big(\lambda(T-t')-\lambda'\big) \bigg]\,,
  \end{split}
\end{align*}
where $\mathcal{W}$ is defined in (\ref{w-div-f}).
\end{proof}
\section{Proof of Theorem~\ref{thm-fluct-relation-coordinate}} 
\label{app-4}
In this appendix, we provide the proof of
Theorem~\ref{thm-fluct-relation-coordinate} in Subsection~\ref{subsec-fluct-thm-coordinate}. 
\begin{proof}[Proof of Theorem~\ref{thm-fluct-relation-coordinate}]
  We consider the quantities on both sides of the equality
  (\ref{fluct-relation-coordinate}).
  For the left hand side of (\ref{fluct-relation-coordinate}), 
let us fix $(y',t') \in  \mathbb{R}^n \times [0,T]$ and define the function $u$ by 
\begin{align}
u\big(y,t\,;y',t'\big) = 
\mathbf{E}_{y',t'}^R\bigg[
\exp\bigg(\int_{t'}^t \eta\big(y^R(s), T-s\big) ds\bigg) 
\delta\big(y^R(t)-y\big)\bigg]\,,
\label{u-exp-form-coordinate}
\end{align}
for $(y,t) \in  \mathbb{R}^n \times [t',T]$.
It is known that $u$ satisfies the PDE 
\begin{align}
  \begin{split}
  &\frac{\partial u}{\partial t} = \big(\mathcal{L}^R\big)^* u
  + \eta(y,T-t) \,u \,, \quad \forall~(y, t) \in  \mathbb{R}^n\times (t',T] \,,\\
  & u(y, t\,;y',t')=\delta(y-y')\,, \quad \mbox{if}~~t=t'\,,
  \end{split}
  \label{pde-u-forward-coordinate}
\end{align}
where the operator $\mathcal{L}^R$ is defined in
  (\ref{l-reversed-coordinate}) and 
$\big(\mathcal{L}^R\big)^*$ denotes its formal $L^2$ adjoint. A direct calculation shows that 
\begin{align}
  \begin{split}
    \big(\mathcal{L}^R\big)^*\phi = & \bigg[\frac{\partial}{\partial y_i}\Big((Pa)_{ij} \frac{\partial V}{\partial y_j}\Big)+ 
    \frac{\partial}{\partial y_i} \Big((\Psi^{-1})_{\gamma\gamma'} (a\nabla
    \xi_\gamma)_i f^{-}_{\gamma'}\Big) \bigg]\phi \\
    & + \bigg[
    (Pa)_{ij} \frac{\partial V}{\partial y_j} + \frac{1}{\beta}
    \frac{\partial (Pa)_{ij}}{\partial y_j} + (\Psi^{-1})_{\gamma\gamma'}
    (a\nabla \xi_\gamma)_i f^{-}_{\gamma'} \bigg] \frac{\partial \phi}{\partial
    y_i} + \frac{1}{\beta} (Pa)_{ij} \frac{\partial^2 \phi}{\partial y_i \partial y_j} \,,
  \end{split}
  \label{l-r-trans-coordinate}
\end{align}
for a smooth function $\phi$.  

For the right hand side of (\ref{fluct-relation-coordinate}), fixing 
  $(y', t') \in \mathbb{R}^n \times [0, T]$, we define the function $g$ for
  $(y,t) \in \mathbb{R}^n \times [t',T]$ as
\begin{align*}
  g(y,t) = 
  \mathbf{E}_{y,T-t}\bigg[&e^{-\beta \mathcal{W}}
\exp\bigg(\int_{T-t}^{T-t'} \eta\big(y(s), s\big) ds\bigg) \delta\big(y(T-t')-y'\big) \bigg]\,,
\end{align*}
where $\mathcal{W}$ is defined in (\ref{w-coordinate}), and the dynamics
  $y(\cdot)$ satisfies the SDE (\ref{dynamics-f}). 
Using the same argument as in Lemma~\ref{lemma-g}, we can verify that 
$g$ satisfies the PDE
\begin{align}
  \begin{split}
    &\frac{\partial g}{\partial t} = \overline{\mathcal{L}}\, g +
    \eta(\cdot, T-t) g \,,\qquad
  \forall\, (y,t) \in  \mathbb{R}^n \times (t',T] \,,\\
  & g(y,t) = \delta(y-y') \,, \qquad \mbox{if}~~t=t'\,,
\end{split}
\label{g-delta-coordinate}
\end{align}
where the operator $\overline{\mathcal{L}}$ is defined as 
\begin{align}
  \begin{split}
  \overline{\mathcal{L}}\,\phi =& 
    \bigg[-\beta (\Psi^{-1})_{\gamma\gamma'} (a\nabla \xi_\gamma)_i f^-_{\gamma'} \frac{\partial
    V}{\partial y_i} + \frac{\partial}{\partial y_i}
    \Big((\Psi^{-1})_{\gamma\gamma'} (a\nabla \xi_\gamma)_i
    f^-_{\gamma'}\Big)\bigg] \phi\\
    & + \mathcal{L}^{\perp} \phi + (\Psi^{-1})_{\gamma\gamma'}
    (a\nabla\xi_\gamma)_if^-_{\gamma'} \frac{\partial \phi}{\partial y_i} 
\end{split}
    \label{bar-l-coordinate}
\end{align}
for a smooth function $\phi$. Now consider the function
$\omega(y,t) = e^{-\beta V(y)} g(y,t)$. A direct calculation shows that 
\begin{align}
  \begin{split}
    e^{-\beta V} \mathcal{L}^{\perp} g = &
    e^{-\beta V} 
      \bigg[-(Pa)_{ij} \frac{\partial V}{\partial y_j}\frac{\partial \big(e^{\beta
      V}\omega\big)}{\partial y_i} +
  \frac{1}{\beta} \frac{\partial (Pa)_{ij}}{\partial y_j}
    \frac{\partial  \big(e^{\beta V}\omega\big)}{\partial y_i} + \frac{1}{\beta} (Pa)_{ij}
    \frac{\partial^2 \big(e^{\beta V}\omega\big)}{\partial y_i\partial
    y_j}\bigg]\,\\
    = & \bigg[\frac{\partial}{\partial y_i}\Big((Pa)_{ij} \frac{\partial
    V}{\partial y_j}\Big)\bigg]\omega  + 
    \bigg[(Pa)_{ij} \frac{\partial V}{\partial y_j} + \frac{1}{\beta}
    \frac{\partial (Pa)_{ij}}{\partial y_j}\bigg] \frac{\partial \omega}{\partial
    y_i} + \frac{1}{\beta} (Pa)_{ij} \frac{\partial^2 \omega}{\partial
    y_i\partial y_j}\,, \\
    e^{-\beta V} \frac{\partial g}{\partial y_i} = &
    e^{-\beta V} \frac{\partial \big(e^{\beta V} \omega\big)}{\partial y_i} = \beta 
    \frac{\partial V}{\partial y_i}\omega + \frac{\partial \omega}{\partial y_i} \,.
  \end{split}
  \label{e-v-l-coordinate}
\end{align}
Combining (\ref{l-reversed-coordinate}), (\ref{g-delta-coordinate}),
(\ref{bar-l-coordinate}),
(\ref{e-v-l-coordinate}), it follows that the function $\omega$ satisfies the PDE
\begin{align*}
  &\frac{\partial \omega}{\partial t} = e^{-\beta V}
  \Big[\overline{\mathcal{L}}\,g + \eta(\cdot, T-t) g\Big]=
\big(\mathcal{L}^R\big)^*
  \,\omega + \eta(y,T-t)\,\omega \,,\quad \forall\,(y, t) \in  \mathbb{R}^n \times (t',T]  \,,\\
  &\omega(y,t) = e^{-\beta V(y')} \delta(y-y')\,,\quad \mbox{if}~~ t=t'\,.
\end{align*}
Comparing this with the equation of function $u$ in (\ref{pde-u-forward-coordinate}),
we obtain $$e^{-\beta V(y')} u(y,t\,;y',t') =
\omega(y,t),$$ which is equivalent to (\ref{fluct-relation-coordinate}). 
\end{proof}
\renewcommand{\bibsection}{\subsection*{\large References}}  
\bibliographystyle{abbrv}
\bibliography{reference}
\end{document}